\documentclass{article}
\usepackage[margin=1.5in]{geometry}
\usepackage{amsmath, amsthm, amssymb}
\usepackage[dvipsnames]{xcolor}
\usepackage{graphicx}
\usepackage{verbatim}
\usepackage{natbib}
\usepackage{caption}
\usepackage{subcaption}
\usepackage{enumerate}
\usepackage{bbm}
\usepackage{multirow}
\usepackage{array}
\usepackage{tikz}
\usetikzlibrary{arrows, chains,shapes.multipart}
\usepackage[hyphens]{url}
\usepackage[hidelinks]{hyperref}
\hypersetup{colorlinks,citecolor=blue,urlcolor=blue, linkcolor=blue,breaklinks=true}
\usepackage{stefan_tex}
\usepackage{cancel}

\usepackage{authblk}
\usepackage{tikzsymbols}

\newcommand{\indp}{\;\rotatebox[origin=c]{90}{$\models$}\;}

\newcommand{\me}{\operatorname{ME}}
\newcommand{\piv}{\vec{\pi}}
\newcommand{\dep}{\operatorname{dep}}
\newcommand{\arr}{\operatorname{arr}}
\newcommand{\nuv}{\vec{\nu}}
\newcommand{\errr}{\operatorname{error}}
\newcommand{\onev}{\vec{1}}
\newcommand{\zerov}{\vec{0}}
\newcommand{\ev}{\vec{e}}
\newcommand{\LE}{\operatorname{WDE}}

\newcommand{\Cancel}[2][black]{{\color{#1}\cancel{\color{black}#2}}}
\newcommand{\model}{\operatorname{WDE}}
\newcommand{\ratio}{\operatorname{ratio}}
\newcommand{\mJ}{\mathcal{J}}

\theoremstyle{plain}
\newtheorem{prop}{Proposition}

\newtheorem{lemm}{Lemma}

\newtheorem{theo}{Theorem}

\theoremstyle{definition}
\newtheorem{exam}{Example}
\newtheorem{defi}{Definition}
\newtheorem{assu}{Assumption}

\theoremstyle{remark}

\newtheorem{rema}{Remark}

\definecolor{ForestGreen}{RGB}{34,139,34}

\newcommand\blfootnote[1]{%
  \begingroup
  \renewcommand\thefootnote{}\footnote{#1}%
  \addtocounter{footnote}{-1}%
  \endgroup
}

\begin{document}

\title{Experimenting Under Stochastic Congestion}
\author[1]{Shuangning Li}
\author[2]{Ramesh Johari}
\author[3]{Xu Kuang}
\author[3]{Stefan Wager}
\affil[1]{University of Chicago Booth School of Business}
\affil[2]{Department of Management Science and Engineering, Stanford University}
\affil[3]{Stanford Graduate School of Business}

\maketitle

\begin{abstract}

We study randomized experiments in a service system when stochastic congestion can arise
from temporarily limited supply or excess demand. Such congestion gives rise to
cross-unit interference between the waiting customers, and analytic strategies that do not account for this
interference may be biased. In current practice, one of the most widely used
ways to address stochastic congestion is to use switchback experiments
that turn a target intervention on and off for the whole system in alternation.
We find, however, that under a queueing model for stochastic congestion, the
standard way of analyzing switchbacks is inefficient and that estimators
that leverage the queueing model can be materially more accurate.
Additionally, we show how the queueing model enables estimation of total policy gradients from unit-level randomized experiments, thus giving
practitioners an alternative experimental approach they can use without needing to
precommit to a fixed switchback length before data collection.
\end{abstract}

\section{Introduction}

Randomized experiments\blfootnote{\hspace{-6mm} Xu Kuang published under a different full name in
earlier versions of this manuscript. Please use ``S. Li, R. Johari, X. Kuang and S. Wager''
when citing this paper.}
provide a simple and robust way to measure the effect of an intervention,
and are widely used in a variety of application areas ranging from medicine and social programs
to industrial settings \citep{imbens2015causal}. More recently, there has been considerable interest
in using randomized experiments (often called A/B tests in this context) to guide data-driven decision making in complex, interconnected service systems such as online marketplaces or ride-hailing platforms.

{Stochastic congestion} that results from dynamic fluctuations in demand and supply is a salient feature in these service systems. For example, in a ride-hailing and delivery platform that seeks to match  driver supply to arriving customers and orders, spikes in demand can lead to temporary unavailability of drivers, resulting in delays experienced by all customers. In a healthcare context, unpredictable fluctuations in demand can cause severe congestion in emergency departments, leading to extended wait times and degraded outcomes for the patients that arrive subsequently \citep{xu2016using}.  Importantly, stochastic congestion can have a significant effect on the  design of experiments, because any intervention that changes the amount of demand or supply would also
likely alter the state of congestion, which in turn affects the experience of other participants. These cross-unit interactions invalidate classical
approaches to analyzing randomized trials \citep{halloran1995causal}.

The goal of this paper is to study experimental design in a service system with dynamic interference arising
from stochastic congestion. The problem of experimental design under cross-unit interference---including in online
marketplaces---has received sustained attention in the literature
\citep{arownow2017,bajari2021multiple, hudgens_2008, johari2022experimental, manski2013identification, munro2021treatment, wager2021experimenting}. A limitation in these existing results, however, is that they mostly focus on static models of interference, while in contrast, the stochastic congestion as discussed above is highly dynamic in nature, rendering existing models and approaches inapplicable.

Practitioners running experiments under stochastic congestion are well aware
of potential bias in classical analyses of randomized trials. In general, however, this challenge is not met by modeling and explicitly accounting
for congestion effects; rather, it is met by recognizing that the effect of any additional congestion caused by one customer dissipates over a relatively short time horizon and therefore can be accounted for if an intervention is held for a sufficiently long period of time. One of the most popular approaches to doing so
is the switchback experiment \citep{bojinov2022design, cochran1941double}, which alternates
(``switches'') the entire population between treatment and control settings over time.
The fact that all users are exposed to only one intervention at any given point in time during the
experiment also helps to correctly account for the interference induced by stochastic congestion
\citep{bojinov2022design,hu2022switchback}. This design has been actively embraced by ride-sharing and delivery platforms \citep[e.g.,][]{chamandy2016experimentation,DoorDash_2018}, serving as an empirical tool
to quantify the effects of feature alterations, algorithm changes or pricing policies on a spectrum of performance metrics, such as the number of
customers served.

\subsection{Overview of Main Results}

We examine whether it is possible to use knowledge of the congestion mechanism
to make better use of data collected from a switchback in the presence of stochastic congestion. The answer to this question turns out to be affirmative.
Drawing inspiration from the queueing systems literature, we consider a single-server queueing model with price- and congestion-sensitive customers to capture the core congestion dynamics arising in many service systems  (Figure \ref{fig:queue_illus}). In this setting, customers arrive to seek service and decide between paying a price $p$ and queueing up, versus leaving for a default outside option. The decision maker wishes to determine how the long-term average admission rate (which further impacts revenue and welfare) depends on the price $p$, and does so by running a two-price switchback experiment. Stochastic congestion plays an important role in this system, because the price $p$ not only affects whether a customer joins the system at the current moment, but also future admissions by altering the congestion level to which subsequent customers will likely be exposed.

Within this framework, we leverage the structural information from the stochastic model to design new estimators for analyzing data from
a switchback. Notably, one of these estimators requires minimal assumptions on the queueing model beyond it being work-conserving. In a Markovian setting, we characterize the asymptotic variance of each estimator and show that, by this metric, one of our proposed estimators is always at least as accurate as (and usually more accurate than) the standard one. Moreover, the accuracy improvement over the status quo can widen significantly in certain models of customer preferences (e.g., whether they prefer longer or shorter queues).

We also find that our queueing model enables the use of simpler, user-level randomization in experimental design. Specifically, in this approach, arriving customers are individually randomized to either a ``high'' or ``low'' price, which are small, symmetric perturbations from a reference price. Our results show that model-inspired estimators perform just as effectively under this unit-level randomization as they do in standard switchback designs. Simulations further suggest that this design can yield lower variance in finite samples. Additionally, in nonstationary environments, we find that combining user-level randomization with our model-inspired estimators is particularly advantageous for addressing nonstationarity. From a practical perspective, this approach eliminates the need to precommit to a fixed switchback length before data collection, making it easier to implement in real-world settings.

Broadly speaking, our results highlight the promise of leveraging structural information about stochastic congestion to improve the analysis of data generated from switchback experiments. As a by-product, we also uncover interesting insights into how customers' preferences can affect estimator accuracy, which may be of independent interest to system operators and experiment designers.

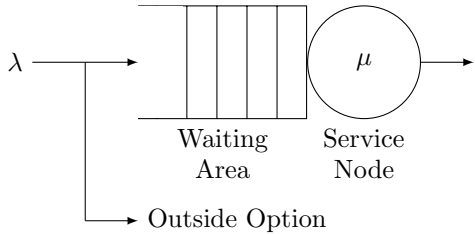
\begin{figure}[t]
\centering
\begin{tikzpicture}[start chain=going right,>=latex,node distance=0pt]
\node[rectangle split, rectangle split parts=6,
draw, rectangle split horizontal,text height=1cm,text depth=0.5cm,on chain,inner ysep=0pt] (wa) {};
\fill[white] ([xshift=-\pgflinewidth,yshift=-\pgflinewidth]wa.north west) rectangle ([xshift=-15pt,yshift=\pgflinewidth]wa.south);

\node[draw,circle,on chain,minimum size=1.5cm] (se) {$\mu$};

\draw[->] (se.east) -- +(20pt,0);
\draw[<-] (wa.west) -- +(-40pt,0) node[left] {$\lambda$};
\node[align=center,below, execute at begin node=\setlength{\baselineskip}{12pt}] at (wa.south) {Waiting \\ Area};
\node[align=center,below,execute at begin node=\setlength{\baselineskip}{12pt}] at (se.south) {Service \\ Node};
\draw ([xshift=-20pt]wa.west) -- +(0pt,-60pt);
\draw[->] ([xshift=-20pt, yshift = -60pt]wa.west) -- +(20pt,0pt) node[right] {Outside Option};
\end{tikzpicture}
\caption{A queueing system with an outside option. Agents are heterogeneous in their
preferences, and consider both queue length and the price $p$ in choosing whether
to join the queue. However, once agents join the queue, they are all processed at
the same rate $\mu$.}
\label{fig:queue_illus}
\end{figure}

\subsection{Literature Review}

\paragraph{Causal Inference under Interference}  One of our main contributions fits, broadly speaking, within the literature of causal inference and experimental design under interference. As discussed above, experimental design under cross-unit interference has been widely studied in various contexts. \citet{arownow2017}, \citet{hudgens_2008}, and \citet{manski2013identification}
show how to represent interference as an extension to the
standard potential outcomes model used for treatment effect estimation \citep{imbens2015causal}.
\citet{bajari2021multiple} and \citet{johari2022experimental} develop methods for two-sided marketplaces where we can
separately intervene on, e.g., buyers and sellers.
\citet{munro2021treatment} and \citet{wager2021experimenting} study interference in large markets via the lens of mean-field asymptotics. Other notable recent contributions in this area include \citet{athey2018exact}, \citet{basse2018}, \citet{leung2020}, \citet{li2022}, \citet{savje2021},
\citet{tchetgen2012}, \citet{toulis2013estimation}, and \citet{ugander_graph_cluster}.

This literature, however, predominantly assumes a static model for interference. In the work of \citet{arownow2017}, \citet{hudgens_2008} and \citet{manski2013identification},
all units are intervened on at the same time, and
all causal effects (including interference effects) are realized immediately.
In \citet{munro2021treatment} and \citet{wager2021experimenting}, the use of large-market asymptotics means that all market-level
characteristics (e.g., prices) concentrate, thus rendering the problem effectively static.  Another common assumption in this literature is that interference is constrained by a network, such that all treatment decisions are made at the same time, and the outcome of
one unit depends only on its treatment and the treatments of its neighbors in the network. Here, in contrast, we consider interference arising from a dynamic process. As a notable exception from the above framing of interference in a static setting, \citet{ogburn2014causal} use dynamic structural
equation models to study contagion effects in epidemiology.
However, the model they use---and methods they derive from the model---are
not directly applicable to our setting with stochastic congestion effects.

\paragraph{Queueing and Stochastic Systems} The modeling and control of stochastic congestion has been a sustained focus of  queueing theory since the seminal work of \citet{erlang1909theory}.  Our contribution to this literature lies in introducing an experimental design approach to policy optimization in a queueing system, one that relies on minimal modeling assumptions and does not require the decision maker to infer the model parameters explicitly. We capture congestion dynamics using a queueing model that is well known and extensively studied in the context of optimal pricing and admission control. A similar model with a single-server queue was first proposed by \cite{naor1969regulation}, where the operator charges a flat price for entry and each arrival can decide whether to queue up for service depending on the current level of congestion.
\cite{naor1969regulation} derives properties for the revenue-optimizing static price. The model was later extended by \cite{chen2001state} and \cite{borgs2014optimal} who studied the impact of state-dependent pricing and the scaling of optimal prices, respectively. Similar queueing models have also been studied in the context of admission control \citep{johansen1980control, stidham1985optimal, spencer2014queueing, xu2015necessity, xu2016using,murthy2022admission}.  These results, however, all assume that the operator is interested in finding the optimal policy, while possessing full knowledge of all system parameters; the latter is very difficult to achieve in complex, practical scenarios, such as an online platform.

Our work intersects with a literature on inference and learning in queueing systems, but there are some notable differences.  For example, \citet{clarke1957maximum} focuses on an $M/M/1$ queue and discusses the estimation problem of the arrival rate $\lambda$ and the service rate $\mu$; see \citet{bhat1987statistical} for a survey of this early line of work and \citet{asanjarani2021survey} for a comprehensive review of the literature on parameter and state estimation for queueing systems. This line of work focuses on explicit state or parameter estimation, as opposed to the inference of overall effects of an intervention in our case.
More recently, a literature has emerged on learning in queues in more complex settings. The problem of learning scheduling and assignment rules in queues is investigated from a multi-armed-bandit perspective by \citet{choudhury2021job, krishnasamy2021learning, freund2022efficient, zhong2022learning}. Unlike our approach, however, this literature is largely interested in using a bandit algorithm to identify the best service assignment on a job-by-job basis, in contrast to our focus on using experiments to learn a system-level optimal policy. \citet{kuang2024detecting} study the estimation of changes in service rates in a multi-server system from observational routing data, focusing on observational studies rather than randomized experiments as in our case. \citet{walton2021learning} provides an overview of recent developments in learning in queueing networks.

The queueing models considered in this paper are special cases of general Markov chains, and if we view $p$ as a parameter that specifies a certain control policy, we can treat the model in this paper as a Markov decision process (MDP). An interesting recent paper in this area that is related to our work at a high level is
\citet{farias2022Markovian}. They study experimentation and off-policy evaluation under what they call Markovian interference, a form of interference where treatments may affect state transitions in an MDP. A main difference between this work and ours is that \citet{farias2022Markovian} focus on analyzing the performance of biased estimators, because unbiased estimators would in general require excessively large sample sizes in unstructured, high-dimensional MDPs. All estimators examined in this paper, on the other hand, are asymptotically unbiased,
and therefore desirable in terms of obtaining the exact treatment effect. We are able to do so because the Markov chain associated with the queueing model is simpler and more structured. As another point of departure, we also examine the implications for estimation and experiment design in non-stationary environments. Finally, in terms of performance guarantees, \citet{farias2022Markovian} provide upper and lower bounds on the estimator's asymptotic variance. In contrast, by exploiting the structures of the queueing model, we are able to derive exact formulae for asymptotic variances of the estimators under stationary versions of our model.

\section{A Queueing Model of Stochastic Congestion}

We model congestion using a single-server continuous-time Markovian queueing model with heterogeneous
customers. We assume that the system is controlled by a continuous decision
variable $p$ (for simplicity we refer to $p$ as the price customers need to pay to receive service), and
that the decision maker is seeking to optimize $p$. Customers arrive according to a Poisson process.
As illustrated in Figure \ref{fig:queue_illus}, as soon as customers arrive they are offered a price
and observe the current queue length, and then choose whether
to join the queue depending on their personal preferences that take into account price and delays.
Once customers join the system they are served in a first-come-first-served manner,
with no preemption. The amount of time needed to serve each customer is independent, and exponentially
distributed with mean $1/\mu$.

Our model of customer heterogeneity implies that, equivalently, we can model customers as arriving
to the system  according to a state-dependent Poisson process with rate $\lambda_k(p)$, where $k$ is
the current queue length. A useful special case, which we will refer to as the $M/M/1$ queue, is when
for all $p$ the arrival rates $\lambda_k(p)$ are uniform across different states, with $\lambda_k(p) = \lambda_0(p)$ for all $k \geq 0$.
To understand $\lambda_k(p)$, we can think of the system as consisting of a queue with an outside option. Let $\lambda_{\operatorname{raw}}$ be the ``raw'' arrival rate to the system. If, conditional on the queue length, arriving customers independently decide to join the queue with probability $\operatorname{prob}(k,\mu,p)$, then we can set $\lambda_k(p) = \lambda_{\operatorname{raw}} \operatorname{prob}(k, \mu, p)$. We make no other assumptions on $\operatorname{prob}(k,\mu,p)$. Below, we present two plausible examples of $\operatorname{prob}(k,\mu,p)$.

\begin{exam}[Waiting cost]
\label{exam:join_rule1}
Consider a store that sells only one product at price $p$. Customers come to the store and wait in line to buy the product. For customer $i$, a plausible decision rule is the following:
\begin{equation}
\label{eqn:join_rule1}
\text{Join if } \quad u_i \geq c_i \EE{\text{waiting time} \mid \text{queue length}} + p.
\end{equation}
Here $u_i$ is the utility if customer $i$ goes through the queue and gets the product and $c_i$ is the cost of waiting per unit time. Since the service rate is $\mu$, we can rewrite \eqref{eqn:join_rule1} as
\begin{equation}
\label{eqn:join_rule2}
\text{Join if } \quad u_i \geq c_i k/\mu + p,
\end{equation}
where $k$ is the current queue length.
Then in this case, we have that $\operatorname{prob}(k, \mu, p)= \PP[(u_i, c_i) \sim F]{u_i \geq c_i k/\mu + p}$, where $F$ is the distribution of $(u_i, c_i)$.
\end{exam}

\begin{exam}[Group conformity]
\label{exam:join_rule2}
Consider a store that sells only one product at price $p$. Customers come to the store and wait in line to buy the product. If customers believe the length of the queue indicates the popularity of the store and thus the quality of the product, they may be more willing to join the queue when the length of the queue is longer. In this case, we will observe that, unlike in Example \ref{exam:join_rule1}, $\operatorname{prob}(k, \mu, p)$ is an increasing function of $k$.
\end{exam}

In this model, any choice of $p$ leads to a certain steady-state queue-length distribution, with an average
rate $V(p)$ of customers processed by the system. The average processing rate accounts for two effects
of the price: the fact that, in each state, demand (i.e., the customer entry rate) is moderated by
the price, and the fact that the steady-state fraction of time spent in different states is
indirectly determined by these entry rates.
$V(p)$ is an important quantity with ties to many key operational measures. For example, it may be of interest
for a social-welfare-maximizing platform that wishes to maximize the number of customers served given
a constrained resource. Further, the rate $V(p)$ is obviously essential to quantify the average steady state revenue, expressed by $pV(p)$.

Given this setup, our focus here is on estimating the policy gradient $V'(p)$. This quantity would be of
direct interest to any policymaker wanting to optimize the system via a gradient-based approach; see, e.g.,
\citet{wager2021experimenting} for formal results. A simple baseline estimator for $V'(p)$ would be to
run a two-treatment switchback experiment where we consider a ``high'' treatment that applies a price slightly
higher than $p$ and a ``low'' treatment that applies a price slightly lower than $p$, and then estimate $V'(p)$
by measuring the difference in average observed processing rates in the ``high'' versus ``low'' states.
We will seek to improve on this baseline by exploiting the structure of our posited queueing model.

\begin{rema}[State-dependent pricing]
The framework presented above is sufficiently general to encompass state-dependent pricing. Specifically, assume that the system sets a price of $b_k$ when the queue length is $k$. We may be interested in understanding, for instance, how the average customer processing rate changes if the price increases by a small percentage (e.g., 1\%) across all states. Let $\lambda_k(p)$ represent the arrival rate in state $k$ when the price is set to $b_k p$. We can then examine the policy gradient $V'(p)$ at $p = 1$ to analyze the impact of such small percentage price adjustments.
\end{rema}

We assume throughout the paper that there is an upper limit on the queue length. In the context of Example \ref{exam:join_rule1}, if the utilities $u_i$ are drawn from a bounded set, then it follows that the queue length will have a finite support.
\begin{assu}[Bounded queue length]
\label{assu:queue_length}
There exists a constant $K \in \mathbb{N}_+$ such that for any price $p$, $K = \min\cb{k:\lambda_k(p) = 0}$.
\end{assu}

\subsection{Three representations of the policy gradient}

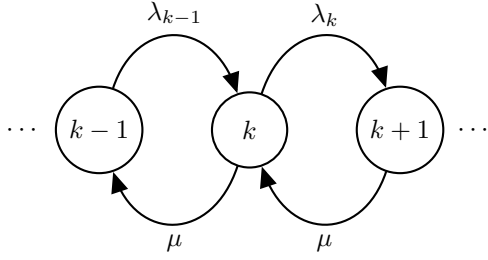
\begin{figure}[t]
\centering
\begin{tikzpicture}[line cap=round,line join=round,>=triangle 45,x=1.0cm,y=1.0cm]
  \node (dot1)  at (-1cm,1cm) {$\cdots$};
  \node (dot2) at (5cm,1cm) {$\cdots$};
 \node (one) [draw, thick, circle, minimum size=1cm] at (0,1cm) {$k-1$};
 \node (two) [draw, thick, circle, minimum size=1cm] at (2,1cm) {$k$};
 \node (three) [draw, thick, circle, minimum size=1cm] at (4,1cm) {$k+1$};

 \draw [->,  thick] (one) .. controls +(0.5,1.5) and +(-0.5,1.5) .. node [midway, above] {$\lambda_{k-1}$} (two);
 \draw [->, thick] (two) .. controls +(0.5,1.5) and +(-0.5,1.5) ..  node [midway, above] {$\lambda_{k}$} (three);

 \draw [->, thick] (three) .. controls +(-0.5,-1.5) and +(0.5,-1.5) ..  node [midway, below] {$\mu$} (two);
 \draw [->, thick] (two) .. controls +(-0.5,-1.5) and +(0.5,-1.5) .. node [midway, below] {$\mu$} (one);
\end{tikzpicture}
\caption{Illustration of the queue length process as a continuous-time Markov chain. }
\end{figure}

In our setting, the queue length process is a continuous-time birth-death Markov chain.  Let $\pi_k(p)$ denote
the steady-state probability of being in state $k$. The average steady-state processing rate
is equal to the average steady-state arrival rate and can be expressed as
\begin{equation}
\label{eq:Vp}
V(p)  = \bar{\lambda}(p) = \sum_{k} \pi_k(p)  \lambda_k(p),
\end{equation}
where $\bar{\lambda}(p)$ is the average steady-state arrival rate, and $\lambda_k(p)$ is the arrival rate with price $p$ when the queue length is $k$.
In order to build accurate estimators for our target estimand $V'(p)$, we start by
giving multiple equivalent representations for it under our model.

We begin by using the chain rule and express $V'(p)$ as
$$\bar{\lambda}'(p)  = \sum_k \lambda_k'(p)\pi_k(p) + \sum_k \pi_k'(p)\lambda_k(p).$$
This reveals a decomposition of the policy gradient into a direct effect and an indirect effect.
Specifically, the term $\sum_k \lambda_k'(p)\pi_k(p)$ represents the direct effect, as it captures how the arrival rate is affected directly by alterations in the price $p$.
On the other hand, the term $\sum_k \pi_k'(p)\lambda_k(p)$ corresponds to the indirect effect, as it reflects the impact of changing the price $p$ on the average arrival rate mediated through alterations in the steady-state queue-length distribution. This arises because changing the price $p$ causes shifts in the steady-state probabilities, resulting in different weightings for the average arrival rate across different queue lengths.

Due to the specific queueing structure present in this context, there is a correspondence between $\lambda_k$ and $\pi_k$. Consequently, we are able to express the estimand $V'(p)$ in several distinct forms, each of which can serve as a target during the estimation stage.

\begin{theo}
\label{theo:estimands}
Under Assumption \ref{assu:queue_length}, our target estimand can be expressed in the following three ways:
\begin{equation}
\label{eqn:lambda_p_dev}
V'(p) = \begin{cases}
\bar{\lambda}'(p) & \text{(model-free expression),} \\
-\mu \pi_0'(p) & \text{(idle-time expression),} \\
\mu \pi_{0}(p) \sum_{k = 0}^{K-1}  \frac{\lambda'_{k}(p)}{\lambda_{k}(p)} \sum_{i = k+1}^{K} \pi_i(p)
& \text{(weighted direct effect).}
\end{cases}
\end{equation}
\end{theo}

The first expression above is immediate from \eqref{eq:Vp}, and holds without any modeling
assumptions. The second expression follows from a general principle called customer conservation \citep{krakowski1973conservation}, whereby the frequency of entries into the system equals
the frequency of departures out of this system, i.e., $\bar{\lambda}(p) = (1 - \pi_0(p))\mu$.
The proof of this second expression relies on the fact that customer conservation
holds under our model; however, we note that the principle (and thus the result) in fact holds
much more broadly, including in G/G/1 queues.

The third expression is the most intricate one and depends on our queue satisfying the
detailed balance condition $\lambda_k(p) \pi_k(p) = \mu \pi_{k+1}(p)$ for all $k$. What is
interesting about this expression is that it only depends on the direct effects of a
price change (here, $\lambda_k'(p)$ can be interpreted as the direct effect of increasing $p$
when the queue has length $k$), and on quantities $\lambda_k(p)$ and $\pi_k(p)$ which can be
estimated by observing the queue in steady state (i.e., without price perturbations). Notably,
all of the terms in this expression can be estimated via simple plug-ins without actually
changing the steady state of the system.  By contrast, a plug-in estimate for $\pi_k'(p)$ requires
observation of system convergence to its stationary distribution at different price levels.

\subsection{Plug-in estimation}

We now consider using the three representations from Theorem \ref{theo:estimands}
to build estimators for the policy gradient using data collected from a switchback
experiment. Specifically, we will consider switchbacks that alternate prices between
$p+ \zeta$ and $p-\zeta$ for some small price perturbation $\zeta > 0$. When the expressions
from Theorem \ref{theo:estimands} have terms that involve derivatives
(e.g., $\lambda'_k(p)$), we estimate them via finite differences; whereas terms without
derivatives (e.g., $\lambda_k(p)$) are estimated via simple averaging. Given that we
focus on experiments that only have small price perturbations (and so finite difference
estimators will be noisy), the terms involving derivatives will generally be the dominant
sources of variance.

We consider two specific ways of running switchbacks: The interval switchback (IS) and the regenerative switchback (RS). In an interval switchback experiment, we divide the time horizon into intervals of equal length $l$. At each time $t = zl$ for some $z \in \mathbb{Z}_+$, the price is set to be $p+\zeta$ or $p-\zeta$ with probability $1/2$  independently, and the price is kept the same in the time interval $[zl, (z+1)l)$ \citep{bojinov2022design, hu2022switchback}. In a regenerative switchback experiment, the price can only be changed when the queue length is $k_r$. At each visit of state $k_r$, the price is set to be $p+\zeta$ or $p-\zeta$ with probability $1/2$  independently \citep{johari2019simons, glynn2020adaptive}.

\begin{figure}[t]
\includegraphics[width = \textwidth]{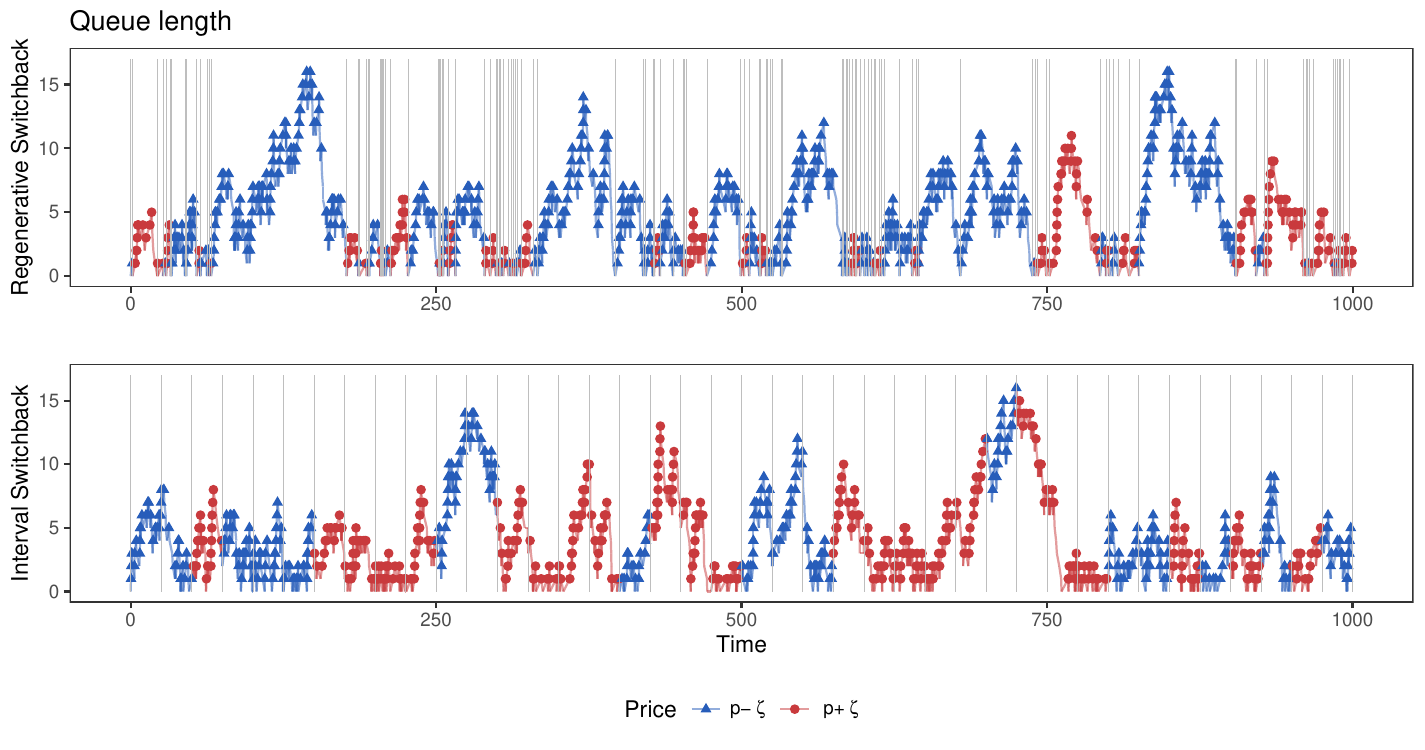}
\caption{Illustration of the two types of switchback experiments. Every dot indicates a new arrival. Blue dots/curve correspond to a price of $p- \zeta$, and red dots/curve correspond to a price of $p+\zeta$. In regenerative switchback experiments, the price is only switched when the queue is empty. In interval switchback experiments, the price is only switched at some fixed time (when $t = 25z$, for $z \in \mathbb{Z}$ in this example).}
\label{fig:illus}
\end{figure}

Targeting different representations of the estimand enables us to construct different estimators using data from a switchback experiment. Let $T$ be the time horizon of the experiments. Let $T_{+}$ ($T_{-}$) be the amount of time the price is at $p+\zeta$ ($p-\zeta$). Let $N_{+}$ ($N_{-}$) be the number of arrivals to the system when the price is at $p+\zeta$ ($p-\zeta$). Targeting $\bar{\lambda}'(p)$, we get a plug-in model-free estimator
\begin{equation}
\label{eqn:first_estimator}
\hat{\tau}_{\bar{\lambda}} = \frac{1}{2\zeta}\p{\frac{N_{+}}{T_{+}} - \frac{N_{-}}{T_{-}}}.
\end{equation}
Meanwhile, for the idle-time-based estimator, writing
$T_{0, +}$ ($T_{0, -}$) for the amount of time the queue length is zero and the price is at $p+\zeta$ ($p-\zeta$), let
\begin{equation}
	\label{eqn:second_estimator}
\hat{\tau}_{\pi_0} = -\frac{\mu}{2\zeta}\p{\frac{T_{0, +}}{T_{+}} - \frac{T_{0, -}}{T_{-}}}.
\end{equation}
Finally, for the weighted direct effect estimator, let $T_{k}$  be the amount of time the queue length is $k$. Let $T_{k, +}$ ($T_{k, -}$) be the amount of time the price is at $p+\zeta$ ($p-\zeta$) and the queue length is $k$. Let $N_{k, +}$ ($N_{k, -}$) be the number of arrivals to the system when the price is at $p+\zeta$ ($p-\zeta$) and the queue length is $k$. We construct \smash{$\hat{\Delta}_{ k}$} as a finite-difference estimator for $\lambda_k'(p)/\lambda_k(p)$:
\begin{equation}
	\label{eqn:Delta}
   \hat{\Delta}_{k} = \frac{1}{\zeta} \frac{N_{k, +}/T_{k, +} - N_{ k, -}/T_{k, -}}{N_{k, +}/T_{k, +} + N_{k, -}/T_{k, -}}.
\end{equation}
Finally, plugging this into our expression from Theorem \ref{theo:estimands} lets
us estimate the policy gradient as a weighted average of the \smash{$\hat{\Delta}_{ k}$}:
\begin{equation}
	\label{eqn:WDE_def}
\hat{\tau}_{\model} = \mu \frac{T_{0}}{T} \sum_{k = 0}^{K-1}  \hat{\Delta}_{k}  \sum_{i = k+1}^{K} \frac{T_{i}}{T}.
\end{equation}
We summarize the notation for the estimators in Table~\ref{table:estimator_more}.

\begin{table}[t]
	\begin{tabular}{ |m{3.1cm}|m{1.65cm}|m{2.4cm}|m{4.95cm}|  }
		\hline
		& Model-free & Idle-time-based & Weighted direct effect \\
		\hline
		Target expression & $\bar{\lambda}'(p)$ &$-\mu \pi_0'(p)$& $\mu \pi_{0}(p) \sum_{k = 0}^{K-1}  \frac{\lambda'_{k}(p)}{\lambda_{k}(p)} \sum_{i = k+1}^{K} \pi_i(p)$\\
		\hline
		Estimator& $\hat{\tau}_{\bar{\lambda}}$  &  $\hat{\tau}_{\pi_0}$& $\hat{\tau}_{\model}$\\
		\hline
	\end{tabular}
	\caption{Summary of the estimators considered. } 
	\label{table:estimator_more}
\end{table}

\section{Asymptotic Behavior under Stationarity}
In this section, we study the asymptotic behaviors of the estimators under stationarity, with a growing time horizon $T \to \infty$.
All estimators depend on the price perturbation $\zeta$ and the time horizon $T$, and to make this explicit we write them as $\hat{\tau}(T,\zeta)$.
We focus on the regime where the perturbation size $\zeta$ shrinks as the time horizon $T$ grows; we write $\zeta_T$ to emphasize this dependence. We make the following assumption on the scaling of $\zeta_T$.

\begin{assu}[Price perturbation]
\label{assu:zeta}
The price perturbation $\zeta_T$ satisfies $\sqrt{T} \zeta_T \to \infty$ and $\sqrt{T} \zeta_T^2 \to 0$ as $T \to \infty$.
\end{assu}

We further make an assumption on the scale and smoothness of the functions $\lambda_k(p)$.
\begin{assu}[Boundedness and smoothness]
\label{assu:bound_smooth}
There exist positive constants $B_0$, $B_1$ and $B_2$ such that $B_0 \leq \lambda_{k}(p) \leq B_1$, and $\abs{\lambda''_{k}(p)} \leq B_2$ for any price $p$ and $k \in \cb{0, 1, \dots, K-1}$.
\end{assu}

Finally, we introduce Assumption~\ref{assu:two_intervals}.
This assumption requires that in the interval switchback experiment, the interval length be sufficiently long so that any carryover effects from the previous period are weak enough not to bias the estimator, while also remaining asymptotically negligible relative to $T$ so that the total number of intervals grows with~$T$.
The latter condition ensures that the total durations under treatment and control are approximately balanced.
In Bernoulli switchback experiments, if the number of switches is fixed, there is a non-negligible chance that the design becomes unbalanced, and allowing the number of intervals to increase mitigates this issue.

\begin{assu}[Interval length]
	\label{assu:two_intervals}
	In the interval switchback experiment, the interval length $l_T$ satisfies $l_T/T \to 0$ and $l_T / \sqrt{T} \to \infty$ as $T \to \infty$.
\end{assu}

Under the above assumptions, we establish central limit theorems for all proposed estimators. We provide proof outlines in Section \ref{sec:proof_outlines}. Additional proofs of the propositions are in Appendix \ref{section:addition_proofs}.
\begin{theo}[Model-free estimator]
\label{theo:first_esti}
Under Assumptions \ref{assu:queue_length}-\ref{assu:two_intervals}, assume further that the data are generated from either the interval switchback experiment or the regenerative switchback experiment. Then,
\begin{equation}
\sqrt{T \zeta_T^2}\p{ \hat{\tau}_{\bar{\lambda}}(T, \zeta_T) - V'(p)} \Rightarrow \mathcal{N}\p{0, \sigma_{\bar{\lambda}}^2(p)},
\end{equation}
as $T \to \infty$, where
\begin{equation}
\label{eqn:sigma_lambdabar_def}
\sigma_{\bar{\lambda}}^2(p) =  (1-\pi_0(p))\mu + 2\mu \pi_0(p)\p{\sum_{k = 1}^{K-1}\frac{S_k(p)S_{k+1}(p)}{\pi_k(p)}  - (1-\pi_0(p))\sum_{k =1}^K \frac{S_k^2(p)}{\pi_k(p)}},
\end{equation}
and $S_k(p) = \sum_{j = k}^{K} \pi_j(p)$.
\end{theo}

\begin{theo}[Idle-time-based estimator]
\label{theo:second_esti}
Under Assumptions \ref{assu:queue_length}-\ref{assu:two_intervals}, assume further that the data are generated from either the interval switchback experiment or the regenerative switchback experiment. Then,
\begin{equation}
\sqrt{T \zeta_T^2}\p{ \hat{\tau}_{\pi_0}(T, \zeta_T) -  V'(p)} \Rightarrow \mathcal{N}\p{0, \sigma_{\pi_0}^2(p)},
\end{equation}
as $T \to \infty$, where
\begin{equation}
\label{eqn:def_sigma_pi0}
\sigma_{\pi_0}^2(p) =  2\mu \pi_0^2(p) \sum_{k = 1}^{K} \frac{S_k^2(p)}{\pi_k(p)}, \textnormal{ and } S_k(p) = \sum_{j = k}^{K} \pi_j(p).
\end{equation}
\end{theo}

\begin{theo}[Weighted direct effect estimator]
\label{theo:model_CLT}
Under Assumptions \ref{assu:queue_length}-\ref{assu:two_intervals}, assume further that the data are generated from either the interval switchback experiment or the regenerative switchback experiment. Then,
\begin{equation}
\sqrt{T \zeta_T^2} \p{ \hat{\tau}_{\model}(T, \zeta_T) -   V'(p)}
\Rightarrow \mathcal{N}\p{0, \sigma_{\model}^2(p)},
\end{equation}
as $T \to \infty$, where
\begin{equation}
\label{eqn:def_sigma_model}
\sigma_{\model}^2(p) =  \mu \pi_0^2(p) \sum_{k = 1}^{K} \frac{S_k^2(p)}{\pi_k(p)}, \textnormal{ and } S_k(p) = \sum_{j = k}^{K} \pi_j(p).
\end{equation}
\end{theo}

	\begin{rema}[Choice of $\zeta_T$]
		To provide some intuition for the choice of $\zeta_T$ in Assumption \ref{assu:zeta}:
		we need $\zeta$ to be small enough so that the estimator bias is small, and large enough so that the estimator variance is small.
		We use the model-free estimator $\hat{\tau}_{\bar{\lambda}} = \frac{1}{2\zeta}({N_{+}}/{T_{+}} - {N_{-}}/{T_{-}})$ as an example to illustrate the scaling of $\zeta$.
		For the bias,
		\[ \abs{\operatorname{Bias}(\hat{\tau}_{\bar{\lambda}})} = \abs{\EE{\hat{\tau}_{\bar{\lambda}}} - V'(p)} = \abs{\frac{V(p + \zeta) - V(p - \zeta)}{2\zeta} - V'(p)} \leq \frac{C_V}{2} \zeta. \]
			Here $C_V$ is a finite bound on $|V''|$ over the price range of interest. Such a bound follows from Assumption~\ref{assu:bound_smooth}, the finite capacity $K$, and the steady-state formula for $V$.
		For the variance, a Markov chain CLT implies
		\[
		\Var{\hat{\tau}_{\bar{\lambda}}} = \mathcal{O}\!\left(\frac{1}{T \zeta^2}\right).
		\]
		We need the variance to vanish, which requires $T \zeta^2 \to \infty$, i.e., $\zeta \gg T^{-1/2}$.
		To make the estimator asymptotically unbiased, the squared bias must decay faster than the variance, i.e., $\zeta^2 \ll 1/(T \zeta^2)$, which gives $\zeta \ll T^{-1/4}$.
		Together these yield $T^{-1/2} \ll \zeta \ll T^{-1/4}$.
	\end{rema}

	\begin{rema}[Interval length]
		The requirement that each interval length grow faster than $\sqrt{T}$ in Assumption \ref{assu:two_intervals} arises from controlling carryover bias under a continuous-time Markov chain with mixing time of constant order. Within each period of length $l_T$, a constant amount of time at the start of the period can be influenced by the previous period (see \citet{hu2022switchback}). Thus, over total horizon $T$, the total contaminated time is $C \cdot (T / l_T)$ for some constant $C$. For a central limit theorem on the $\sqrt{T}$ scale, this bias must be negligible relative to $\sqrt{T}$, i.e.
		\[
		C \frac{T}{l_T} \ll \sqrt{T}
		\quad \Longrightarrow \quad
		l_T \gg \sqrt{T}.
		\]
		This ensures that the cumulative effect of carryover bias vanishes asymptotically, allowing valid inference.
	\end{rema}

Comparing the three asymptotic variances, we find that the asymptotic variance of the weighted direct effect estimator is always half that of the idle-time-based estimator. Furthermore, the following theorem establishes that the asymptotic variance of the weighted direct effect estimator is no larger than that of the model-free estimator.

\begin{theo}[Variance comparison]
\label{theo:variance_comp}
Assume that there is a constant $K \in \mathbb{N}_+ \cup \{\infty\}$ such that for any price $p$, $K = \min\cb{k:\lambda_k(p) = 0}$.\footnote{We take the convention that $\min(\emptyset) = \infty$.} Then,
\begin{equation}
\label{eqn:var_comp1}
\sigma_{\model}^2(p) = \frac{1}{2} \sigma_{\pi_0}^2(p),
\end{equation}
and
\begin{equation}
\label{eqn:var_comp2}
\sigma_{\model}^2(p) \leq \sigma_{\bar{\lambda}}^2(p),
\end{equation}
where $\sigma_{\bar{\lambda}}^2(p), \sigma_{\pi_0}^2(p)$ and $\sigma_{\model}^2(p)$ are defined in \eqref{eqn:sigma_lambdabar_def}, \eqref{eqn:def_sigma_pi0} and \eqref{eqn:def_sigma_model} respectively.
In \eqref{eqn:var_comp2}, the equality holds if and only if the queue is an $M/M/1$ queue when the price is set at $p$, i.e., $\lambda_k(p) =  \lambda_0(p)$ for all $k \geq 0$.
\end{theo}

\begin{rema}[$K = \infty$]
    Theorems \ref{theo:first_esti}-\ref{theo:model_CLT} are established under Assumption \ref{assu:queue_length}, which requires the queue length to be bounded above by a positive integer $K$. Nonetheless, the expressions for the variances (Equations \eqref{eqn:sigma_lambdabar_def}, \eqref{eqn:def_sigma_pi0} and \eqref{eqn:def_sigma_model}) are still well-defined when $K = \infty$. In Theorem \ref{theo:variance_comp}, we establish that equality in \eqref{eqn:var_comp2} holds if and only if the queue is an $M/M/1$ queue; however, for an $M/M/1$ queue, the queue length is not bounded and so the central limit theorem results proven above cannot be applied directly (although we do still conjecture that they hold).
\end{rema}

\subsection{Variance estimators and confidence intervals}
\label{subsection:variance_estimator}

Theorems \ref{theo:first_esti}-\ref{theo:model_CLT} provide us with central limit theorems that come with explicit formulae for the variance. These variance formulae can be utilized to obtain estimators for the variance. Moreover, since the estimators derived from these theorems exhibit an asymptotic normal distribution, we can take advantage of this property to construct confidence intervals.

In particular, if we define $T_k$ as the total time the queue length is equal to $k$, then the ratio $T_k/T$ serves as a natural estimator for the probability $\pi_k$. We denote this estimator as $\hat{\pi}_k = T_k/T$. We note that the asymptotic variances of the estimators are functions of the true probabilities $\pi_k$. To obtain estimators for the variances, we can employ a plug-in approach by substituting the estimated probabilities $\hat{\pi}_k$ into the formulae involving $\pi_k$:
\begin{equation}
    \hat{\sigma}_{\bar{\lambda}}^2 =  (1-\hat{\pi}_0)\mu + 2\mu \hat{\pi}_0\p{\sum_{k = 1}^{K-1}\frac{\hat{S}_k\hat{S}_{k+1}}{\hat{\pi}_k}  - (1-\hat{\pi}_0)\sum_{k =1}^K \frac{\hat{S}_k^2}{\hat{\pi}_k}},
\end{equation}
\begin{equation}
    \hat{\sigma}_{\pi_0}^2 =  2\mu \hat{\pi}_0^2 \sum_{k = 1}^{K} \frac{\hat{S}_k^2}{\hat{\pi}_k}, \qquad \textnormal{and} \qquad
    \hat{\sigma}_{\model}^2 =  \mu \hat{\pi}_0^2 \sum_{k = 1}^{K} \frac{\hat{S}_k^2}{\hat{\pi}_k},
\end{equation}
where $\hat{S}_k = \sum_{j = k}^{K} \hat{\pi}_j$.

\begin{theo}[Variance estimators]
    \label{theo:variance_estimator}
Under Assumptions \ref{assu:queue_length}-\ref{assu:two_intervals}, the variance estimators are consistent:
\begin{equation}
    \hat{\sigma}_{\bar{\lambda}}^2 \stackrel{p}{\to}
    \sigma_{\bar{\lambda}}^2, \qquad
    \hat{\sigma}_{\pi_0}^2 \stackrel{p}{\to}
    \sigma_{\pi_0}^2, \qquad
    \hat{\sigma}_{\model}^2 \stackrel{p}{\to}
    \sigma_{\model}^2,
\end{equation}
as $T \to \infty$.
\end{theo}

Because the variance estimators are consistent, by Slutsky's theorem, we have that
\begin{equation}
\frac{\sqrt{T \zeta_T^2}}{\hat{\sigma}_{\bar{\lambda}}}\p{ \hat{\tau}_{\bar{\lambda}}(T, \zeta_T) - V'(p)} \Rightarrow \mathcal{N}\p{0, 1},
\end{equation}
as $T \to \infty$. Similar results hold for the other estimators.
Then, constructing confidence intervals becomes straightforward.
For a desired confidence level of $(1-\alpha)$, we can construct a confidence interval for $V'(p)$ based on the estimator $\hat{\tau}_{\bar{\lambda}}(T, \zeta_T)$ as follows:
\begin{equation}
    \hat{\tau}_{\bar{\lambda}}(T, \zeta_T) \pm Z_{\alpha/2}\frac{\hat{\sigma}_{\bar{\lambda}}(p)}{\sqrt{T \zeta_T^2}},
\end{equation}
where $Z_{\alpha/2}$ represents the $\alpha/2$-upper quantile of the standard Gaussian distribution. Again, the same procedure can be applied to the other estimators.

An immediate practical implication of the variance estimators is that they can be used to guide the choice of estimator in applications. Because all three estimators are consistent for the same target $V'(p)$ but may have different finite-sample variances, a practitioner can compute the plug-in variance estimates $\hat{\sigma}_{\bar{\lambda}}^2$, $\hat{\sigma}_{\pi_0}^2$, and $\hat{\sigma}_{\model}^2$ from the observed data and select the estimator with the smallest estimated variance.

In Appendix~\ref{appendix:variance_estimator_simulation}, we present simulation evidence showing that the proposed variance estimators are generally quite accurate.

\subsection{Some illustrations of the asymptotic variance}
\label{section:illus_asymp_var}
To illustrate the theoretical results, we compute the asymptotic variances of the estimators in a few concrete examples. In Figures \ref{fig:MM1}-\ref{fig:conformity}, we plot the steady-state probabilities and the asymptotic variances (scaled by $T\zeta_T^2$) of the estimators.

\paragraph {$M/M/1$ queue}
We start with a simple $M/M/1$ queue. In Figure~\ref{fig:MM1}, we set $\mu = 1$ and $\lambda_k = \lambda$ for $k \geq 0$, where we take $\lambda = 0.5$ in the left panel and let $\lambda$ be a varying parameter in the right panel. As shown in Theorem \ref{theo:variance_comp}, the asymptotic variances of the weighted direct effect estimator and the model-free estimator are the same. We observe that the two curves of asymptotic variances ($\sigma_{\model}^2(p)$ and $\sigma_{\bar{\lambda}}^2(p)$) overlap perfectly, while $\sigma_{\pi_0}^2(p)$ is twice as large.

\paragraph{Zero-deflated or zero-inflated $M/M/1$ queue} We then consider a slightly modified $M/M/1$ queue in Figure~\ref{fig:MM1_zero_inflate}. We take $\lambda_0$ to be different from the other $\lambda_k$. We set $\mu = 1$ and $\lambda_k = 0.5$ for $k \geq 1$. In the left panel, we take $\lambda_0 = 1$ and in the right panel, we let $\lambda_0$ be a varying parameter. We observe from the right panel that $\sigma_{\model}^2(p)$ is strictly smaller than $\sigma_{\bar{\lambda}}^2(p)$ except when $\lambda_0 = 0.5$. Note that 0.5 is the value of other $\lambda_k$, and thus when $\lambda_0 = 0.5$, this becomes an $M/M/1$ queue. Comparing $\sigma_{\bar{\lambda}}^2(p)$ and $\sigma_{\pi_0}^2(p)$, we find that $\sigma_{\pi_0}^2(p)$ becomes smaller than $\sigma_{\bar{\lambda}}^2(p)$ when $\lambda_0$ is large. Under this zero-deflated $M/M/1$ queueing model, when $\lambda_0$ is large, $\pi_0$ is small.

\paragraph{Joining probability: power law} We consider a more realistic setting: customers are more willing to join the queue when the queue is shorter (so that the expected waiting time is lower). In Figure~\ref{fig:power_law}, we take $\mu = 1$ and $\lambda_k = 2(k+1)^{-\alpha}$ for $k \in \cb{ 0, \dots, 14}$. In the left panel, we take $\alpha = 0.4$ and in the right panel, we let $\alpha$ be a varying parameter. We observe from the right panel that $\sigma^2_{\bar{\lambda}}(p)$ is strictly larger than $\sigma^2_{\pi_0}(p)$ and $\sigma^2_{\model}(p)$. In the left panel, we observe that the queue length is concentrated around 3. This is because when the queue is short, the arrival rate is higher and the queue quickly becomes longer, whereas when the queue is long, the arrival rate is lower and the queue quickly becomes shorter.

\paragraph{Group conformity: preference for longer queues} We consider the setting introduced in Example~\ref{exam:join_rule2}: customers are more willing to join longer queues because they believe that queue length is an indicator of the quality of the product. In Figure~\ref{fig:conformity}, we take $\mu = 1$ and $\lambda_k = \lambda(0.5 + (k-7)/200)$ for $k \in \cb{0, \dots, 14}$. In the left panel, we take $\lambda = 2$ and in the right panel, we let $\lambda$ be a varying parameter. In the right panel, we see that $\sigma^2_{\model}(p)$ is always lower than the other two and that $\sigma^2_{\pi_0}(p)$ is smaller than $\sigma^2_{\bar{\lambda}}(p)$ when $\lambda$ is large. In the left panel, we observe an interesting phenomenon: with high probability, the queue is either very long or very short. This is because when the queue is long, more customers are attracted, and the queue typically stays relatively large; conversely, for the same reason, when the queue is short, it typically stays relatively short. Therefore, we observe the queue oscillating between very long and very short states. We recognize that this example is somewhat unconventional within the queueing literature.

\mbox{}\\
Broadly, it is worth noting that our focus is not on directly modeling customer behavior. Instead, we treat customer behavior as ``black-boxed'' within the state-dependent arrival rate. The tools we developed here are intended to accommodate a variety of customer utility patterns.

\begin{figure}
\includegraphics[width = \textwidth]{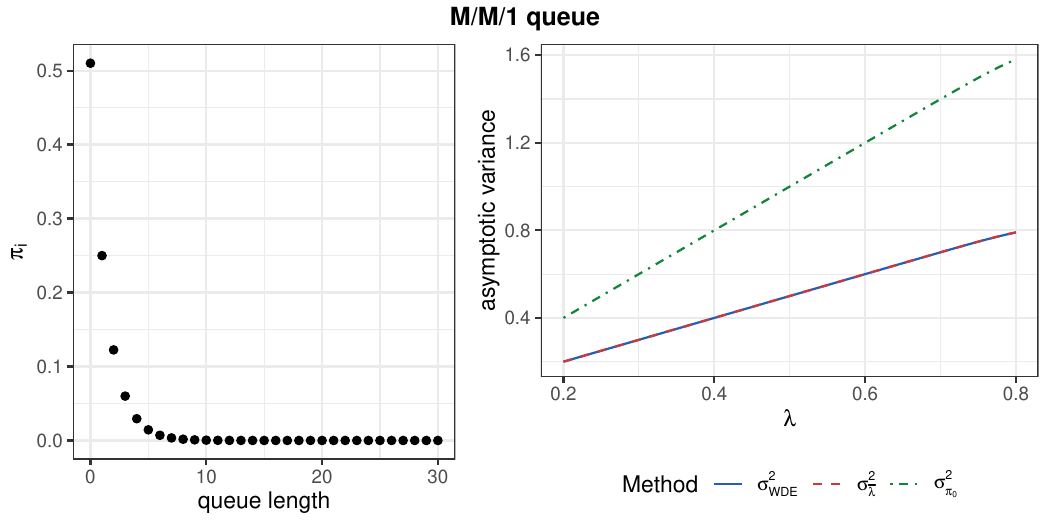}
\caption{In an $M/M/1$ queue, $\sigma_{\model}^2(p)$ and $\sigma_{\bar{\lambda}}^2(p)$ are the same. }
\label{fig:MM1}
\end{figure}

\begin{figure}
\includegraphics[width = \textwidth]{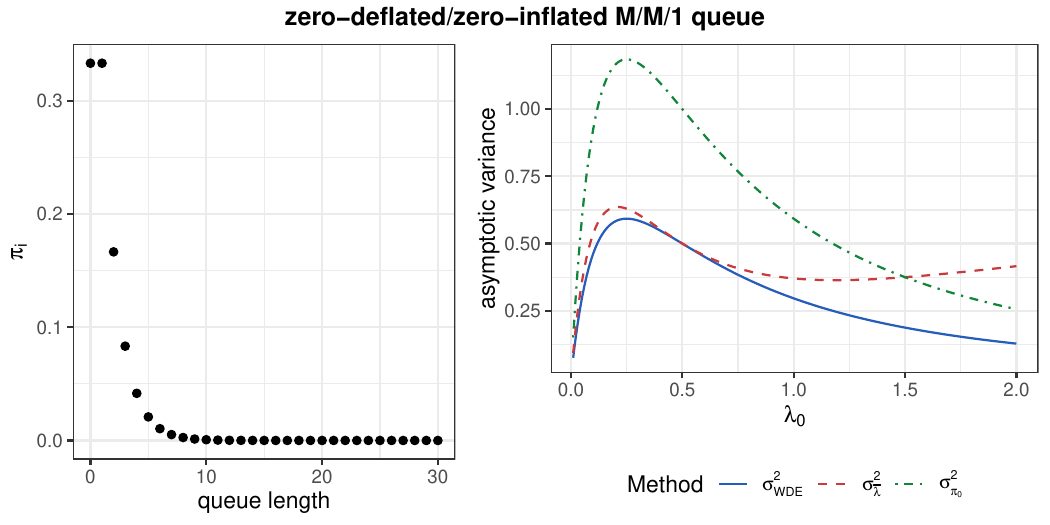}
\caption{In a zero-deflated or zero-inflated $M/M/1$ queue, $\sigma_{\model}^2(p)$ is strictly smaller than $\sigma_{\bar{\lambda}}^2(p)$ except when $\lambda_0 = 0.5$ (corresponding to an $M/M/1$ queue). When $\lambda_0$ is large, $\sigma_{\pi_0}^2(p)$ is smaller than $\sigma_{\bar{\lambda}}^2(p)$.}
\label{fig:MM1_zero_inflate}
\end{figure}

\begin{figure}
\includegraphics[width = \textwidth]{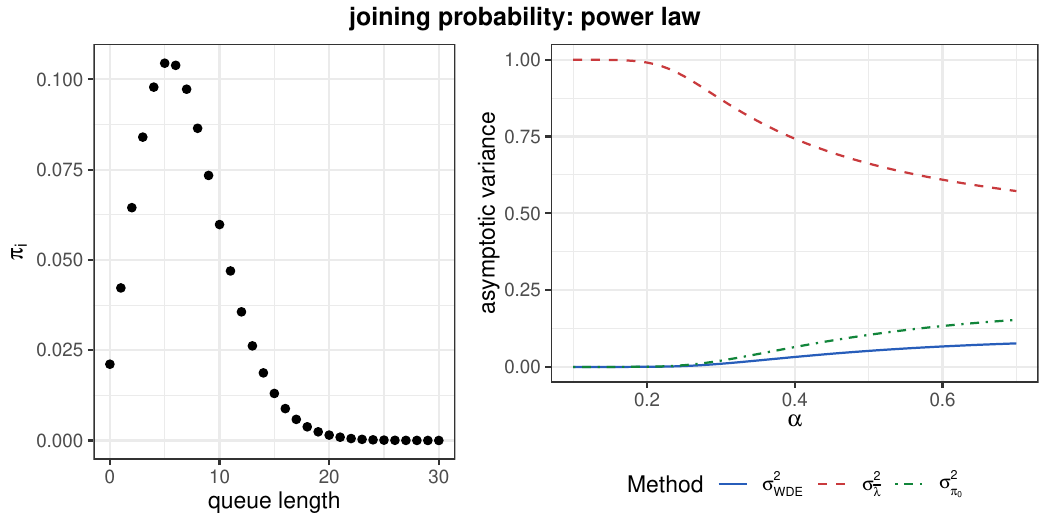}
\caption{A setting where customers are more willing to join the queue when the queue is shorter. In this setting,  $\sigma^2_{\bar{\lambda}}(p)$ is strictly larger than $\sigma^2_{\pi_0}(p)$ and $\sigma^2_{\model}(p)$. The queue length is fairly concentrated.}
\label{fig:power_law}
\end{figure}

\begin{figure}
\includegraphics[width = \textwidth]{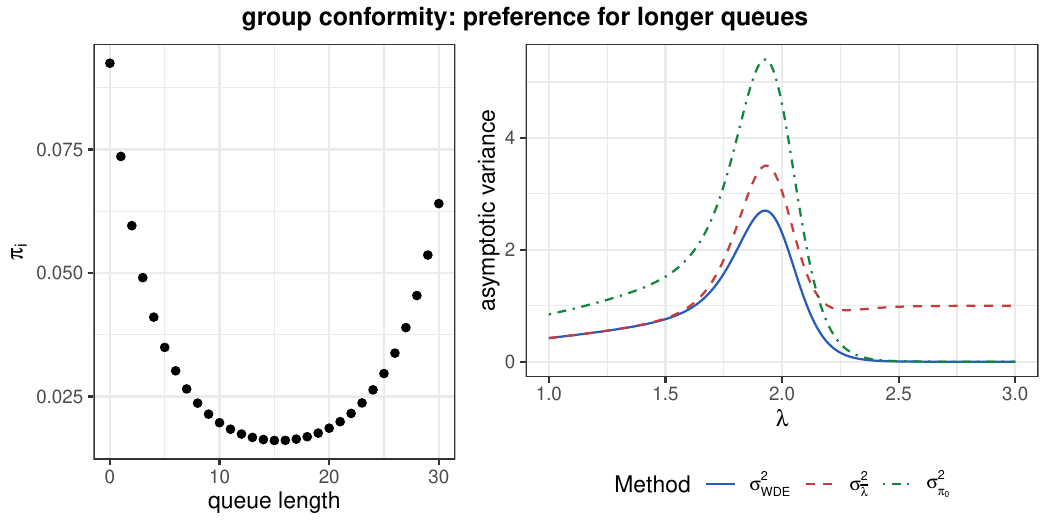}
\caption{A setting where customers are more willing to join longer queues. In this setting, $\sigma^2_{\model}(p)$ is always lower than the other two and $\sigma^2_{\pi_0}(p)$ is smaller than $\sigma^2_{\bar{\lambda}}(p)$ when $\lambda$ is large. The queue oscillates between very long and very short states.
}
\label{fig:conformity}
\end{figure}

\section{Randomization at the User Level}
\label{section:user-level}
We have discussed both interval and regenerative switchback experiments. Interval switchback experiments are commonly used in practice; however,  when one wishes to run an interval switchback experiment, a key step is determining the length of each interval. This can be challenging because the interval length must be specified in advance, before the experiment is conducted.
In general, a very short interval length can lead to a strong carryover effect from one period to the next, potentially causing bias in the analysis if the queue-length process does not have enough time to converge to steady state, whereas longer intervals can lead to excess variance.

Interestingly, however, a close examination of the functional form of the three estimators we consider---the model-free estimator, the idle-time-based estimator, and the weighted direct effect (WDE) estimator---suggests that the WDE estimator may be less sensitive to carryover bias from using short switchback intervals. This is because, intuitively, bias due to carryover effects should primarily arise through plug-in estimation of changes $\pi_k'(p)$ to the stationary distribution, as the stationary distribution cannot be directly measured when the system is between equilibria. While both the model-free and idle-time-based expressions of the policy gradient depend explicitly on $\pi_k'(p)$, the WDE representation does not, so plug-in estimators based on the WDE expression do not involve plug-in estimation of $\pi_k'(p)$. The only plug-in gradient required by the WDE estimator is $\lambda_k'(p)$, i.e., the direct effect, which can be estimated out of equilibrium (because it quantifies user behavior conditional on state). This suggests that the WDE estimator should not suffer from carryover bias even with short switchback intervals.

Motivated by the above discussion, we naturally wonder what happens when we consider the interval switchback experiment and let the interval length approach zero; in particular, a reasonable conjecture is that in this regime, the WDE estimator will achieve both low bias {\em and} low variance. As the interval lengths become infinitesimally small, the experiment essentially transitions into a user-level randomization scheme. In a user-level experiment, rather than applying the same price to all customers arriving within a given time window, we randomize the price for each individual customer. Specifically, for customer $i$, the price is set as:
\begin{equation}
p_{i} = p + \zeta \varepsilon_{i}, \quad \text{where } \varepsilon_{i} \stackrel{\mathrm{i.i.d.}}{\sim} {\pm 1} \text{ uniformly at random}.
\end{equation}
Note that we are switching between two prices that are very close to each other, which renders the user-level experiment equivalent to the local perturbation experiment proposed by \citet{wager2021experimenting}.

In \eqref{eqn:Delta} and \eqref{eqn:WDE_def}, we defined the WDE estimator using quantities including $T_{k,+}$ and $T_{k,-}$. When the interval length shrinks to zero, both $T_{k,+}$ and $T_{k,-}$ reduce to $T_k/2$, which allows the expression for the WDE estimator to simplify. In particular,
\begin{equation}
	\hat{\Delta}_{k} = \frac{1}{\zeta} \frac{N_{k, +} - N_{ k, -}}{N_{k, +} + N_{k, -}},
\end{equation}
and
\begin{equation}
	\hat{\tau}_{\model} = \mu \frac{T_{0}}{T} \sum_{k = 0}^{K-1}  \hat{\Delta}_{k}  \sum_{i = k+1}^{K} \frac{T_{i}}{T}.
\end{equation}
We refer to this estimator as the WDE estimator under the user-level randomized experiment.

Empirically, we find that $\hat{\tau}_{\LE}$ under the user-level experiment demonstrates finite-sample advantages over $\hat{\tau}_{\model}$ under switchback experiments. In particular, we compare the performance of the two estimators in the settings discussed in Section~\ref{section:illus_asymp_var}, using $T = 2000$. Details of the simulation setup are provided in Appendix~\ref{appendix:simulation_details_third}. Figure~\ref{fig:MSE_comp_third} reports the mean squared errors of the two estimators. We find that the estimator based on the user-level experiment outperforms the one based on the usual switchback experiment.
\begin{figure}[t]
	\includegraphics[width = \textwidth]{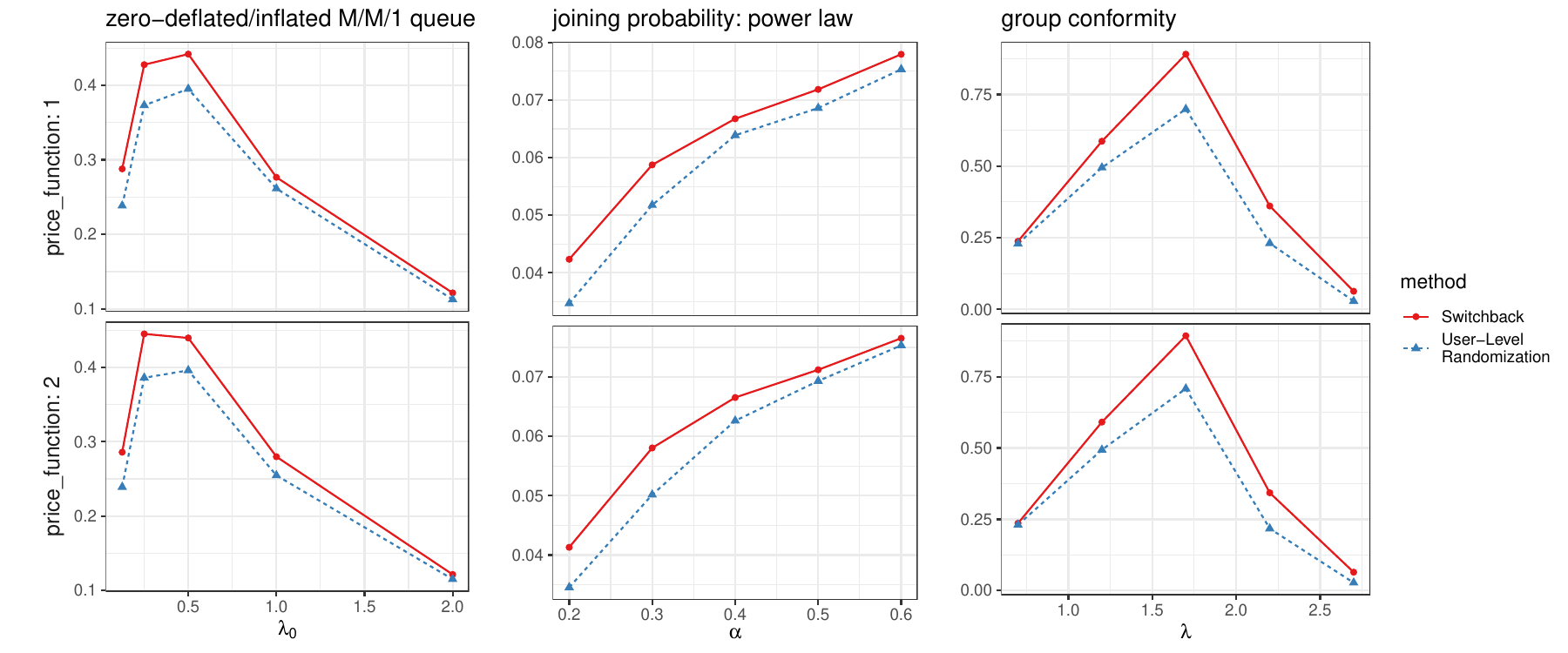}
	\caption{Mean squared errors (scaled by $T\zeta_T^2$) of the weighted direct effect estimators: comparing the interval switchback experiment to the user-level experiment.}
	\label{fig:MSE_comp_third}
\end{figure}

One intuition that may help explain the finite-sample variance advantage of user-level randomization is the following. The switchback experiment can be viewed as a form of simple random sampling, whereas user-level randomization resembles stratified sampling \citep{parsons2014stratified}. In switchback experiments, we independently sample performance under two queues, akin to drawing from the population without considering within-queue heterogeneity. In contrast, user-level randomization effectively ``samples'' customer arrivals within the same queue, creating natural strata based on the queue length $k$. Within each stratum, the two treatment groups are compared and then aggregated across strata.
It is well known that stratified sampling typically yields lower variance than simple random sampling, particularly in finite samples where within-stratum variation can be controlled.

Theoretically, we establish in Theorem~\ref{theo:LE_CLT} below that the weighted direct effect estimator with a user-level experiment enjoys the same asymptotic variance guarantee as that with a usual interval switchback experiment.
\begin{theo}
\label{theo:LE_CLT}
Under Assumptions \ref{assu:queue_length}-\ref{assu:bound_smooth}, using the weighted direct effect estimator with a user-level experiment yields
\begin{equation}
\sqrt{T \zeta_T^2} \p{ \hat{\tau}_{\LE}(T, \zeta_T) -  V'(p)} \Rightarrow \mathcal{N}\p{0, \sigma_{\LE}^2(p)},
\end{equation}
as $T \to \infty$, where
\begin{equation}
\label{eqn:def_sigma_LE}
\sigma_{\LE}^2(p) =  \mu \pi_0^2(p) \sum_{k = 1}^{K} \frac{S_k^2(p)}{\pi_k(p)}, \textnormal{ and } S_k(p) = \sum_{j = k}^{K} \pi_j(p).
\end{equation}
\end{theo}

\section{Non-Stationary Environments}
\label{section:nonstationarity}

The discussion thus far has centered on an idealized context characterized by full stationarity.
However, in practical scenarios, significant non-stationarity can be observed in the arrival rates. For instance, examining the half-hourly arrival rates to an emergency department, one might observe patterns akin to those in Figure \ref{fig:ED_data}.

\begin{figure}
\centering
\includegraphics[width = 0.7\textwidth]{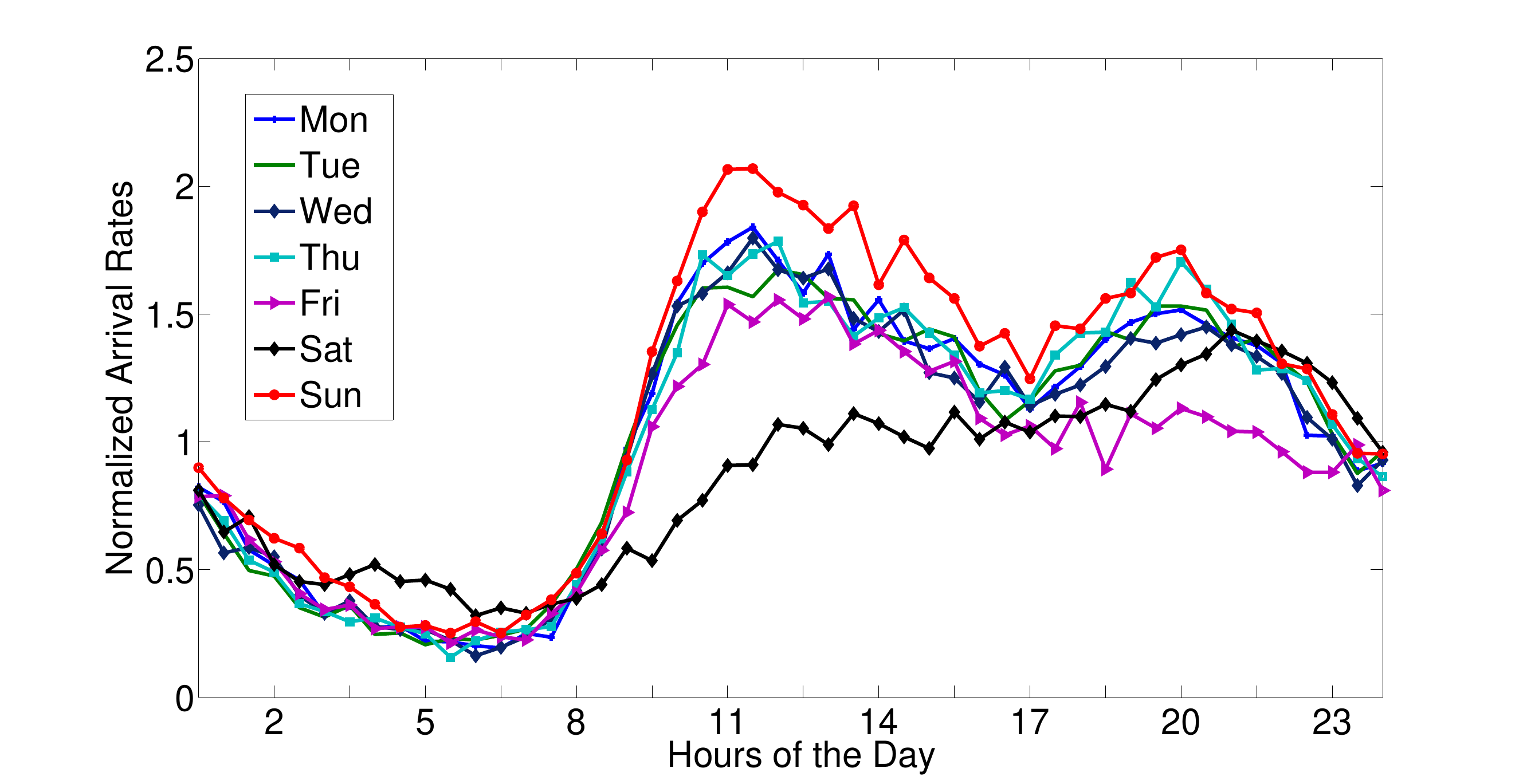}
\caption{Half-hourly arrival rates obtained from emergency department records in the SEEStat database \citep{see2009seestat}. There is significant intraday variation: arrivals are more prevalent around midday and less prevalent in the early morning. Day-to-day variation is also present: arrival rates not only tend to be lower on Saturdays but also tend to peak later in the day rather than around noon.
}
\label{fig:ED_data}
\end{figure}

In the presence of non-stationarity, additional considerations are necessary when running interval switchback experiments. One key challenge lies in selecting an appropriate interval length. On the one hand, as discussed in Section \ref{section:user-level}, if the interval length is too short, a strong carryover effect from the previous period to the new period may occur. On the other hand, if the interval length is too long, non-stationarity can significantly impact the results. For example, in an extreme case, switching once a year might mean that observed differences between two periods are more likely due to changes in the external environment rather than the intended price change.

Given this observation, the user-level experiment combined with the weighted direct effect estimator appears to be more attractive in settings with non-stationarity. As discussed in Section \ref{section:user-level}, the weighted direct effect estimator is less sensitive to very short interval lengths, and very short interval lengths are helpful in dealing with non-stationarity.

There is one additional caution to be taken when applying the weighted direct effect estimator (using a switchback or user-level experiment) in a non-stationary environment. The logic behind the weighted direct effect estimator is rooted in the steady-state balance equation of the underlying Markov chain, which is an intrinsically stationary concept. Therefore, when using it in a non-stationary setting, we find it necessary to construct temporally ``local'' versions of the estimators for $\lambda_k'(p)/\lambda_k(p)$ and $\pi_k(p)$, which can be thought of as capturing the system dynamics within small, and therefore relatively stationary, time intervals. Specifically, we partition the entire time horizon into consecutive, non-overlapping windows of length $s$, indexed by $w \in \mathbb{N}$. Within the $w$-th window, we construct $\hat{\Delta}_{k, w}$ as a local estimator for $\lambda_k'(p)/\lambda_k(p)$:
\begin{equation}
	\hat{\Delta}_{k,w} = \frac{1}{\zeta} \frac{N_{k,w, +}/T_{k,w, +} - N_{k,w, -}/T_{k,w, -}}{N_{k,w, +}/T_{k,w, +} + N_{k,w, -}/T_{k,w, -}}.
\end{equation}
Subsequently, the weighted direct effect estimator can be similarly constructed:
\begin{equation}
	\hat{\tau}_{\model}(s) = \mu \frac{s}{T}\sum_w \left[\frac{T_{0,w}}{s} \sum_{k = 0}^{K-1} \hat{\Delta}_{k,w} \sum_{i = k+1}^{K} \frac{T_{i,w}}{s} \right].
\end{equation}
In these equations, $T_{k,w}$ denotes the amount of time the queue length is $k$ within window $w$; $T_{k,w, +}$ ($T_{k,w, -}$) denotes the amount of time the price is at $p+\zeta$ ($p-\zeta$) and the queue length is $k$ in window $w$; and $N_{k,w, +}$ ($N_{k, w, -}$) denotes the number of arrivals to the system when the price is at $p+\zeta$ ($p-\zeta$) and the queue length is $k$ in window $w$. Similarly, when we run a user-level experiment, we can construct:
\begin{equation}
\hat{\Delta}_{k,w} = \frac{1}{\zeta} \frac{N_{k,w, +} - N_{k,w, -}}{N_{k,w, +} + N_{k,w, -}},
\end{equation}
\begin{equation}
\label{eqn:nonstationary_estimator}
\hat{\tau}_{\LE}(s) = \mu \frac{s}{T}\sum_w \left[\frac{T_{0,w}}{s} \sum_{k = 0}^{K-1} \hat{\Delta}_{k,w} \sum_{i = k+1}^{K} \frac{T_{i,w}}{s} \right].
\end{equation}
To distinguish from the interval length, we refer to the above window length $s$ as the \textit{kernel length}. The kernel length represents the duration of the window used to construct local estimators in our data analysis after the experiments have been conducted.

\subsection{Formal results}

The simulation results from the previous section suggest that the user-level experiment combined with the weighted direct effect estimator is a favorable option in non-stationary environments. In this subsection, we provide further theoretical justification for the use of the weighted direct effect estimator with the user-level experiment. We show that as long as the non-stationarity evolves slowly over time—i.e., the system exhibits relatively stable behavior over short periods—the $\hat{\tau}_{\LE}(s)$ estimator defined in \eqref{eqn:nonstationary_estimator} remains a consistent estimator for the target of interest.

\begin{assu}[Finitely many change points]
\label{assu:change_point}
There are $B$ regimes separated by $B-1$ change points in the arrival rate: there exist time points $0 = t^{(0)} < t^{(1)} < \dots < t^{(B-1)} < t^{(B)} = 1$ and arrival rate functions $\lambda_k^{(1)}(p), \dots, \lambda_k^{(B)}(p)$ such that, for $b \in \{1,\dots,B\}$ and any $t \in (t^{(b-1)} T, t^{(b)} T]$, the state-dependent arrival rate at time $t$ is
\begin{equation}
    \lambda_{k,t}(p) = \lambda_k^{(b)}(p).
\end{equation}
\end{assu}
Under Assumption \ref{assu:change_point}, the average processing rate $V(p)$ becomes
\begin{equation}
    V(p) = \sum_{b = 1}^B(t^{(b)} - t^{(b-1)}) V^{(b)}(p)
         = \sum_{b = 1}^B(t^{(b)} - t^{(b-1)})\p{\sum_{k=0}\pi_k^{(b)}(p)\lambda_k^{(b)}(p)},
\end{equation}
and thus the quantity of interest is the gradient:
\begin{equation}
    V'(p) = \sum_{b = 1}^B(t^{(b)} - t^{(b-1)}) V^{(b)}{}'(p).
\end{equation}
\begin{theo}
\label{theo:non_stationarity}
    Under Assumptions \ref{assu:queue_length}, \ref{assu:zeta}, \ref{assu:bound_smooth} and \ref{assu:change_point}, assume further that $s_T/T \to 0$ and $s_T \zeta_T^2 \to \infty$. Let $\tilde{\tau}_{\LE}(s_T)$ be a truncated version of $\hat{\tau}_{\LE}(s_T)$, as defined in \eqref{eqn:truncation} in Section \ref{subsection:non_station_proof}. Using the weighted direct effect estimator with a user-level experiment yields
    \begin{equation}
        \tilde{\tau}_{\LE}(s_T) \stackrel{p}{\to} V'(p),
    \end{equation}
    as $T \to \infty$.
\end{theo}

In Section~\ref{subsection:non-stationarity_diverging}, we provide additional results showing that the consistency result continues to hold under a relaxed version of Assumption~\ref{assu:change_point}, where the number of change points is allowed to grow.

\subsection{A non-stationary simulator}

In this section, we develop a non-stationary simulator using real-world arrival rate data. We then use it to assess the performance of the different switchback designs and estimators in a non-stationary environment. A main insight from this analysis is that the user-level experiment combined with the weighted direct effect estimator appears to perform best across all designs, even in the face of non-stationarity.

Following \citet{xu2016using}, we construct a set of half-hourly arrival rate traces using emergency department records from the SEEStat database \citep{see2009seestat}. In particular, we focus on the HomeHospital dataset within SEEStat and calculate the average arrival rates to the emergency department based on records from 2004. Figure \ref{fig:ED_data} provides an illustration of these data.

Specifically, we set the experiment's horizon at $T = 4~\text{weeks}$.
We employ a Poisson process with time-varying rates to model the raw arrivals, utilizing the half-hourly arrival rate shown in Figure~\ref{fig:ED_data} for these time-varying rates.
The joining probability is modeled using the proportional balking model \citep{hassin2003queue}, where the state-dependent arrival rate is given by:
\begin{equation}
	\lambda_{t,k}(p) = \frac{C_t(2-p)}{k+1},
\end{equation}
where $C_t$ is a time-dependent quantity that does not depend on the queue length or the price. In the notation of Appendix~\ref{appendix:nonstationary_simulator}, $C_t=4a_wb_{d,t}$.
The primary aim of this simulator is to replicate the level of non-stationarity commonly encountered in practical settings. For detailed information on the simulator, see Appendix \ref{appendix:nonstationary_simulator}.

We test the performance of the three estimators discussed in the paper on the simulator. For the model-free estimator and the idle-time-based estimator, we set the kernel length equal to $T$ and vary the interval length. For the weighted direct effect estimator, based on the intuition developed in Section \ref{section:user-level}, we set the interval length to be 0, i.e., we consider the user-level experiment and vary the kernel length. Figure \ref{fig:rMSE_three_estimators} plots the root mean squared errors of the estimators on the $y$-axis against the interval length/kernel length on the $x$-axis.

\begin{figure}
	\centering
	\includegraphics[width=\textwidth]{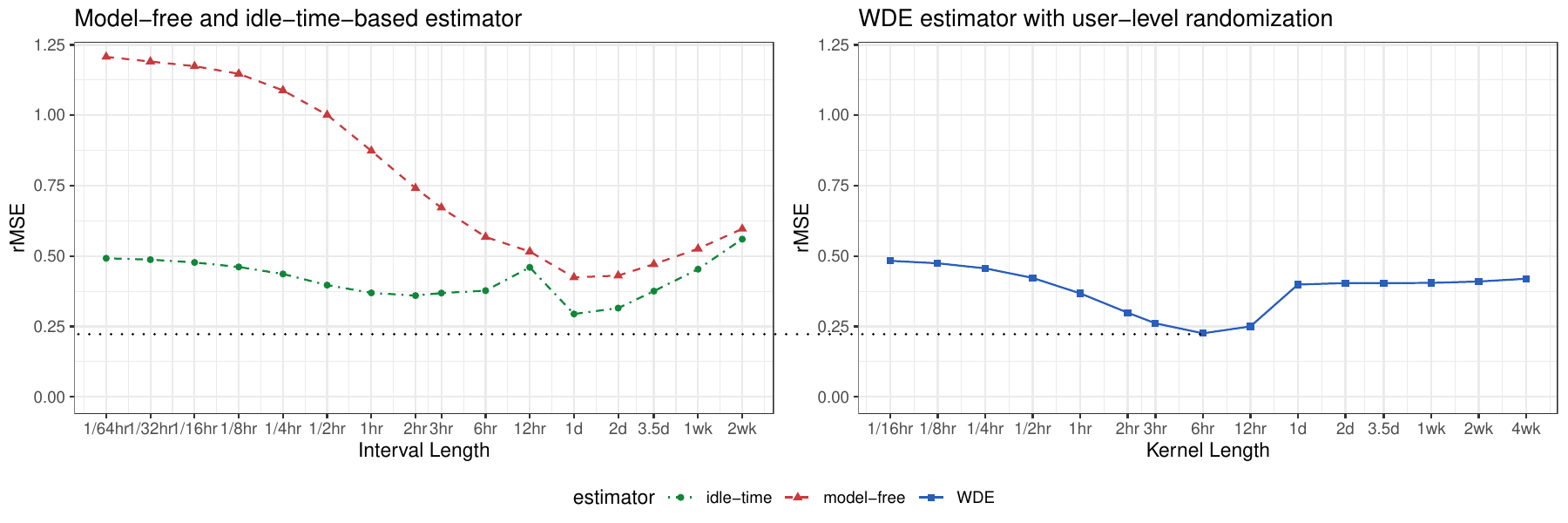}
		\caption{Root Mean Squared Error of the Three Estimators, with the Kernel Length Set to $T$ for the Model-Free and Idle-Time-Based Estimators and the Interval Length Set to Zero for the Weighted Direct Effect Estimator.}
	\label{fig:rMSE_three_estimators}
\end{figure}

From Figure \ref{fig:rMSE_three_estimators}, we observe that the performance ranking of the three estimators aligns with our theoretical predictions: $\hat{\tau}_{\model}$ outperforms both $\hat{\tau}_{\pi_0}$ and $\hat{\tau}_{\bar{\lambda}}$. More specifically, the weighted direct effect estimator combined with user-level randomization and a kernel length of six hours offers the lowest root mean squared error among all combinations of estimator and experimental design.

Looking more closely at Figure \ref{fig:rMSE_three_estimators}, we find that the choice of the correct interval length (for the model-free and idle-time-based estimators) or kernel length (for the weighted direct effect estimator) appears to be crucial, as varying lengths result in markedly different performances of the estimators. For the interval length, in this emergency department case study, the more pronounced non-stationarity within a single day compared to between days might partly explain why the model-free and idle-time-based estimators perform relatively well when the interval length is set to one day. For the kernel length, there appears to be a variance-bias tradeoff: when the kernel length is small, the bias of the local estimators is small, but the variance is large because the effective sample size is smaller. Conversely, when the kernel length is large, the bias increases as we pool data across time points that can be very different, but the variance decreases.

\begin{figure}
	\centering
	\includegraphics[width=0.7\textwidth]{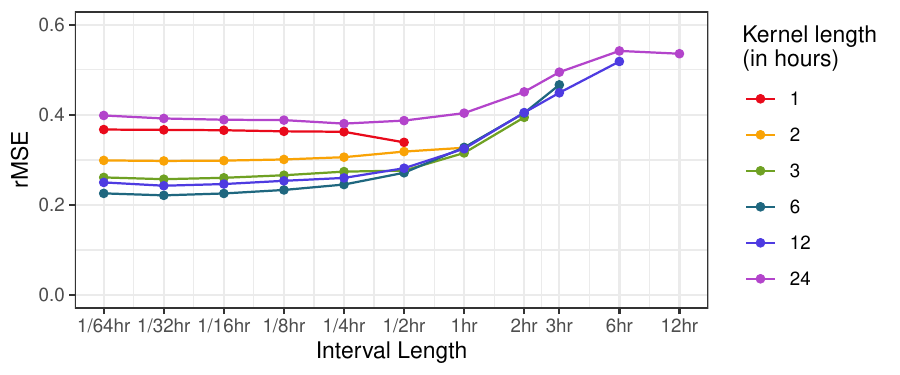}
		\caption{Root mean squared error of the weighted direct effect estimator with different interval lengths and kernel lengths: the error is minimal at interval lengths close to zero for all reasonable choices of kernel lengths.}
	\label{fig:RMSE_WDE_diff_kernel_length}
\end{figure}

Finally, to validate our intuition that the user-level experiment is a good choice of experiment when using the weighted direct effect estimator, we examine the performance of the weighted direct effect estimator with different interval lengths and kernel lengths on the simulator. Figure~\ref{fig:RMSE_WDE_diff_kernel_length} demonstrates that the root mean squared error of the weighted direct effect estimator is minimized at interval lengths close to zero for all reasonable choices of kernel lengths.

\subsection{Choosing hyperparameters: kernel versus~interval length}
We have seen from the earlier sections that the accuracy of both the switchback experiment and the user-level randomization experiment is sensitive to the choice of hyperparameters: the interval length for switchbacks and the kernel length for user-level randomization, respectively. While there does not appear to be any simple and effective rule for determining the optimal value of either hyperparameter, we believe that it is still preferable to fine-tune the kernel length rather than the interval length, making the user-level randomization design more attractive.

First, it may be easier to identify a performant kernel length. This is because the data collection procedure under user-level randomization does not depend on the kernel length, so the hyperparameter can be chosen after the experiment has been completed. This is especially convenient when historical data on the system dynamics are scarce, as we can use data collected during the experiment itself to determine an appropriate kernel length. For instance, in the example presented earlier in this section, it is visually clear from the aggregate observational data (Figure \ref{fig:ED_data}) that the system exhibits intraday variation on a timescale of approximately 12 hours. The decision maker may thus set the kernel length to 12 hours, which, as it turns out, is a nearly optimal choice (Figure \ref{fig:rMSE_three_estimators}). By contrast, the interval length directly affects how the experiment is run, i.e., by specifying the frequency of treatment switches, and must therefore be determined before the experiment. In situations where historical observational data from the system are scarce or even entirely absent before the experiment, choosing the correct interval length can be very challenging.

Another consideration that favors tuning the kernel length over the interval length is the resulting optimal and worst-case accuracy. If we compare the two plots in Figure \ref{fig:rMSE_three_estimators}, we see that the RMSE of WDE under the optimal kernel length is smaller than that of both estimators under the switchback design, with the gap being more prominent relative to the model-free estimator. Therefore, even if we assume that one could optimally fine-tune the hyperparameters, using WDE with user-level randomization still delivers more bang for the buck. Conversely, user-level experiments with WDE seem to be more robust to poor choices of hyperparameter values than switchback experiments. Again in Figure \ref{fig:rMSE_three_estimators}, we see that the RMSE of WDE under the worst choice of kernel length is lower than the worst RMSEs of the switchback experiments, with the model-free estimator again performing much worse.

\section{Implications for Practice}

Our paper has significant implications for practitioners, which we summarize in this section.
In our discussion, we draw a distinction between insights concerning statistical analysis of data generated by an existing (switchback) experiment, versus those for designing new experimental procedures as well as the ensuing data analysis.

For the problem of statistical analysis, our main insight is that one could potentially achieve significantly better power and accuracy if one is able to leverage a problem's underlying stochastic congestion model, as demonstrated by our results on the weighted direct effect estimator. A potential pushback to this viewpoint is that the real system may not operate in accordance with some of the assumptions of the stochastic model or that such a model may not be readily available to the practitioner. To that end, our results suggest that even using some partial structure, such as the idle-time-based estimator, could provide a boost in power compared to naive model-free estimators while requiring much milder structural assumptions. More broadly, our results should encourage practitioners to invest in structural understanding of their systems, as such modeling can provide significant benefits in experimental analysis.

Shifting our attention to the problem of experimental design, we observe that the design of a switchback experiment is often motivated by the need to mitigate environmental nonstationarity. An important but challenging managerial decision in launching such an experiment is the choice of a switching interval, where the decision maker must balance the carryover effect (when the intervals are too short) against the impact of nonstationarity (when the intervals are too long). In this context, our results suggest an alternative approach: by using user-level randomization in the experimental design and the WDE estimator in the subsequent statistical analysis, the manager can not only avoid having to commit to a predetermined switching interval but also potentially achieve higher accuracy. These results further hint at a novel experimental paradigm, whereby the use of stochastic modeling allows the manager to shift the focus from directly estimating total treatment effects to estimating only direct effects. By doing so, she can then shift some of the burden of hyperparameter fine-tuning for running the experiment, which can be difficult and irreversible, to fine-tuning parameters in post-experiment statistical analysis, which can be easier and more flexible to perform.

\section{Proof outlines of Theorems \ref{theo:estimands}-\ref{theo:non_stationarity}}
\label{sec:proof_outlines}

\subsection{Proof of Theorem \ref{theo:estimands}}
The balance equation gives that for any $k \leq K$,
\begin{equation}
\pi_{k -1}(p) \lambda_{k-1}(p) = \pi_k(p) \mu.
\end{equation}
Making use of this equation, we can express $ V'(p)$ as
\begin{equation}
\begin{split}
V'(p)
&= \frac{d}{d p} \p{\sum_{k = 0}^{K-1} \pi_k(p) \lambda_k(p)}
= \frac{d}{d p} \p{\sum_{k = 0}^{K-1} \pi_{k+1}(p) \mu }
= \mu \frac{d}{d p} \p{\sum_{k = 0}^{K-1} \pi_{k+1}(p)}\\
&= \mu \frac{d}{d p}   \p{1 - \pi_{0}(p)}
= - \mu \pi_0'(p).
\end{split}
\end{equation}

Furthermore, from the balance equation, we see that $\pi_k$ can be expressed as a cumulative product from $\pi_0(p)$:
\begin{equation}
\pi_k(p) = \pi_0(p) \cdot \prod_{i=0}^{k-1} (\lambda_i(p)/\mu).
\end{equation}
To facilitate differentiation with respect to $p$, we can further convert the above product into a sum by taking logarithms on both sides of the equation:
\begin{equation}
\log \pi_k(p) =  \log \pi_0(p)  +  \sum_{i=0}^{k-1} (\log \lambda_i (p) - \log \mu).
\end{equation}
Taking the derivatives with respect to $p$ on both sides thus leads to the conclusion:
\begin{equation}
\label{eqn:perc_change}
\frac{\pi_{k}'(p)}{\pi_k(p)} =   \frac{\pi_{0}'(p)}{\pi_{0}(p)}  +  \sum_{i = 0}^{k-1} \frac{\lambda'_{i}(p)}{\lambda_{i}(p)}.
\end{equation}
Note that the above holds for any $1 \leq k \leq K$. For $k \geq K+1$, it is clear that $\pi_k(p) = 0$ for any $p$ and thus $\pi_k'(p) = 0$.

Equation \eqref{eqn:perc_change} further implies that
\begin{equation}
\pi_{k}'(p) =  \pi_k(p) \frac{\pi_{0}'(p)}{\pi_{0}(p)}  +  \pi_k(p) \sum_{i = 0}^{k-1} \frac{\lambda'_{i}(p)}{\lambda_{i}(p)}.
\end{equation}
Summing over $k$ gives
\begin{equation}
\begin{split}
-\pi_0'(p) &= \sum_{k = 1}^{K} \pi_{k}'(p) = \frac{\pi_{0}'(p)}{\pi_{0}(p)} \sum_{k = 1}^{K} \pi_{k}(p)  + \sum_{k = 1}^{K} \pi_{k}(p)  \sum_{i = 0}^{k-1} \frac{\lambda'_{i}(p)}{\lambda_{i}(p)}\\
& = \frac{\pi_{0}'(p)}{\pi_{0}(p)}(1 - \pi_{0}(p)) + \sum_{k = 0}^{K-1}  \frac{\lambda'_{k}(p)}{\lambda_{k}(p)} \sum_{i = k+1}^{K} \pi_i(p).
\end{split}
\end{equation}
Thus
\begin{equation}
- \frac{\pi_{0}'(p)}{\pi_{0}(p)} = \sum_{k = 0}^{K-1}  \frac{\lambda'_{k}(p)}{\lambda_{k}(p)} \sum_{i = k+1}^{K} \pi_i(p).
\end{equation}
Hence we can write $\bar{\lambda}'(p)$ as
\begin{equation}
\bar{\lambda}'(p) = -\mu \pi_0'(p) = \mu \pi_{0}(p) \sum_{k = 0}^{K-1}  \frac{\lambda'_{k}(p)}{\lambda_{k}(p)} \sum_{i = k+1}^{K} \pi_i(p).
\end{equation}

\subsection{Proof of Theorem \ref{theo:first_esti}}
We prove the theorem in two steps. In the first step, we establish a central limit theorem without specifying the asymptotic mean and variance. In the second step, we give explicit formulae for the mean and variance.

By definition, we can write the estimator as
\begin{equation}
\hat{\tau}_{\bar{\lambda}} = \frac{1}{2\zeta_T}\p{\frac{N_{+}(T, \zeta_T)}{T_{ +}(T, \zeta_T)} - \frac{N_{-}(T, \zeta_T)}{T_{ -}(T, \zeta_T)}}.
\end{equation}
In the first step, we establish the following proposition.
\begin{prop}
\label{prop:first_esti_triangular}
There exist $\tilde{\sigma}(p): \mathbb{R} \to \mathbb{R} $ and  $\tilde{\mu}(p) : \mathbb{R} \to \mathbb{R} $ such that under Assumptions \ref{assu:queue_length}-\ref{assu:two_intervals}, for data generated from either the interval switchback experiment or the regenerative switchback experiment,
\begin{equation}
	\label{eqn:CLT_unknown_para}
\sqrt{T}
\begin{pmatrix}
\frac{N_{+}(T, \zeta_T)}{T_{ +}(T, \zeta_T)} - \tilde{\mu}(p+\zeta_T)\\
\frac{N_{-}(T, \zeta_T)}{T_{ -}(T, \zeta_T)} - \tilde{\mu}(p-\zeta_T)\\
\end{pmatrix}
\Rightarrow \mathcal{N}(\zerov,  \tilde{\Sigma}(p)),
\end{equation}
as $T \to \infty$, where
$
\tilde{\Sigma} = \begin{pmatrix}
	\tilde{\sigma}^2(p) & 0\\
	0 & \tilde{\sigma}^2(p)
\end{pmatrix}
$.
\end{prop}
To show this proposition, we use some triangular array arguments to deal with the changing price perturbation $\zeta_T$. We note that the asymptotic variance $\tilde{\Sigma}(p)$ does not depend on the price perturbation, and we do not have an explicit form for it so far.

In the second step, we give explicit formulae for $\tilde{\mu}$ and $\tilde{\Sigma}$. We achieve this by establishing the following two propositions.
\begin{prop}
\label{prop:first_esti_zero}
Set the price at $p$ throughout the experiment. Let $N_{\arr}(T, p)$ be the total number of arrivals in time $[0,T]$. Under Assumptions \ref{assu:queue_length}-\ref{assu:bound_smooth}, as $T \to \infty$,
\begin{equation}
\sqrt{T}
\p{\frac{N_{\arr}(T, p)}{T} - \tilde{\mu}(p)}
\Rightarrow \mathcal{N}(0,  \tilde{\sigma}^2(p)/2).
\end{equation}
The function $\tilde{\mu}$ and the value of $\tilde{\sigma}^2(p)$ are the same as in Proposition \ref{prop:first_esti_triangular}.
\end{prop}
\begin{prop}
\label{prop:coro:num_arrival}
Set the price at $p$ throughout the experiment. Let $N_{\arr}(T, p)$ be the total number of arrivals in time $[0,T]$. Under Assumptions \ref{assu:queue_length}-\ref{assu:bound_smooth}, there is a constant $C(K, p) > 0$ such that
\begin{equation}
	\abs{\EE{N_{\arr}(T, p)}  - T \mu(1-\pi_0(p)) } \leq  C(K, p),
\end{equation}
\begin{equation}
	\abs{\Var{N_{\arr}(T, p)}  - T \sigma_{\bar{\lambda}}^2(p) } \leq  C(K, p),
\end{equation}
where $\sigma_{\bar{\lambda}}^2(p)$ is defined in \eqref{eqn:sigma_lambdabar_def}.
\end{prop}
Note that here in Proposition \ref{prop:coro:num_arrival}, $1-\pi_0(p)$ is the steady-state probability of the system being non-empty, and thus $T(1-\pi_0(p))$ is the expected amount of time the system is non-empty (in steady state).

Proposition \ref{prop:first_esti_zero} gives a central limit theorem for $N_{\arr}(T, p)$, while Proposition \ref{prop:coro:num_arrival} gives the limits of the mean and variance of $N_{\arr}(T, p)$. Together with a uniform integrability result (Lemma \ref{lemma:uniform_intergral}), these results imply that the asymptotic mean and variance in the central limit theorem are the same as the limiting mean and variance in Proposition \ref{prop:coro:num_arrival}, i.e.,
\begin{equation}
\label{eqn:mu_sigma_formula}
\tilde{\mu}(p) = \mu(1-\pi_0(p)), \textnormal{ and } \tilde{\sigma}^2(p) = 2 \sigma_{\bar{\lambda}}^2(p).
\end{equation}
Since the choice of the price $p$ is arbitrary in Propositions \ref{prop:first_esti_zero}-\ref{prop:coro:num_arrival}, we can plug \eqref{eqn:mu_sigma_formula} back into Proposition \ref{prop:first_esti_triangular} and get
\begin{equation}
\sqrt{T}
\begin{pmatrix}
\frac{N_{+}(T, \zeta_T)}{T_{+}(T, \zeta_T)} - \mu(1-\pi_0(p+\zeta_T))\\
\frac{N_{-}(T, \zeta_T)}{T_{-}(T, \zeta_T)} - \mu(1-\pi_0(p-\zeta_T))\\
\end{pmatrix}
\Rightarrow \mathcal{N}(\zerov,  \tilde{\Sigma}(p)),
\end{equation}
where
$
\tilde{\Sigma} = \begin{pmatrix}
	2 \sigma_{\bar{\lambda}}^2(p) & 0\\
	0 & 2 \sigma_{\bar{\lambda}}^2(p)
\end{pmatrix}
$. Therefore,
\begin{equation}
\sqrt{T}\p{\p{\frac{N_{+}(T, \zeta_T)}{T_{+}(T, \zeta_T)} - \frac{N_{-}(T, \zeta_T)}{T_{-}(T, \zeta_T)}} + \mu\p{\pi_0\p{p+\zeta_T} - \pi_0\p{p-\zeta_T}}}
 \Rightarrow\mathcal{N} (0, 4\sigma_{\bar{\lambda}}^2(p)),
\end{equation}
as $T \to \infty$, where $ \sigma_{\bar{\lambda}}^2(p)$ is defined in \eqref{eqn:sigma_lambdabar_def}. Finally, since $|\lambda_k''(p)| \leq B_2$ holds for any $k$ for some constant $B_2$, we have that $\abs{\pi''_0(p)} \leq B_{\pi}$ for some constant $B_{\pi}$, and thus
\begin{equation}
-\mu(\pi_0\p{p+\zeta_T} - \pi_0\p{p-\zeta_T})
= -2\zeta_T\mu\pi_0'(p) + \oo(\zeta_T^2)
= 2\zeta_T V'(p)  + \oo(\zeta_T^2).
\end{equation}
Since $\sqrt{T} \zeta_T^2 \to 0$, we have that
\begin{equation}
\sqrt{T\zeta_T^2}\p{\frac{1}{2\zeta_T}\p{\frac{N_{+}(T, \zeta_T)}{T_{+}(T, \zeta_T)} - \frac{N_{-}(T, \zeta_T)}{T_{-}(T, \zeta_T)}} - V'(p)}
 \Rightarrow\mathcal{N} (0, \sigma_{\bar{\lambda}}^2(p)).
\end{equation}

\subsection{Proof of Theorem \ref{theo:second_esti}}
Similar to the proof of Theorem \ref{theo:first_esti}, we prove the theorem in two steps. In the first step, we establish a central limit theorem without specifying the asymptotic mean and variance. In the second step, we give explicit formulae for the mean and variance.

By definition, we can write the estimator as
\begin{equation}
\hat{\tau}_{\pi_0} = -\frac{\mu}{2\zeta}\p{\frac{T_{0, +}}{T_{+}} - \frac{T_{0, -}}{T_{ -}}}.
\end{equation}
In the first step, we establish the following proposition.
\begin{prop}
\label{prop:second_esti_triangular}
There exist $\check{\sigma}(p): \mathbb{R} \to \mathbb{R} $ and  $\check{\mu}(p) : \mathbb{R} \to \mathbb{R} $ such that under Assumptions \ref{assu:queue_length}-\ref{assu:two_intervals}, for data generated from either the interval switchback experiment or the regenerative switchback experiment,
\begin{equation}
	\label{eqn:CLT_unknown_para2}
\sqrt{T}
\begin{pmatrix}
\frac{T_{0, +}(T, \zeta_T)}{T_{ +}(T, \zeta_T)} - \check{\mu}(p+\zeta_T)\\
\frac{T_{0, -}(T, \zeta_T)}{T_{ -}(T, \zeta_T)} - \check{\mu}(p-\zeta_T)\\
\end{pmatrix}
\Rightarrow \mathcal{N}(\zerov,  \check{\Sigma}(p)),
\end{equation}
as $T \to \infty$, where
$
\check{\Sigma} = \begin{pmatrix}
	\check{\sigma}^2(p) & 0\\
	0 & \check{\sigma}^2(p)
\end{pmatrix}
$.
\end{prop}
Similar to Proposition \ref{prop:first_esti_triangular}, to show this proposition, we use some triangular array arguments to deal with the changing price perturbation $\zeta_T$. We note that the asymptotic variance $\check{\Sigma}(p)$ does not depend on the price perturbation, and we do not have an explicit form for it so far.

In the second step, we give explicit formulae for $\check{\mu}$ and $\check{\Sigma}$. We start by establishing a central limit theorem for the fixed-price setting.
\begin{prop}
\label{prop:second_esti_zero}
Set the price at $p$ throughout the experiment. Let $T_0(T, p)$ be the amount of time in state 0. Under Assumptions \ref{assu:queue_length}-\ref{assu:bound_smooth}, as $T \to \infty$,
\begin{equation}
\sqrt{T}
\p{\frac{T_0(T, p)}{T} - \check{\mu}(p)}
\Rightarrow \mathcal{N}(0,  \check{\sigma}^2(p)/2).
\end{equation}
The function $\check{\mu}$ and the value of $\check{\sigma}^2(p)$ are the same as in Proposition \ref{prop:second_esti_triangular}.
\end{prop}
Then we note that classical results in Markov processes also give a central limit theorem of $T_0(T, p)/T$. Let $Q(p)$ be the transition rate matrix of the continuous-time Markov chain induced by the queue length process. Let $Q^{\#}(p) = -(\onev\piv^\top - Q)^{-1} (I - \onev \piv^\top)$ be the group inverse of the matrix $Q(p)$. In the literature, $-Q^{\#}(p)$ is also known as the fundamental matrix \citep{asmussen2003applied}. A special case of Theorem 4.11 in \citet{asmussen2003applied} establishes a central limit theorem of $T_0(T, p)/T$; in particular, the asymptotic variance is expressed in terms of $Q^{\#}(p)$:

\begin{theo}[Corollary of Theorem 4.11 in \citet{asmussen2003applied}]
\label{theo:CLT_literature}
Set the price at $p$ throughout the experiment. Let $T_0(T, p)$ be the amount of time in state 0. Under Assumptions \ref{assu:queue_length}-\ref{assu:bound_smooth}, as $T \to \infty$,
\begin{equation}
\sqrt{T}
\p{\frac{T_0(T, p)}{T} - \pi_0(p)}
\Rightarrow \mathcal{N}(0,  -2\pi_0(p)Q^{\#}_{0,0}(p)).
\end{equation}
\end{theo}

Proposition \ref{prop:second_esti_zero} and Theorem \ref{theo:CLT_literature} give central limit theorems for the same sequence of random variables, and thus the means and asymptotic variances match, i.e.,
\begin{equation}
\label{eqn:mu_sigma_formula2}
\check{\mu}(p) = \pi_0(p), \textnormal{ and } \check{\sigma}^2(p) = -4\pi_0(p)Q^{\#}_{0,0}(p).
\end{equation}
It then remains to express $Q^{\#}_{0,0}(p)$ explicitly. The following proposition gives an explicit form of $Q^{\#}_{k,i}(p)$.
\begin{prop}
\label{prop:Q_inverse}
The matrix $Q^\#(p)$ defined in \eqref{eqn:Q_group_inverse_def} satisfies that
\begin{equation}
Q^\#_{k,i}(p) = \frac{\pi_i(p)}{\mu} \p{-\sum_{j = 1}^{\min(i,k)} \frac{1}{\pi_j(p)} + \sum_{j = 1}^k \frac{S_j(p)}{\pi_j(p)} + \sum_{j = 1}^i \frac{S_j(p)}{\pi_j(p)} - \sum_{j = 1}^{K} \frac{S_j^2(p)}{\pi_j(p)}},
\end{equation}
where $S_j(p) = \sum_{l = j}^K \pi_l(p)$.
\end{prop}
In particular, this implies that $Q^\#_{0,0}(p) = -\frac{\pi_0(p)}{\mu}\sum_{k = 1}^{K} \frac{S_k^2(p)}{\pi_k(p)}$, and thus
\begin{equation}
\label{eqn:mu_sigma_formula3}
\check{\sigma}^2(p) = \frac{4\pi_0^2(p)}{\mu}\sum_{k = 1}^{K} \frac{S_k^2(p)}{\pi_k(p)}
= \frac{2\sigma_{\pi_0}^2(p)}{\mu^2},
\end{equation}
where $S_k(p) = \sum_{j = k}^{K} \pi_j(p)$ and $\sigma_{\pi_0}^2(p)$ is defined in \eqref{eqn:def_sigma_pi0}.

Now, since the choice of the price $p$ is arbitrary in Propositions \ref{prop:second_esti_zero}-\ref{prop:Q_inverse} and Theorem \ref{theo:CLT_literature}, we can plug \eqref{eqn:mu_sigma_formula2} and \eqref{eqn:mu_sigma_formula3} back into Proposition \ref{prop:second_esti_triangular} and get
\begin{equation}
\label{eqn:time_converge}
\sqrt{T}
\begin{pmatrix}
\frac{T_{0, +}(T, \zeta_T)}{T_{ +}(T, \zeta_T)} - \pi_0(p+\zeta_T)\\
\frac{T_{0, -}(T, \zeta_T)}{T_{ -}(T, \zeta_T)} - \pi_0(p-\zeta_T)\\
\end{pmatrix}
\Rightarrow \mathcal{N}(\zerov,  \check{\Sigma}(p)),
\end{equation}
where
$
\check{\Sigma} = \begin{pmatrix}
	2\sigma_{\pi_0}^2(p)/\mu^2 & 0\\
	0 & 2\sigma_{\pi_0}^2(p)/\mu^2
\end{pmatrix}
$. Therefore,
\begin{equation}
\sqrt{T}\p{\p{\frac{T_{0, +}(T, \zeta_T)}{T_{ +}(T, \zeta_T)} - \frac{T_{0, -}(T, \zeta_T)}{T_{ -}(T, \zeta_T)}} - \p{\pi_0\p{p+\zeta_T} - \pi_0\p{p-\zeta_T}}}
 \Rightarrow\mathcal{N} \p{0, \frac{4\sigma_{\pi_0}^2(p)}{\mu^2}},
\end{equation}
as $T \to \infty$, where $\sigma_{\pi_0}^2(p)$ is defined in \eqref{eqn:def_sigma_pi0}. Finally, since $|\lambda_k''(p)| \leq B_2$ holds for any $k$ for some constant $B_2$, we have that $\abs{\pi''_0(p)} \leq B_{\pi}$ for some constant $B_{\pi}$, and thus
\begin{equation}
-\mu(\pi_0\p{p+\zeta_T} - \pi_0\p{p-\zeta_T})
= -2\zeta_T\mu\pi_0'(p) + \oo(\zeta_T^2)
= 2\zeta_T V'(p)  + \oo(\zeta_T^2).
\end{equation}
Since $\sqrt{T} \zeta_T^2 \to 0$, we have that
\begin{equation}
\sqrt{T\zeta_T^2}\p{-\frac{\mu}{2\zeta_T}\p{\frac{T_{0,+}(T, \zeta_T)}{T_{ +}(T, \zeta_T)} - \frac{T_{ 0, -}(T, \zeta_T)}{T_{ -}(T, \zeta_T)}} - V'(p)}
 \Rightarrow\mathcal{N} (0, \sigma_{\pi_0}^2(p)).
\end{equation}

\subsection{Proof of Theorem \ref{theo:model_CLT}}

We prove the theorem in three steps. In the first step, we establish a multivariate central limit theorem for the ratios $N_{k,+}/T_{k,+}$ and $N_{k,-}/T_{k,-}$ without specifying the asymptotic mean and variance. In the second step, we give explicit formulae for the mean and variance. Interestingly, we find that for $k \neq l$, the ratios $N_{k,+}/T_{k,+}$ and $N_{l,+}/T_{l,+}$ are asymptotically independent. Finally, in the third step, we combine the results and obtain a central limit theorem for $\hat{\Delta}_k(T, \zeta_T)$ and thus for $\hat{\tau}_{\model}(T, \zeta_T)$.

\begin{prop}
\label{prop:CLT_varying_para}
Define $\ratio_{+}$ to be a vector such that $\ratio_{+, k} = N_{k,+}(T, \zeta_T)/T_{k,+}(T, \zeta_T)$. Define $\ratio_{-}$ similarly.
There exist $\tilde{\mu}_{\ratio}(p): \mathbb{R} \to \mathbb{R}^K $ and  $\tilde{\Sigma}_{\ratio}(p) : \mathbb{R} \to \mathbb{R}^{K \times K} $ such that under Assumptions \ref{assu:queue_length}-\ref{assu:two_intervals}, for data generated from either the interval switchback experiment or the regenerative switchback experiment,
\begin{equation}
\sqrt{T} \p{\ratio_{+} - \tilde{\mu}_{\ratio}(p+\zeta_T), \ratio_{-} - \tilde{\mu}_{\ratio}(p-\zeta_T)}
 \Rightarrow  \mathcal{N}\p{\zerov, \check{\Sigma}_{\ratio}(p)},
\end{equation}
where
$
\check{\Sigma}_{\ratio}(p) = \begin{pmatrix}
	\tilde{\Sigma}_{\ratio}(p) & 0\\
	0 & \tilde{\Sigma}_{\ratio}(p)
\end{pmatrix}
$.
\end{prop}
Similar to the proof of Theorems \ref{theo:first_esti} and \ref{theo:second_esti}, to show the above proposition, we use some triangular array arguments to deal with the changing price perturbation $\zeta_T$. Here, the asymptotic variance $\tilde{\Sigma}_{\ratio}(p)$ does not depend on $\zeta_T$, and we do not have an explicit formula for it yet.

We then shift our focus to the fixed-price case. Here, we can give explicit formulae for $\tilde{\mu}_{\ratio}(p)$ and $\tilde{\Sigma}_{\ratio}(p)$.
\begin{prop}
\label{prop:CLT_fixed_para_tilde}
Set the price at $p$. Let $N_k(T,p)$ be the number of jumps from $k$ to $k+1$. Let $T_k(T,p)$ be the amount of time in state $k$. Under Assumptions \ref{assu:queue_length} and \ref{assu:bound_smooth},
\begin{equation}
\sqrt{T} (\tilde{\operatorname{err}}_0, \tilde{\operatorname{err}}_1, \dots, \tilde{\operatorname{err}}_{K-1}) \Rightarrow  \mathcal{N}\p{\vec{0}, \tilde{\Sigma}_{\ratio}(p)/2},
\end{equation}
where
$\tilde{\operatorname{err}}_k = N_k(T,p)/T_k(T,p) - \tilde{\mu}_{\ratio, k}(p)$. Here,  $\tilde{\mu}_{\ratio, k}$ and $\tilde{\Sigma}_{\ratio}$ are defined in Proposition \ref{prop:CLT_varying_para}.
\end{prop}
\begin{prop}
\label{prop:CLT_fixed_para}
Set the price at $p$. Let $N_k(T,p)$ be the number of jumps from $k$ to $k+1$. Let $T_k(T,p)$ be the amount of time in state $k$. Under Assumptions \ref{assu:queue_length} and \ref{assu:bound_smooth}, for any $k \in \cb{0, 1, \dots, K-1}$,
\begin{equation}
\sqrt{T} \p{\frac{N_k(T,p)}{T_k(T,p)} -  \lambda_{k}(p)} \Rightarrow  \mathcal{N}\p{0,
\frac{\lambda_{k}(p)}{\pi_k(p)}},
\end{equation}
as $T \to \infty$. Furthermore, the joint distribution converges to independent Gaussians:
\begin{equation}
\sqrt{T} (\operatorname{err}_0, \operatorname{err}_1, \dots, \operatorname{err}_{K-1}) \Rightarrow  \mathcal{N}\p{\vec{0}, \Sigma_{\ratio}(p)},
\end{equation}
where
$\operatorname{err}_k = N_k(T,p)/T_k(T,p) -  \lambda_{k}(p)$,
$\Sigma_{{\ratio},k,k}(p) = \lambda_{k}(p)/\pi_k(p)$ and $\Sigma_{\ratio, k, l} = 0$ for $k \neq l$.
\end{prop}
Propositions \ref{prop:CLT_fixed_para_tilde} and \ref{prop:CLT_fixed_para} together give an explicit formula for $\tilde{\mu}_{\ratio}$ and $\tilde{\Sigma}_{\ratio}$:
\begin{equation}
\label{eqn:mu_Sigma_ratio_formula}
\tilde{\mu}_{\ratio,k} = \lambda_k(p), \quad \tilde{\Sigma}_{\ratio,k,k} = \frac{2\lambda_k(p)}{\pi_k(p)}, \quad \textnormal{and } \tilde{\Sigma}_{\ratio,k,l} =0  \text{ for } k \neq l.
\end{equation}

Turning our attention back to the estimator of interest, we recall that
\begin{equation}
   \hat{\Delta}_{ k} = \frac{1}{\zeta_T} \frac{N_{ k, +}/T_{ k, +} - N_{ k, -}/T_{ k, -}}{N_{ k, +}/T_{ k, +} + N_{ k, -}/T_{ k, -}},
\end{equation}
Proposition \ref{prop:CLT_varying_para}, together with \eqref{eqn:mu_Sigma_ratio_formula}, implies that as $T \to \infty$,
\begin{equation}
\sqrt{T \zeta_T^2} \p{\hat{\Delta}_{ k}(T, \zeta_T) -  \frac{\lambda_{k}(p+\zeta_T) - \lambda_{k}(p-\zeta_T)}{\zeta_T(\lambda_{k}(p+\zeta_T) + \lambda_{k}(p-\zeta_T))}} \Rightarrow  \mathcal{N}\p{0,
\frac{1}{\lambda_{k}(p) \pi_k(p)}},
\end{equation}
and that $\hat{\Delta}_{ k}(T, \zeta_T)$ is asymptotically independent of the remaining estimators $\hat{\Delta}_{ l}(T, \zeta_T)$ for $l\neq k$. Then, since $\abs{\lambda_k''(p)} \leq B_2$ holds for any $k$,
\begin{equation}
\frac{\lambda_{k}(p+\zeta_T) - \lambda_{k}(p-\zeta_T)}{\zeta_T(\lambda_{k}(p+\zeta_T) + \lambda_{k}(p-\zeta_T))}
= \frac{\lambda_{k}'(p)}{\lambda_k(p)} + \oo(\zeta_T).
\end{equation}
Finally, since $\sqrt{T} \zeta_T^2 \to 0$, we have that as $T \to \infty$,
\begin{equation}
\sqrt{T \zeta_T^2} \p{\hat{\Delta}_{ k}(T, \zeta_T) -  \frac{\lambda_{k}'(p)}{\lambda_k(p)}} \Rightarrow  \mathcal{N}\p{0,
\frac{1}{\lambda_{k}(p) \pi_k(p)}}.
\end{equation}
Furthermore, the joint distribution converges to independent Gaussians:
\begin{equation}
\sqrt{T \zeta_T^2} (\operatorname{err}_{\Delta,0}, \operatorname{err}_{\Delta,1}, \dots, \operatorname{err}_{\Delta,K-1}) \Rightarrow  \mathcal{N}\p{\vec{0}, \Sigma_{\Delta}},
\end{equation}
where
$\operatorname{err}_{\Delta,k} = \hat{\Delta}_{k}(T, \zeta_T) -  \lambda_{k}'(p)/\lambda_{k}(p)$,
$\Sigma_{\Delta,k,k} = 1/(\pi_k(p)\lambda_{k}(p))$ and $\Sigma_{\Delta, k, l} = 0$ for $k \neq l$.

Then, recall that
\begin{equation}
\hat{\tau}_{\model} = \mu \frac{T_{ 0}}{T} \sum_{k = 0}^{K-1}  \hat{\Delta}_{ k}  \sum_{i = k+1}^{K} \frac{T_{i}}{T}.
\end{equation}
Note that the variance in $T_{i}/T$ is negligible in the overall variance calculation. In particular, we have that $T_{i}/T - \pi_i(p) = \oo_p\p{1/\sqrt{T}} = o_p\p{1/\sqrt{T \zeta_T^2}}$.

Then, we note that the asymptotic variance of $\hat{\Delta}_{k}(T, \zeta_T)$ equals $1/(\pi_k(p)\lambda_k(p)) = 1/(\pi_{k+1}(p)\mu)$. Since the $\hat{\Delta}_k(T, \zeta_T)$ are asymptotically independent, the asymptotic variance of the estimator equals
\begin{equation}
    (\mu^2 \pi_0^2(p)) \sum_{k = 0}^{K-1} \frac{1}{\pi_{k+1}(p)\mu} \p{\sum_{i = k+1}^{K} \pi_i(p)}^2
    = \mu\pi_0^2(p) \sum_{k = 1}^K \frac{S_k^2(p)}{\pi_k(p)} = \sigma^2_{\model}(p),
\end{equation}
where $S_k(p) = \sum_{j = k}^{K} \pi_j(p)$ and $\sigma^2_{\model}(p)$ is defined in \eqref{eqn:def_sigma_model}. Furthermore, the estimator is asymptotically normal, i.e.,
\begin{equation}
\sqrt{T \zeta_T^2} \p{ \hat{\tau}_{ \model}(T, \zeta_T) -  \mu \pi_{0}(p) \sum_{k = 0}^{K-1}  \frac{\lambda'_{k}(p)}{\lambda_{k}(p)} \sum_{i = k+1}^{K} \pi_i(p)} \Rightarrow \mathcal{N}\p{0, \sigma_{\model}^2(p)},
\end{equation}
and hence by \eqref{eqn:lambda_p_dev},
\begin{equation}
\sqrt{T \zeta_T^2} \p{ \hat{\tau}_{ \model}(T, \zeta_T) -  V'(p)} \Rightarrow \mathcal{N}\p{0, \sigma_{\model}^2(p)}.
\end{equation}

\subsection{Proof of Theorem \ref{theo:variance_comp}}
The equality \eqref{eqn:var_comp1} is trivial. For \eqref{eqn:var_comp2}, note that
\begin{equation}
\begin{split}
\frac{\sigma_{\bar{\lambda}}^2(p) - \sigma_{\model}^2(p)}{\mu}
& = (1-\pi_0(p)) + 2 \pi_0(p)\sum_{k = 1}^{K-1}\frac{S_k(p)S_{k+1}(p)}{\pi_k(p)} \\
& \qquad  \qquad - 2 \pi_0(p)(1-\pi_0(p))\sum_{k =1}^K \frac{S_k^2(p)}{\pi_k(p)} - \pi_0^2(p) \sum_{k = 1}^K \frac{S_k^2(p)}{\pi_k(p)}\\
& =  (1-\pi_0(p)) + \pi_0^2(p) \sum_{k = 1}^K \frac{S_k^2(p)}{\pi_k(p)}\\
&\qquad  \qquad  + 2 \pi_0(p)\sum_{k = 1}^{K-1}\frac{S_k(p)S_{k+1}(p)}{\pi_k(p)} - 2 \pi_0(p)\sum_{k =1}^K \frac{S_k^2(p)}{\pi_k(p)}\\
&= \sum_{k=1}^K \pi_k(p) +  \pi_0^2(p) \sum_{k = 1}^{K} \frac{S_k^2(p)}{\pi_k(p)}
- 2 \pi_0(p)\sum_{k =1}^{K} S_k(p)\\
& = \sum_{k=1}^K \p{\pi_k(p) +  \pi_0^2(p) \frac{S_k^2(p)}{\pi_k(p)}
- 2 \pi_0(p) S_k(p)}\\
& =  \sum_{k=1}^K \p{\sqrt{\pi_k(p)} -  \pi_0(p) \frac{S_k(p)}{\sqrt{\pi_k(p)}}}^2 \geq 0.
\end{split}
\end{equation}
The equality holds if and only if $\sqrt{\pi_k(p)} = \pi_0(p) \frac{S_k(p)}{\sqrt{\pi_k(p)}}$, i.e., $S_k(p) = \pi_k(p)/\pi_0(p)$. The condition $S_k(p) = \pi_k(p)/\pi_0(p)$ further implies that
$\pi_k(p) = (\pi_{k}(p) - \pi_{k+1}(p))/\pi_0(p)$. This then implies that $\pi_{k+1}(p) = (1- \pi_0(p))\pi_k(p)$, and hence $\lambda_k(p) = (1 - \pi_0(p))\mu$ for all $k \geq 0$.

\subsection{Proof of Theorem \ref{theo:variance_estimator}}
Since the asymptotic variances of the estimators are functions of $\pi_k(p)$, it suffices to show that $\hat{\pi}_k$ is consistent for $\pi_k(p)$. Note that by definition, $\hat{\pi}_k = T_k/T$.

For the interval switchback and the regenerative switchback, \eqref{eqn:time_converge} implies that $T_{0,+}/T_+ \stackrel{p}{\to} \pi_0(p)$ and $T_{0,-}/T_- \stackrel{p}{\to} \pi_0(p)$. Since $T_0 = T_{0,+} + T_{0,-}$ and $T = T_+ + T_-$, we have that
\begin{equation}
   \min \p{\frac{T_{0,-}}{T_{-}}, \frac{T_{0,+}}{T_{+}}} \leq  \frac{T_0}{T} \leq \max \p{\frac{T_{0,-}}{T_{-}}, \frac{T_{0,+}}{T_{+}}}.
\end{equation}
Therefore, $T_0/T \stackrel{p}{\to} \pi_0(p)$. The same argument can be used for other $k$: $T_k/T \stackrel{p}{\to} \pi_k(p)$. We omit the details for conciseness.

\subsection{Proof of Theorem \ref{theo:LE_CLT}}

Since the length of the queue is bounded above by $K$, we can rewrite the estimator as
\begin{equation}
\hat{\tau}_{\LE}(T, \zeta_T) = \mu \frac{T_{ 0}}{T} \sum_{k = 0}^{K-1}\hat{\Delta}_{k} \sum_{i = k+1}^{K} \frac{T_{i}}{T}.
\end{equation}

We start by noting that the variance in $T_{i}/T$ is negligible in the overall variance calculation. To show this, we apply Theorem 4.11 in \citet{asmussen2003applied} and obtain, for any $i \in \cb{0, 1, \dots, K}$, $T_{i}/T - \pi_i(p, \zeta_T) = \oo_p\p{1/\sqrt{T}}$,
where $\pi_i(p, \zeta_T)$ is the steady-state probability of state $i$ when the state-dependent arrival rate is set to $(\lambda_i(p-\zeta_T) + \lambda_i(p+\zeta_T))/2$.
Therefore, we have that $T_{ i}/T - \pi_i(p, \zeta_T) = o_p\p{1/\sqrt{T \zeta_T^2}}$.
By Lemma \ref{lemma:pi_converge}, we have that $\abs{\pi_i(p, \zeta_T) - \pi_i(p)} \leq C \zeta_T^2 = o\p{1/\sqrt{T}}$ for some constant $C$. Therefore, we have that
\begin{equation}
\label{eqn:LP_time_converge}
    \frac{T_{i}}{T} - \pi_i(p) = o_p\p{1/\sqrt{T \zeta_T^2}}.
\end{equation}

Then it remains to quantify the joint distribution of $\hat{\Delta}_0, \dots, \hat{\Delta}_{K-1}$. To this end, we establish the following proposition:

\begin{prop}
\label{prop:CLT_ratio}
Under Assumptions \ref{assu:queue_length}-\ref{assu:bound_smooth}, for any $k \in \cb{0, 1, \dots, K-1}$,
\begin{equation}
\sqrt{T \zeta_T^2} \p{\hat{\Delta}_{k}(p, T, \zeta_T) -  \frac{\lambda_{k}'(p)}{\lambda_{k}(p)}} \Rightarrow  \mathcal{N}\p{0,
\frac{1}{\pi_k(p)\lambda_{k}(p)}},
\end{equation}
as $T \to \infty$. Furthermore, the joint distribution converges to independent Gaussians:
\begin{equation}
\sqrt{T \zeta_T^2} (\operatorname{err}_0, \operatorname{err}_1, \dots, \operatorname{err}_{K-1}) \Rightarrow  \mathcal{N}\p{\vec{0}, \Sigma_{\LE}},
\end{equation}
where
$\operatorname{err}_k = \hat{\Delta}_{k}(p, T, \zeta_T) -  \lambda_{k}'(p)/\lambda_{k}(p)$,
$\Sigma_{{\LE},k,k} = 1/(\pi_k(p)\lambda_{k}(p))$ and $\Sigma_{{\LE},k,l} = 0$ for $k \neq l$.
\end{prop}
The proposition states that each $\hat{\Delta}_k$ has asymptotic variance $1/(\pi_k(p)\lambda_k(p))$, equivalently $1/(\pi_{k+1}(p)\mu)$. Since the $\hat{\Delta}_k$ are asymptotically independent, the asymptotic variance of the estimator equals
\begin{equation}
    (\mu^2 \pi_0^2(p)) \sum_{k = 0}^{K-1} \frac{1}{\pi_{k+1}(p)\mu} \p{\sum_{i = k+1}^{K} \pi_i(p)}^2
    = \mu\pi_0^2(p) \sum_{k = 1}^K \frac{S_k^2(p)}{\pi_k(p)} = \sigma^2_{\LE}(p),
\end{equation}
where $S_k(p) = \sum_{j = k}^{K} \pi_j(p)$ and $\sigma^2_{\LE}(p)$ is defined in \eqref{eqn:def_sigma_LE}. Furthermore, by Proposition \ref{prop:CLT_ratio}, we can show that the estimator is asymptotically normal, i.e.,
\begin{equation}
\sqrt{T \zeta_T^2} \p{ \hat{\tau}_{\LE}(T, \zeta_T) -  \mu \pi_{0}(p) \sum_{k = 0}^{K-1}  \frac{\lambda'_{k}(p)}{\lambda_{k}(p)} \sum_{i = k+1}^{K} \pi_i(p)} \Rightarrow \mathcal{N}\p{0, \sigma_{\LE}^2(p)},
\end{equation}
and hence by \eqref{eqn:lambda_p_dev},
\begin{equation}
\sqrt{T \zeta_T^2} \p{ \hat{\tau}_{\LE}(T, \zeta_T) -  V'(p)} \Rightarrow \mathcal{N}\p{0, \sigma_{\LE}^2(p)}.
\end{equation}

\subsection{Details and Proof of Theorem \ref{theo:non_stationarity}}
\label{subsection:non_station_proof}
Let $C$ be a constant such that $C > \max_b \abs{V^{(b)}{}'(p)}/\mu$. We let $\tilde{\tau}_{\LE}(s_T)$ be a truncated version of $\hat{\tau}_{\LE}(s_T)$:
\begin{equation}
\label{eqn:truncation}
\tilde{\tau}_{\LE}(s_T) = \mu \frac{s_T}{T}\sum_w \max\p{\min\p{\frac{T_{0,w}}{s_T} \sum_{k = 0}^{K-1} \hat{\Delta}_{k,w} \sum_{i = k+1}^{K} \frac{T_{i,w}}{s_T} ,C},-C}.
\end{equation}

\paragraph{Proof of Theorem \ref{theo:non_stationarity}}
For $u \in (0,1)$, define the function
\begin{equation}
    f_T(u) = \max\p{\min\p{\frac{T_{0,w}}{s_T} \sum_{k = 0}^{K-1} \hat{\Delta}_{k,w} \sum_{i = k+1}^{K}\frac{T_{i,w}}{s_T} ,C},-C},
\end{equation}
for $(w-1)s_T/T\leq u < w s_T/T$. Then $\tilde{\tau}_{\LE}(s_T) = \mu\int_0^1 f_T(u)\,du$.

For any $u \in (0,1)$ that is not a change point, let $b(u)\in\{1,\dots,B\}$ be the unique index such that $t^{(b(u)-1)} < u < t^{(b(u))}$, and let $w_T(u)$ denote the corresponding window number, i.e., $(w_T(u)-1)s_T/T \leq u < w_T(u)s_T/T$. For sufficiently large $T$, the state-dependent arrival rate throughout the entire window $((w_T(u)-1)s_T, w_T(u)s_T)$ is $\lambda_k^{(b(u))}(p)$. Furthermore, the window length $s_T$ increases, satisfying $\sqrt{s_T}\zeta_T \to \infty$ and $\sqrt{s_T}\zeta_T^2 \to 0$. Thus, we can apply Proposition \ref{prop:CLT_ratio} to the corresponding windows and obtain
\begin{equation}
\hat{\Delta}_{k,w_T(u)} \stackrel{p}{\to} \frac{\lambda_k^{(b(u))}{}'(p)}{\lambda_k^{(b(u))}(p)}.
\end{equation}
Additionally, by invoking arguments from the proof of Theorem \ref{theo:LE_CLT}, we derive that
\begin{equation}
\frac{T_{i,w_T(u)}}{s_T} \stackrel{p}{\to} \pi_i^{(b(u))}(p).
\end{equation}
Consequently, it becomes evident that
\begin{equation}
f_T(u) \stackrel{p}{\to} \pi_0^{(b(u))}(p) \sum_{k=0}^{K-1}\frac{\lambda_k^{(b(u))}{}'(p)}{\lambda_k^{(b(u))}(p)} \sum_{i=k+1}^K \pi_i^{(b(u))}(p) = \frac{V^{(b(u))}{}'(p)}{\mu}.
\end{equation}

We next show that $\int_0^1 f_T(u)\,du \stackrel{p}{\to} \int_0^1 V^{(b(u))}{}'(p)\,du / \mu = V'(p)/\mu$. Let $\mu_B$ denote the Borel measure. For any $\epsilon > 0$,
\begin{equation}
\begin{split}
    &\PP{ \abs{ \int_0^1 f_T(u)\,du -  \int_0^1 \frac{V^{(b(u))}{}'(p)}{\mu}\,du} > \epsilon}\\
    &\qquad \qquad \leq \PP{ \mu_B\p{\cb{u: \abs{f_T(u) - V^{(b(u))}{}'(p)/\mu} > \epsilon/2}} > \epsilon/(4C) }\\
    &\qquad \qquad \leq \frac{4C}{\epsilon}\EE{ \mu_B\p{\cb{u: \abs{f_T(u) - V^{(b(u))}{}'(p)/\mu} > \epsilon/2}} }\\
    & \qquad \qquad = \frac{4C}{\epsilon} \int_0^1 \PP{\abs{f_T(u) - V^{(b(u))}{}'(p)/\mu} > \epsilon/2}\,du
    \to 0.
    \end{split}
\end{equation}
Here, the second inequality follows from Markov's inequality, and the last line follows from the Dominated Convergence Theorem.
Therefore,
\begin{equation}
    \int_0^1 f_T(u)\,du \stackrel{p}{\to} \int_0^1 V^{(b(u))}{}'(p)\,du / \mu = V'(p)/\mu,
\end{equation}
and thus $\tilde{\tau}_{\LE}(s_T) \stackrel{p}{\to} V'(p)$ as $T \to \infty$.

\subsection{Non-stationarity with a diverging number of change points}
\label{subsection:non-stationarity_diverging}
\begin{assu}[Periodic non-stationarity]
	\label{assu:change_point2}
	Let the time horizon be divided into $m_T$ blocks of equal length $B_T = T / m_T$.
	Within each block, there are finitely many change points at which the arrival rate may shift,
	and the pattern of these changes repeats identically across all blocks.
		Specifically, there exist time points
		$0 = t^{(0)} < t^{(1)} < \dots < t^{(B-1)} < t^{(B)} = 1$
		and corresponding arrival rate functions
		$\lambda_k^{(1)}(p), \dots, \lambda_k^{(B)}(p)$
		such that, for any $b\in\{1,\dots,B\}$, $l\in\{0,\dots,m_T-1\}$, and
		$t \in ((l+t^{(b-1)})B_T,\, (l+t^{(b)})B_T]$,
	the state-dependent arrival rate at time $t$ satisfies
	\begin{equation}
		\lambda_{k,t}(p) = \lambda_k^{(b)}(p).
	\end{equation}
\end{assu}

Similar to Assumption \ref{assu:change_point}, under Assumption \ref{assu:change_point2}, the average processing rate $V(p)$ becomes
\begin{equation}
	V(p) = \sum_{b = 1}^B(t^{(b)} - t^{(b-1)}) V^{(b)}(p)
	= \sum_{b = 1}^B(t^{(b)} - t^{(b-1)})\p{\sum_{k=0}\pi_k^{(b)}(p)\lambda_k^{(b)}(p)},
\end{equation}
and thus the quantity of interest is the gradient:
\begin{equation}
	V'(p) = \sum_{b = 1}^B(t^{(b)} - t^{(b-1)}) V^{(b)}{}'(p).
\end{equation}

\begin{theo}
	\label{theo:non_stationarity2}
	Under Assumptions \ref{assu:queue_length}, \ref{assu:zeta}, \ref{assu:bound_smooth} and \ref{assu:change_point2}, assume further that $s_T/B_T \to 0$ and $s_T \zeta_T^2 \to \infty$. Let $\tilde{\tau}_{\LE}(s_T)$ be a truncated version of $\hat{\tau}_{\LE}(s_T)$, as defined in \eqref{eqn:truncation} in Section \ref{subsection:non_station_proof}. Using the weighted direct effect estimator with a user-level experiment yields
	\begin{equation}
		\tilde{\tau}_{\LE}(s_T) \stackrel{p}{\to} V'(p),
	\end{equation}
	as $T \to \infty$.
\end{theo}

\begin{proof}
The proof follows the same structure as that of Theorem \ref{theo:non_stationarity}.
The only modification is a slightly different definition of the function $f_T$ used in the dominated convergence argument, where we permute the time intervals to align periods with identical arrival rates.

For any $u \in (0,1)$ that is not one of the points $t^{(0)},\dots,t^{(B)}$, let $b(u)\in\{1,\dots,B\}$ be the unique index such that $t^{(b(u)-1)}<u<t^{(b(u))}$. We can uniquely write
\begin{equation}
		u = t^{(b(u)-1)} + (l+\alpha)\frac{B_T}{T}\p{t^{(b(u))} - t^{(b(u)-1)}}
\end{equation}
for some $l\in\{0,\dots,m_T-1\}$ and $\alpha \in [0,1)$. Define the measure-preserving rearrangement
\begin{equation}
		\tilde{u}(u) = \frac{B_T}{T}\p{l + t^{(b(u)-1)} + \p{t^{(b(u))} - t^{(b(u)-1)}}\alpha}.
\end{equation}
This maps the rearranged interval $(0,1)$ back to the original time scale and groups subintervals corresponding to the same arrival rate. Let $w_T(u)$ denote the window index corresponding to $\tilde{u}(u)$, i.e., $(w_T(u)-1)s_T/T \leq \tilde{u}(u) < w_T(u)s_T/T$.
Define
\begin{equation}
	f_T(u) = \max\p{\min\p{\frac{T_{0,w_T(u)}}{s_T} \sum_{k = 0}^{K-1} \hat{\Delta}_{k,w_T(u)} \sum_{i = k+1}^{K}\frac{T_{i,w_T(u)}}{s_T} ,C},-C}.
\end{equation}
Let $G_T$ be the set of $u$ for which the window containing $\tilde u(u)$ does not cross a change point. The rearrangement is measure preserving, and each of the at most $m_TB$ change points (including block boundaries) contaminates a set of normalized measure at most $2s_T/T$. Consequently,
\begin{equation}
\mu_B(G_T^c) \leq 2m_TB\frac{s_T}{T}
=2B\frac{s_T}{B_T}\to 0.
\end{equation}
For $u\in G_T$ with $b(u)=b$, the window is generated entirely under $\lambda_k^{(b)}(p)$. The argument in the proof of Theorem~\ref{theo:non_stationarity} therefore applies uniformly over such windows (the initial queue length takes values in the finite set $\{0,\dots,K\}$), giving, for every $\epsilon>0$,
\begin{equation}
\max_{1\leq b\leq B}\ \sup_{u\in G_T:\,b(u)=b}
\PP{\abs{f_T(u)-V^{(b)}{}'(p)/\mu}>\epsilon}\to0.
\end{equation}
Because both the truncated statistic and its limit are bounded in absolute value by $C$, Markov's inequality and Fubini's theorem now yield
\begin{equation}
\int_0^1 f_T(u)\,du-
\int_0^1\frac{V^{(b(u))}{}'(p)}{\mu}\,du\stackrel{p}{\to}0.
\end{equation}
Finally, the rearrangement is measure preserving, so the integral of $f_T$ equals the average of the original window statistics. Hence
\begin{equation}
\tilde{\tau}_{\LE}(s_T)=\mu\int_0^1f_T(u)\,du
\stackrel{p}{\to}\int_0^1V^{(b(u))}{}'(p)\,du=V'(p).
\end{equation}

\end{proof}

\section*{Acknowledgments}
This work was supported by the Office of Naval Research [Grant N00014-24-1-2091] and the Stanford GSB Business, Government \& Society (BGS) Initiative.
S.L. thanks Shuangping Li for discussions regarding the proof of Lemma \ref{lemma:uniform_intergral}.

\bibliographystyle{plainnat}
\bibliography{references}

\newpage

\appendix
\section{Additional Proofs}
\label{section:addition_proofs}

\subsection{Notation and preliminaries}
\label{subsection:notation}

For summations, if $a > b$, then we take the convention that $\sum_{i = a}^b x_i = 0$. We use zero-based matrix indexing. Let $\ev_k$ be the vector with its $k$-th entry (zero-based indexing) 1 and other entries 0. Let $\onev$ be the vector with all entries equal to 1. For a matrix $A$, let $\Norm{A}_{op}$ be the operator norm of matrix $A$.

Let $\nuv = (\nu_0, \dots, \nu_K)$ be the initial distribution of the queue length process.
Let
\begin{equation}
\label{eqn:Q_matrix}
Q(p) = \begin{bmatrix}
	-\lambda_0(p) & \lambda_0(p) & & & \\
	\mu & - \lambda_1(p) - \mu & \lambda_1(p) & &\\
	& \mu & - \lambda_2(p) - \mu & \lambda_2(p)  &\\
	& & &\dots &\\
	& & & 	\mu & -\mu
\end{bmatrix}
\end{equation}
be the transition rate matrix (a.k.a.~generator) of the continuous-time Markov chain corresponding to the queue length process. Throughout the section, we assume that $\lambda_k(p) > 0$ for any $p>0$ and $k = 0,1,\dots, K-1$. Let $\piv(p)$ denote the vector of steady-state probabilities. Let $D_{\sqrt{\pi}}(p)$ be the diagonal matrix with entries $D_{\sqrt{\pi},i,i}(p) = \sqrt{\pi_i(p)}$. Let $\tilde{Q}(p) = D_{\sqrt{\pi}}(p) Q(p) D_{\sqrt{\pi}}^{-1}(p)$. By properties of the steady-state probabilities, we have that $\tilde{Q}(p)$ is a symmetric matrix. Let $\tilde{Q}(p) = \tilde{U}(p) D(p) \tilde{U}(p)^\top$ be the eigen-decomposition of $\tilde{Q}(p)$, where $\tilde{U}(p)$ is an orthogonal matrix and $D(p)$ is a diagonal matrix. With this, we can express $Q(p)$ as
\begin{equation}
	\label{eqn:eigen_decom_Q}
	Q(p) = D_{\sqrt{\pi}}^{-1}(p) \tilde{Q}(p) D_{\sqrt{\pi}}(p) =  D_{\sqrt{\pi}}^{-1}(p)  \tilde{U}(p) D(p) \tilde{U}(p)^\top D_{\sqrt{\pi}}(p)
	= U(p) D(p) U^{-1}(p),
\end{equation}
where $U(p) = D_{\sqrt{\pi}}^{-1}(p)  \tilde{U}(p)$.

The continuous-time Markov chain is clearly irreducible, and thus by Perron–Frobenius theorem, $Q(p)$ has a zero eigenvalue, and the rest of the eigenvalues are negative. Without loss of generality, we assume $D_{0,0}(p) = 0$. Let $Q^\#(p)$ be the group inverse of $Q(p)$, i.e.,
\begin{equation}
	\label{eqn:Q_group_inverse_def}
	Q^\#(p) = U(p) D^\#(p) U^{-1}(p),
\end{equation}
where $D^\#_{0,0}(p) = 0$ and $D^\#_{k,k}(p) = D_{k,k}^{-1}(p)$ for $k > 0$.\footnote{Another way of defining $Q^\#$ is $Q^\#(p) = -(\onev\piv^\top - Q)^{-1} (I - \onev \piv^\top)$. The two definitions are equivalent.}

Corresponding to the zero eigenvalue, the right eigenvector of $Q(p)$  is $\onev$, and the left eigenvector is $\piv(p)$. Thus we can decompose $U(p)$ into two parts: $U(p) = (\onev, U_2(p))$, where $U_2(p)$ is a $(K+1) \times K$ matrix. We can do the same for $U^{-1}(p)$:  $U^{-1}(p) = \p{\piv(p), U_2^{-1}(p)}^\top$, where  $U_2^{-1}(p)$ is a $(K+1) \times K$ matrix.

For a matrix $M$, we write $\me(M) = \max_{i,j} \abs{M_{ij}}$ as the maximum of the absolute values of the entries of $M$. For any $\alpha \geq \me(Q(p))$, let $P(p, \alpha) = I + Q(p)/\alpha$. By construction, $P(p, \alpha)$ is a stochastic matrix with nonnegative entries and $P(p, \alpha) \onev =  \onev$. By \eqref{eqn:eigen_decom_Q}, we can also write
\begin{equation}
	\label{eqn:eigen_decom_P}
	P(p, \alpha) = U(p) \p{I + \frac{1}{\alpha}D(p) } U^{-1}(p).
\end{equation}
Because $Q(p)$ is similar to a symmetric matrix, its eigenvalues are real. Gershgorin's circle theorem gives $-2\me(Q(p))\leq D_{k,k}(p)\leq0$. Hence, whenever $\alpha\geq2\me(Q(p))$, all eigenvalues $1+D_{k,k}(p)/\alpha$ of $P(p,\alpha)$ lie in $[0,1]$.

For a fixed price $p_0$ and price perturbation $\zeta_0$, let
\begin{equation}
\label{eqn:alpha_0}
	\alpha_0(p_0, \zeta_0) = 2\sup_{p \in [p_0 - \zeta_0, p_0 + \zeta_0]} \me(Q(p)) = 2\sup_{p \in [p_0 - \zeta_0, p_0 + \zeta_0]} \max_{i,j} \abs{Q_{ij}(p)}.
\end{equation}
Furthermore, let
\begin{equation}
\label{eqn:beta_0}
\beta_0(p_0, \zeta_0) = \inf_{p \in [p_0 - \zeta_0, p_0 + \zeta_0]} \min_{k \geq 1} \abs{D_{k,k}(p)}.
\end{equation}
Note that $\beta_0(p_0, \zeta_0) > 0$ since $\abs{D_{k,k}(p)} > 0$ for any $k \geq 1$ and $p > 0$, and eigenvalues are continuous functions of matrices. Further, the preceding eigenvalue bound gives $\alpha_0 \geq \beta_0$.
In the following sections, we use the notation $C(A,B, \dots, )$ to denote constants that depend on $A, B, \dots, $ but not on other parameters. The notation $C$ may take different values at different places.

\subsection{Lemmas}

\begin{lemm}
\label{lemma:operator_norm}
Let $R_i(p,p_0, \zeta_0) = P^i (p, \alpha_0(p_0, \zeta_0)) - \onev \piv(p)^\top$, where $\alpha_0$ is defined in \eqref{eqn:alpha_0} and $P(p, \alpha) = I + Q(p)/\alpha$. Then there exist constants $C_1(K, p_0, \zeta_0) > 0$ and $C_2(K, p_0, \zeta_0) \in (0,1)$ such that for any $p \in \sqb{p_0 - \zeta_0, p_0 + \zeta_0}$ and any $i \in \cb{0,1,2,\dots}$,
\begin{equation}
    \Norm{R_i(p,p_0, \zeta_0)}_{op} \leq C_1(K, p_0, \zeta_0) C_2^i(K, p_0, \zeta_0),
\end{equation}
where $\Norm{}_{op}$ stands for the operator norm of a matrix.
\end{lemm}
\begin{proof}
We have shown in \eqref{eqn:eigen_decom_P} that $P (p, \alpha_0(p_0, \zeta_0)) = U(p) \p{I + \frac{1}{\alpha_0(p_0, \zeta_0)}D(p) } U^{-1}(p)$. Therefore,
\begin{equation}
   P^i (p, \alpha_0(p_0, \zeta_0)) =  U(p) \p{I + \frac{1}{\alpha_0(p_0, \zeta_0)}D(p) }^i U^{-1}(p).
\end{equation}
As discussed in Section  \ref{subsection:notation}, $D_{0,0}(p)$ is zero, and all the remaining diagonal entries $D_{k,k}(p)$ are negative. We also note from the discussion in Section \ref{subsection:notation} that $U(p) = (\onev, U_2(p))$, where $U_2(p)$ is a $(K+1) \times K$ matrix and  $U^{-1}(p) = \p{\piv(p), U_2^{-1}(p)}^\top$, where $U_2^{-1}(p)$ is a $(K+1) \times K$ matrix. Therefore,
\begin{equation}
\begin{split}
   R_i(p,p_0, \zeta_0) &= P^i (p, \alpha_0(p_0, \zeta_0)) - \onev\piv^\top \\
   &=  U(p) \p{\p{I + \frac{1}{\alpha_0(p_0, \zeta_0)}D(p) }^i - \ev_0\ev_0^\top} U^{-1}(p).
\end{split}
\end{equation}
Let $\Sigma = \p{\p{I + \frac{1}{\alpha_0(p_0, \zeta_0)}D(p) }^i - \ev_0\ev_0^\top}$. The matrix $\Sigma$ is diagonal with $\Sigma_{0,0} = 0$ and
\begin{equation}
0\leq\Sigma_{k,k} = \p{1+ \frac{D_{k,k}(p)}{\alpha_0(p_0, \zeta_0)}}^i \leq \p{1-\frac{\beta_0(p_0, \zeta_0)}{\alpha_0(p_0, \zeta_0)}}^i
\end{equation}
for $k \geq 1$, where $\beta_0$ is defined in \eqref{eqn:beta_0}.
Let $\rho = 1- \beta_0(p_0, \zeta_0)/\alpha_0(p_0, \zeta_0) \in [0,1)$. Since $\tilde U(p)$ is orthogonal and the entries of $\piv(p)$ are bounded away from zero uniformly on the compact interval $[p_0-\zeta_0,p_0+\zeta_0]$, there is a finite constant $C_U$ such that
\begin{equation}
\sup_{p \in [p_0-\zeta_0,p_0+\zeta_0]}
\Norm{U(p)}_{op}\Norm{U^{-1}(p)}_{op} \leq C_U.
\end{equation}
Consequently,
\begin{equation}
\Norm{R_i(p,p_0,\zeta_0)}_{op}
\leq \Norm{U(p)}_{op}\Norm{U^{-1}(p)}_{op}\Norm{\Sigma}_{op}
\leq C_U\rho^i.
\end{equation}
Choosing any $C_2 \in (\rho,1)$ and taking $C_1=C_U$ proves the result.
\end{proof}

\begin{lemm}
	\label{lemma:sum_of_power}
There exist constants $C_1(K, p_0, \zeta_0) >0$ and $C_2(K, p_0, \zeta_0) \in (0,1)$ such that for any $p \in \sqb{p_0 - \zeta_0, p_0 + \zeta_0}$ and any $m \in \mathbb{N}_+$,
\begin{equation}
	\begin{split}
	& \Norm{ \sum_{i = 0}^{m-1} P^i (p, \alpha_0(p_0, \zeta_0)) - m \onev \piv(p)^\top + \alpha_0(p_0, \zeta_0) Q^\#(p) }_{op} \\
	& \qquad \qquad \leq C_1(K, p_0, \zeta_0) C_2^m(K, p_0, \zeta_0).
	\end{split}
\end{equation}
Here $\alpha_0$ is defined in \eqref{eqn:alpha_0}.
\end{lemm}
\begin{proof}
We note that by \eqref{eqn:eigen_decom_P},
\begin{equation}
	\begin{split}
		  \sum_{i = 0}^{m-1} P^i (p, \alpha_0(p_0, \zeta_0))
		  & = \sum_{i = 0}^{m-1} \p{U(p) \p{I + \frac{1}{\alpha_0(p_0, \zeta_0)}D(p) } U^{-1}(p)}^i\\
		 & =  \sum_{i = 0}^{m-1} U(p) \p{I + \frac{1}{\alpha_0(p_0, \zeta_0)}D(p) }^i U^{-1}(p)\\
		 & =  U(p)  \p{\sum_{i = 0}^{m-1}\p{I + \frac{1}{\alpha_0(p_0, \zeta_0)}D(p) }^i } U^{-1}(p).
	\end{split}
\end{equation}
As discussed in Section  \ref{subsection:notation}, $D_{0,0}(p)$ is zero, and all the remaining diagonal entries are negative. Therefore,
\begin{equation}
\sum_{i = 0}^{m-1}\p{1 + \frac{1}{\alpha_0(p_0, \zeta_0)}D_{0,0}(p) }^i  = \sum_{i = 0}^{m-1} 1 = m.
\end{equation}
At the same time, for $k \geq 1$,
\begin{equation}
\sum_{i = 0}^{m-1}\p{1 + \frac{1}{\alpha_0(p_0, \zeta_0)}D_{k,k}(p) }^i  = -\frac{\alpha_0(p_0, \zeta_0)}{D_{k,k}(p)} + \frac{\alpha_0(p_0, \zeta_0)}{D_{k,k}(p)} \p{1 + \frac{D_{k,k}(p)}{\alpha_0(p_0, \zeta_0)} }^m.
\end{equation}
where
\begin{equation}
0 \leq \p{1 + \frac{D_{k,k}(p) }{\alpha_0(p_0, \zeta_0)}}^m
\leq \p{1 - \frac{\beta_0(p_0, \zeta_0)}{\alpha_0(p_0, \zeta_0)}}^m.
\end{equation}
Thus,
\begin{equation}
	\begin{split}
	&\sum_{i = 0}^{m-1} P^i (p, \alpha_0(p_0, \zeta_0))\\
	& \qquad   =  U(p)  \begin{bmatrix}
		m & & & \\
		&-\frac{\alpha_0(p_0, \zeta_0)}{D_{1,1}(p)}  & &\\
		& &\dots &\\
		& &  &-\frac{\alpha_0(p_0, \zeta_0)}{D_{K,K}(p)}
	\end{bmatrix} U^{-1}(p) + U(p)\Sigma U^{-1}(p) \\
	& \qquad  =  (\onev, U_2(p))  \begin{bmatrix}
		m & & & \\
		&-\frac{\alpha_0(p_0, \zeta_0)}{D_{1,1}(p)}  & &\\
		& &\dots &\\
		& &  &-\frac{\alpha_0(p_0, \zeta_0)}{D_{K,K}(p)}
	\end{bmatrix} \p{\piv(p), U_2^{-1}(p)}^\top + U(p)\Sigma U^{-1}(p) \\
   & \qquad  = m \onev \piv(p)^\top - \alpha_0(p_0, \zeta_0) U(p) D^\#(p) U^{-1}(p) + U(p)\Sigma U^{-1}(p)\\
   & \qquad  = m \onev \piv(p)^\top - \alpha_0(p_0, \zeta_0) Q^\#(p)  + U(p)\Sigma U^{-1}(p),
\end{split}
\end{equation}
where $\Sigma$ is a diagonal matrix with $\Sigma_{0,0} = 0$ and $\Sigma_{k,k} = \frac{\alpha_0(p_0, \zeta_0)}{D_{k,k}(p)} \p{1 + \frac{D_{k,k}(p)}{\alpha_0(p_0, \zeta_0)} }^m$ for $k \geq 1$. Therefore,
\begin{equation}
\me(\Sigma) \leq   \frac{\alpha_0(p_0, \zeta_0)}{\beta_0(p_0, \zeta_0)}\p{1 - \frac{\beta_0(p_0, \zeta_0)}{\alpha_0(p_0, \zeta_0)}}^m.
\end{equation}
Let $\rho = 1 - \beta_0(p_0, \zeta_0)/\alpha_0(p_0, \zeta_0) \in [0,1)$. As in the proof of Lemma~\ref{lemma:operator_norm}, the condition numbers $\Norm{U(p)}_{op}\Norm{U^{-1}(p)}_{op}$ are uniformly bounded over $p \in [p_0-\zeta_0,p_0+\zeta_0]$. It follows that
\begin{equation}
\Norm{U(p)\Sigma U^{-1}(p)}_{op}
\leq C_U\frac{\alpha_0(p_0,\zeta_0)}{\beta_0(p_0,\zeta_0)}\rho^m
\end{equation}
for a finite constant $C_U$. Choosing any $C_2\in(\rho,1)$ and adjusting the prefactor $C_1$ proves the result.
\end{proof}

\begin{lemm}[Number of eligible jumps]
\label{lemma:num_eligible_jumps}
Consider the continuous-time Markov chain with transition rate matrix defined in \eqref{eqn:Q_matrix}. Let $\mathcal{J}$ be a set of eligible jumps, i.e., $\mathcal{J}$ is a set of $(i,j)$ pairs such that $i,j \in \cb{0, \dots, K}$. Assume that $(k,k) \notin \mathcal{J}$ for every $k \in \cb{0,\dots,K}$. Here the pair $(i,j)$ stands for a jump from state $i$ to state $j$. Let $N_{\mathcal{J}}(T, p)$ be the number of eligible jumps in time $[0,T]$. Let $\tilde{Q}$ be a matrix such that $\tilde{Q}_{i,j} = Q_{i,j}$ if $(i,j) \in \mathcal{J}$ and otherwise $\tilde{Q}_{i,j} = 0$.

There is a constant $C(K, p_0, \zeta_0, \mu)$ such that for any $p \in [p_0 - \zeta_0, p_0 + \zeta_0]$,
\begin{equation}
	\abs{\EE{N_{\mathcal{J}}(T, p)}  - T \piv(p)^\top \tilde{Q}(p) \onev } \leq  C(K, p_0, \zeta_0, \mu),
\end{equation}
\begin{equation}
	\abs{\Var{N_{\mathcal{J}}(T, p)}  - T \sigma_{\mathcal{J}}^2(p) } \leq  C(K, p_0, \zeta_0, \mu),
\end{equation}
where
\begin{equation}
	\sigma_{\mathcal{J}}^2(p) = \piv(p)^\top\tilde{Q}(p)\onev - 2\piv(p)^\top\tilde{Q}(p)Q^{\#}(p)\tilde{Q}(p)\onev.
\end{equation}
Here $Q^\#$ is the group inverse of $Q$ defined in \eqref{eqn:Q_group_inverse_def}.
\end{lemm}

\begin{proof}
	We write $\alpha = \alpha_0(p_0, \zeta_0) $, where $\alpha_0$ is defined in \eqref{eqn:alpha_0}.
	Throughout this proof, we omit ``$(p)$'' and ``$(\alpha)$'' for simplicity. We use constants $C_1, C_2, \dots$ that could depend on $K, \mu, p_0, \zeta_0$ but not on $p$.

	We start with a uniformization step. The continuous-time Markov chain we considered can be described by a discrete-time Markov chain with transition matrix $P = I + Q/\alpha$ where jumps occur according to a Poisson process with intensity $\alpha$. Let $M(T)$ be the total number of jumps of the discrete-time Markov chain. Note that among the jumps, there are some self jumps that go from state $k$ to state $k$. We are interested in the eligible jumps in $\mathcal{J}$.

	Note that $M(T) \sim \operatorname{Poisson}(\alpha T)$ and that $\Var{N_{\mathcal{J}}(T)} = \Var{\EE{N_{\mathcal{J}}(T) \mid M(T)}} + \EE{\Var{N_{\mathcal{J}}(T) \mid M(T)}}$. We will then study the distribution of $N_{\mathcal{J}}(T)$ conditional on $M(T) = m$. To this end, let $\tilde{P} = \tilde{Q}/\alpha$. By Lemma \ref{lemma:sum_of_power}, the first moment satisfies
\begin{equation}
	\begin{split}
\EE{N_{\mathcal{J}}(T) \mid M(T) = m}
&= \nuv^\top (I + P + P^2 + \dots P^{m-1})\tilde{P} \onev\\
&= \nuv^\top \p{ m \onev \piv^\top - \alpha Q^\#} \tilde{P} \onev + \errr_1\\
&= m  \piv^\top  \tilde{P} \onev - \alpha  \nuv^\top Q^\# \tilde{P} \onev + \errr_1,
\end{split}
\end{equation}
where $\abs{\errr_1} \leq C_1 C_2^m$ for some constants $C_1 > 0$ and $C_2 \in (0,1)$.
Therefore,
\begin{equation}
\begin{split}
\EE{N_{\mathcal{J}}(T)} &= \EE{\EE{N_{\mathcal{J}}(T) \mid M(T)}}
= \EE{M(T)}  \piv^\top  \tilde{P} \onev + \errr_1'\\
&= \piv^\top  \tilde{Q} \onev T + \errr_1',
\end{split}
\end{equation}
where $\abs{\errr_1'} \leq C_1'$ for some $C_1' > 0$.

The second moment requires further analysis. We write $N_{\mathcal{J}}(T) = \sum_{i = 1}^{M(T)} X_i$, where $X_i = 1$ if the $i$-th jump in the discrete-time Markov chain is an eligible jump. Then
\begin{equation}
	\label{eqn:second_moment_0}
	\begin{split}
	\EE{N_{\mathcal{J}}^2(T) \mid M(T) = m}
	&= \EE{\Big(\sum_{i = 1}^m X_i\Big)^2 \mid M(T) = m}\\
	&= \EE{\sum_{i = 1}^m X_i + 2 \sum_{j=1}^{m-1}\sum_{i = 1}^{m-j}X_i X_{i+j} \mid M(T) = m}\\
	&= \sum_{i = 1}^m \nuv^\top P^{i-1} \tilde{P} \onev + 2 \sum_{j = 1}^{m-1}\sum_{i=1}^{m-j} \nuv^\top P^{i-1} \tilde{P}P^{j-1} \tilde{P} \onev.
	\end{split}
\end{equation}
For the first summation, we can again apply Lemma \ref{lemma:sum_of_power} and get
\begin{equation}
		\label{eqn:first_1}
\begin{split}
\sum_{i = 1}^m \nuv^\top P^{i-1} \tilde{P} \onev &= \nuv^\top\p{ m \onev \piv^\top - \alpha Q^\#} \tilde{P} \onev + \errr_1\\
& = m \piv^\top\tilde{P} \onev -  \alpha \nuv^\top  Q^\# \tilde{P} \onev + \errr_1,
\end{split}
\end{equation}
where $\abs{\errr_1} \leq C_1C_2^m$ for some constants $C_1 > 0$ and $C_2 \in (0,1)$. For the second summation, we have that
\begin{equation}
	\label{eqn:second_1}
	\begin{split}
		&2 \sum_{j = 1}^{m-1}\sum_{i=1}^{m-j} \nuv^\top P^{i-1} \tilde{P}P^{j-1} \tilde{P} \onev
		= 2 \sum_{j = 1}^{m-1}\nuv^\top \p{\sum_{i=1}^{m-j} P^{i-1} } \tilde{P}P^{j-1} \tilde{P} \onev\\
		&\qquad \qquad = 2 \sum_{j = 1}^{m-1}\nuv^\top \p{(m-j) \onev \piv^\top - \alpha Q^\# + E_{j}} \tilde{P}P^{j-1} \tilde{P} \onev,
	\end{split}
\end{equation}
where $\Norm{E_j}_{op} \leq C_1C_2^{m-j}$ for some constants $C_1 > 0$ and $C_2 \in (0,1)$. We further note that
\begin{equation}
		\abs{2 \sum_{j = 1}^{m-1}\nuv^\top E_{j} \tilde{P}P^{j-1} \tilde{P} \onev} \leq \sum_{j = 1}^{m-1} C_3C_2^{m-j} \leq C_4,
\end{equation}
for some constants $C_3, C_4 > 0$. Plugging this back into \eqref{eqn:second_1},
we get that
\begin{equation}
	\label{eqn:second_2}
	\begin{split}
&	2 \sum_{j = 1}^{m-1}\sum_{i=1}^{m-j} \nuv^\top P^{i-1} \tilde{P}P^{j-1} \tilde{P} \onev\\
& \qquad =  2 \sum_{j = 1}^{m-1}\nuv^\top \p{(m-j) \onev \piv^\top - \alpha Q^\# + E_{j}} \tilde{P}P^{j-1} \tilde{P} \onev\\
& \qquad = 2 \piv^\top \tilde{P} \p{\sum_{j = 1}^{m-1} (m-j)P^{j-1}}  \tilde{P} \onev - 2  \nuv^\top \alpha Q^\# \tilde{P} \p{\sum_{j = 1}^{m-1} P^{j-1}}  \tilde{P} \onev + \errr_2,
\end{split}
\end{equation}
where $\abs{\errr_2} \leq C_4$. For the first summation in  \eqref{eqn:second_2}, note that
	\begin{equation}
		\begin{split}
	\sum_{j = 1}^{m-1} (m-j) P^{j-1}
	&= \sum_{i = 1}^{m-1} \sum_{j = 1}^i P^{j-1}
	= \sum_{i = 1}^{m-1} \p{ i \onev \piv^\top - \alpha Q^\# + E_i}\\
	& = \sum_{i = 1}^{m-1} \p{ i \onev \piv^\top - \alpha Q^\#} + E_0\\
	& = \frac{m(m-1)}{2} \onev \piv^\top - (m-1)\alpha Q^\# +E_0,
		\end{split}
	\end{equation}
where, again by Lemma \ref{lemma:sum_of_power}, $\Norm{E_i}_{op} \leq C_1C_2^{i}$ for some constants $C_1 > 0$ and $C_2 \in (0,1)$ and $\Norm{E_0}_{op} \leq C_5$ for some constant $C_5 > 0$.
	This then implies that the first summation in  \eqref{eqn:second_2} satisfies
	\begin{equation}
			\label{eqn:second_3}
		\begin{split}
		 &2  \piv^\top \tilde{P} \p{\sum_{j = 1}^{m-1} (m-j)P^{j-1}}  \tilde{P} \onev\\
		& \qquad = 2  \piv^\top \tilde{P}  \p{ \frac{m(m-1)}{2} \onev \piv^\top - (m-1)\alpha Q^\#} \tilde{P} \onev + \errr_3\\
		& \qquad = m(m-1) \p{ \piv^\top\tilde{P} \onev}^2 -  2\alpha (m-1) \piv^\top \tilde{P}  Q^\#\tilde{P} \onev + \errr_3,
		 	\end{split}
	\end{equation}
	where $\abs{\errr_3} \leq C_6$ for some constant $C_6 > 0$.
For the second summation in \eqref{eqn:second_2}, we can again apply Lemma \ref{lemma:sum_of_power} and get that
	\begin{equation}
			\label{eqn:second_4}
		\begin{split}
			&2  \nuv^\top \alpha Q^\# \tilde{P}\p{\sum_{j = 1}^{m-1} P^{j-1}}  \tilde{P} \onev
			= 2  \nuv^\top \alpha Q^\# \tilde{P}\p{\sum_{j = 1}^{m-1} P^{j-1}}  \tilde{P} \onev\\
			&\qquad \qquad = 2  \nuv^\top \alpha Q^\# \tilde{P} \p{ (m-1) \onev \piv^\top - \alpha Q^\#}  \tilde{P} \onev + \errr_4\\
			&\qquad \qquad = 2\alpha(m-1) \nuv^\top  Q^\# \tilde{P} \onev \p{\piv^\top \tilde{P}\onev}-  2 \alpha^2 \nuv^\top  Q^\#  \tilde{P}Q^\#  \tilde{P} \onev + \errr_4,\\
		\end{split}
	\end{equation}
	where  $\abs{\errr_4} \leq C_7C_2^m$ for some constant $C_7 > 0$ and $C_2 \in (0,1)$.
	Finally, combining \eqref{eqn:second_2}, \eqref{eqn:second_3} and \eqref{eqn:second_4}, we have that the second summation in \eqref{eqn:second_moment_0} satisfies
	\begin{equation}
		\begin{split}
			2 \sum_{j = 1}^{m-1}\sum_{i=1}^{m-j} \nuv^\top P^{i-1} \tilde{P}P^{j-1} \tilde{P} \onev
			& =  m(m-1) \p{ \piv^\top\tilde{P} \onev}^2 -  2\alpha m \piv^\top \tilde{P}  Q^\#\tilde{P} \onev\\
			 & \qquad \qquad -  2 \alpha m \nuv^\top Q^\# \tilde{P} \onev \p{\piv^\top \tilde{P}\onev} + \errr_5,
		\end{split}
	\end{equation}
		where $\abs{\errr_5} \leq C_8$ for some constant $C_8 > 0$.
		Now, together with \eqref{eqn:first_1} and \eqref{eqn:second_moment_0}, the above implies that
	\begin{equation}
		\begin{split}
& \EE{N_{\mathcal{J}}^2(T) \mid M(T) = m} \\
&\qquad \qquad = m \piv^\top\tilde{P} \onev + m(m-1) \p{ \piv^\top\tilde{P} \onev}^2 -  2\alpha m \piv^\top \tilde{P}  Q^\#\tilde{P} \onev\\
& \qquad \qquad \qquad \qquad - 2 \alpha m \nuv^\top Q^\# \tilde{P} \onev \p{\piv^\top \tilde{P}\onev} + \errr_6,
		\end{split}
	\end{equation}
where $\abs{\errr_6} \leq C_9$ for some constant $C_9 > 0$.

We are ready to study the variance of $N_{\mathcal{J}}(T)$ conditional on $M(T)$.
\begin{equation}
	\begin{split}
	&	\Var{N_{\mathcal{J}}(T) \mid M(T) = m}\\
& \qquad 	= \EE{N_{\mathcal{J}}^2(T) \mid M(T) = m} - \p{\EE{N_{\mathcal{J}}(T) \mid M(T) = m}}^2\\
& \qquad	=  m \piv^\top\tilde{P} \onev + m(m-1) \p{ \piv^\top\tilde{P} \onev}^2 -  2\alpha m \piv^\top \tilde{P}  Q^\#\tilde{P} \onev \\
& \qquad  \qquad \qquad - 2m\alpha \nuv^\top Q^\# \tilde{P} \onev \p{\piv^\top \tilde{P}\onev}  - \p{m  \piv^\top  \tilde{P} \onev - \alpha  \nuv^\top Q^\# \tilde{P} \onev}^2 + \errr_7\\
& \qquad 	= m\p{ - \p{ \piv^\top\tilde{P} \onev}^2 +  \p{ \piv^\top\tilde{P} \onev} - 2\alpha \piv^\top \tilde{P}  Q^\#\tilde{P} \onev} + \errr_8,
	\end{split}
\end{equation}
	where $\abs{\errr_7} \leq C_{10}$ and $\abs{\errr_8} \leq C_{10}$  for some constant $C_{10} > 0$.

Finally, note that
\begin{equation}
	\begin{split}
\Var{N_{\mathcal{J}}(T)} &= \Var{\EE{N_{\mathcal{J}}(T) \mid M(T)}} + \EE{\Var{N_{\mathcal{J}}(T) \mid M(T)}}\\
& = \p{ \piv^\top  \tilde{P} \onev }^2 \Var{M(T)} + \errr_9 \\
& \qquad  \qquad+ \p{ - \p{ \piv^\top\tilde{P} \onev}^2 +  \p{ \piv^\top\tilde{P} \onev} - 2\alpha \piv^\top \tilde{P}  Q^\#\tilde{P} \onev} \EE{M(T)} \\
& = \p{ \piv^\top  \tilde{P} \onev }^2 \alpha T + \errr_9 \\
& \qquad  \qquad + \p{ - \p{ \piv^\top\tilde{P} \onev}^2 +  \p{ \piv^\top\tilde{P} \onev} - 2\alpha \piv^\top \tilde{P}  Q^\#\tilde{P} \onev} \alpha T \\
& =  \p{  \piv^\top\tilde{P} \onev - 2\alpha \piv^\top \tilde{P}  Q^\#\tilde{P} \onev} \alpha T + \errr_9, \\
& =  \p{  \piv^\top\tilde{Q} \onev - 2 \piv^\top \tilde{Q}  Q^\#\tilde{Q} \onev} T + \errr_9, \\
	\end{split}
\end{equation}
	where $\abs{\errr_9} \leq C_{11}$ for some constant $C_{11} > 0$.
\end{proof}

\begin{lemm}[Number of departures]
\label{lemma:num_arrival}
Consider the continuous-time Markov chain with transition rate matrix defined in \eqref{eqn:Q_matrix}. Treating it as a queue length process, let $N_{\dep}(T, p)$ be the number of departures in time $[0,T]$. There is a constant $C(K, p_0, \zeta_0, \mu)$ such that for any $p \in [p_0 - \zeta_0, p_0 + \zeta_0]$,
\begin{equation}
	\abs{\EE{N_{\dep}(T, p)}  - T \mu(1-\pi_0(p)) } \leq  C(K, p_0, \zeta_0, \mu),
\end{equation}
\begin{equation}
	\abs{\Var{N_{\dep}(T, p)}  - T \sigma_{\dep}^2(p) } \leq  C(K, p_0, \zeta_0, \mu),
\end{equation}
where
\begin{equation}
	\sigma_{\dep}^2(p) = \mu(1 - \pi_0(p)) + 2\mu^2 \sum_{i = 0}^{K-1} \pi_{i+1}(p) Q_{i0}^\#(p).
\end{equation}
Here $Q^\#$ is the group inverse of $Q$ defined in \eqref{eqn:Q_group_inverse_def}.
\end{lemm}

\begin{proof}
Throughout this proof, we omit ``$(p)$'' and ``$(\alpha)$'' for simplicity. We use constants $C_1, C_2, \dots$ that could depend on $K, \mu, p_0, \zeta_0$ but not on $p$.
We will use Lemma \ref{lemma:num_eligible_jumps} to show the results. We note that the eligible jumps are $\mathcal{J} = \cb{(i,i-1): i \in \cb{1,\dots,K}}$. Therefore, in this setting,
	\begin{equation}
 \tilde{Q} =  \begin{bmatrix}
			0 &0 & \dots & & \\
			\mu & 0 &  & &\\
			& \mu & 0 &  &\\
			& & &\dots &\\
			& & & 	\mu & 0
		\end{bmatrix}.
	\end{equation}

We start with the expectation. Note that $\tilde{Q} \onev = (0, \mu, \dots, \mu)^\top$. Therefore,
\begin{equation}
    \piv^\top \tilde{Q} \onev = \mu \p{\sum_{i = 1}^K \pi_i} = \mu(1 - \pi_0).
\end{equation}
For the variance, we note that since $Q^\# \onev = \zerov$,
\begin{equation}
Q^\#\tilde{Q} \onev = \mu Q^\#\begin{bmatrix}
	0\\
	1\\
	\dots\\
	1
\end{bmatrix}
= \mu Q^\# \onev -   \mu Q^\#\begin{bmatrix}
	1\\
	0\\
	\dots\\
	0
\end{bmatrix}
= -\mu\begin{bmatrix}
	Q_{0,0}^\#\\
	Q_{1,0}^\#\\
	\dots\\
	Q_{K,0}^\#
\end{bmatrix}.
\end{equation}
Note also that $\piv^\top\tilde{Q} = \mu (\pi_1, \dots, \pi_K, 0)$. Therefore, $\piv^\top \tilde{Q}  Q^\#\tilde{Q} \onev = -\mu^2\sum_{i = 0}^{K-1} \pi_{i+1}Q_{i,0}^\#$. Hence,
\begin{equation}
\piv^\top\tilde{Q} \onev - 2 \piv^\top \tilde{Q}  Q^\#\tilde{Q} \onev = (1 - \pi_0)\mu + 2 \mu^2\sum_{i = 0}^{K-1} \pi_{i+1}Q_{i,0}^\#.
\end{equation}

\end{proof}

\begin{lemm}[Uniform integrability]
\label{lemma:uniform_intergral}
Set the price at $p$ throughout the experiment. Let $N_{\arr}(T, p)$ be the total number of arrivals in time $[0,T]$. Under Assumptions \ref{assu:queue_length}-\ref{assu:bound_smooth}, there is a constant $C(K, p) > 0$ such that, for all $T\geq1$,
\begin{equation}
	\EE{\p{N_{\arr}(T, p) - \mu(1-\pi_0(p)) T}^4} \leq  C(K, p) T^2.
\end{equation}
\end{lemm}

\begin{proof}
We prove the lemma in a more general case. In particular, we will show that under the conditions of Lemma \ref{lemma:num_eligible_jumps}, there is a constant $C(K, p_0, \zeta_0, \mu)$ such that for any $p \in [p_0 - \zeta_0, p_0 + \zeta_0]$,
\begin{equation}
\label{eqn:uniform_intergral_new_goal}
\EE{\p{N_{\mathcal{J}}(T, p) -  \piv(p)^\top \tilde{Q}(p) \onev T}^4} \leq  C(K, p_0, \zeta_0, \mu) T^2.
\end{equation}
Recall that in Lemma \ref{lemma:num_eligible_jumps}, $\mathcal{J}$ is a set of eligible jumps, i.e., $\mathcal{J}$ is a set of $(i,j)$ pairs such that $i,j \in \cb{0, \dots, K}$, and $N_{\mathcal{J}}(T, p)$ is the number of eligible jumps in time $[0,T]$. We require that $(k,k) \notin \mathcal{J}$ for every $k \in \cb{0,\dots,K}$. We have also defined $\tilde{Q}$ to be a matrix such that $\tilde{Q}_{i,j} = Q_{i,j}$ if $(i,j) \in \mathcal{J}$ and otherwise $\tilde{Q}_{i,j} = 0$. In the context of Lemma \ref{lemma:uniform_intergral}, we take $\mathcal{J} = \cb{(i,i+1): i \in \cb{0,\dots,K-1}}$.

From this point on, we will work under the conditions of Lemma \ref{lemma:num_eligible_jumps} and aim at proving \eqref{eqn:uniform_intergral_new_goal}.

Similar to the proof of Lemma \ref{lemma:num_eligible_jumps}, we make some notation simplification. We write $\alpha = \alpha_0(p_0, \zeta_0) $, where $\alpha_0$ is defined in \eqref{eqn:alpha_0}. Throughout this proof, we omit ``$(p)$'' and ``$(\alpha)$'' for simplicity. We use constants $C_1, C_2, \dots$ that could depend on $K, \mu, p_0, \zeta_0$ but not on $p$.

Again, same as in the proof of Lemma \ref{lemma:num_eligible_jumps}, we start with a uniformization step. The continuous-time Markov chain we considered can be described by a discrete-time Markov chain with transition matrix $P = I + Q/\alpha$ where jumps occur according to a Poisson process with intensity $\alpha$. Let $M(T)$ be the total number of jumps of the discrete-time Markov chain. Note that among the jumps, there are some self jumps that go from state $k$ to state $k$. We are interested in the eligible jumps in $\mathcal{J}$.

Note that $M(T) \sim \operatorname{Poisson}(\alpha T)$. We will then study the distribution of $N_{\mathcal{J}}(T)$ conditional on $M(T) = m$. We define $\tilde{P} = \tilde{Q}/\alpha$. From now on, unless specified otherwise, we condition on $M(T) = m$ and focus on the discrete-time Markov chain.
We write $N_{\mathcal{J}}(T) = \sum_{i = 1}^{M(T)} X_i$, where $X_i = 1$ if the $i$-th jump in the discrete-time Markov chain is an eligible jump. We are interested in the first, second, third and fourth moments of $N_{\mathcal{J}}(T)$. Expanding the parentheses, we get that
\begin{equation}
\begin{split}
N_{\mathcal{J}}(T) &= \sum_{i = 1}^{M(T)} X_i,\\
N_{\mathcal{J}}^2(T) &= \sum_{i = 1}^{M(T)} X_i + 2 \sum_{1\leq i < j\leq M(T)} X_i X_j,\\
N_{\mathcal{J}}^3(T) &= \sum_{i = 1}^{M(T)} X_i + 6 \sum_{1\leq i < j\leq M(T)} X_i X_j + 6 \sum_{1\leq i < j < k \leq M(T)} X_i X_j X_k,\\
N_{\mathcal{J}}^4(T) &= \sum_{i = 1}^{M(T)} X_i + 14 \sum_{1\leq i < j\leq M(T)} X_i X_j + 36 \sum_{1\leq i < j < k \leq M(T)} X_i X_j X_k\\
& \qquad\qquad\qquad\qquad\qquad\qquad\qquad\qquad + 24 \sum_{1\leq i < j < k <l \leq M(T)} X_i X_j X_k X_l.
\end{split}
\end{equation}
This then implies that
\begin{equation}
\label{eqn:Nj_conditional_moments}
\begin{split}
\EE{N_{\mathcal{J}}(T) \mid M(T) = m} & = A_1,\\
\EE{N_{\mathcal{J}}^2(T) \mid M(T) = m}&= A_1 + 2A_2,\\
\EE{N_{\mathcal{J}}^3(T) \mid M(T) = m} &= A_1 + 6A_2 + 6A_3,\\
\EE{N_{\mathcal{J}}^4(T) \mid M(T) = m} &= A_1 + 14A_2 + 36A_3 + 24A_4,
\end{split}
\end{equation}
where
\begin{equation}
\begin{split}
    A_1 &= \sum_{i = 1}^{m} \EE{X_i}, \qquad A_2 = \sum_{1\leq i < j\leq m} \EE{X_i X_j}, \\
    A_3 &= \sum_{1\leq i < j < k \leq m} \EE{X_i X_j X_k}, \qquad A_4 = \sum_{1\leq i < j < k<l \leq m} \EE{X_i X_j X_k X_l}.
\end{split}
\end{equation}

We will then study $A_1, A_2, \dots A_4$. We have shown in the proof of Lemma \ref{lemma:num_eligible_jumps} that $A_1 = \EE{N_{\mathcal{J}}(T) \mid M(T) = m}
= m \piv^\top \tilde{P} \onev - \alpha \nuv^\top Q^{\#} \tilde{P} \onev + \errr$, where $\abs{\errr} \leq C_0 C^m$ for some constants $C_0 > 0$ and $C \in \p{0,1}$.
For $A_2$, note that for $i < j$, $X_i X_j$ is the indicator that the $i$-th jump is eligible and the $j$-th jump is eligible. Therefore, $\EE{X_i X_j} = \nuv^\top P^{i-1}\tilde{P} P^{j-i-1}\tilde{P} \onev$. Hence, we can write
\begin{equation}
\label{eqn:A_2_formula}
    A_2 = \sum_{1\leq i < j\leq m} \EE{X_i X_j} = \sum_{1\leq i < j\leq m} \nuv^\top P^{i-1}\tilde{P} P^{j-i-1}\tilde{P} \onev
    = \nuv^\top \sum_{i+j \leq m}  P^{i-1}\tilde{P} P^{j-1}\tilde{P} \onev,
\end{equation}
where the last equality follows from a change of variable. Similarly, we find
\begin{equation}
\label{eqn:A_3_formula}
\begin{split}
    A_3 =& \sum_{1\leq i < j<k\leq m} \EE{X_i X_j X_k} = \sum_{1\leq i < j<k\leq m} \nuv^\top P^{i-1}\tilde{P} P^{j-i-1}\tilde{P} P^{k-j-1}\tilde{P} \onev\\
    &=  \nuv^\top \sum_{i+j+k \leq m} P^{i-1}\tilde{P} P^{j-1}\tilde{P} P^{k-1}\tilde{P}\onev,
\end{split}
\end{equation}
and
\begin{equation}
\label{eqn:A_4_formula}
\begin{split}
    A_4 =& \sum_{1\leq i < j<k<l\leq m} \EE{X_i X_j X_k X_l}\\
    &= \sum_{1\leq i < j<k <l\leq m} \nuv^\top P^{i-1}\tilde{P} P^{j-i-1}\tilde{P} P^{k-j-1}\tilde{P} P^{l-k-1}\tilde{P}\onev\\
    &=  \nuv^\top \sum_{i+j+k+l \leq m} P^{i-1}\tilde{P} P^{j-1}\tilde{P} P^{k-1}\tilde{P}P^{l-1}\tilde{P}\onev.
\end{split}
\end{equation}

Now we focus on $A_3$. The results for $A_2$ and $A_4$ are very similar. The key term in $A_3$ is $\sum_{i+j+k \leq m} P^{i-1}\tilde{P} P^{j-1}\tilde{P} P^{k-1}\tilde{P}$, which is a sum of the product of terms like $P^{i-1}\tilde{P}$. Let $R_i = P^i - \onev \piv^\top$. Then we can decompose each $P^{i-1}$ into two terms $R_{i-1}$ and $\onev \piv^\top$, and thus
\begin{equation}
\label{eqn:A3_decomposition}
\begin{split}
&P^{i-1}\tilde{P} P^{j-1}\tilde{P} P^{k-1}\tilde{P}\\
&\quad = \p{R_{i-1} + \onev \piv^\top}\tilde{P} \p{R_{j-1} + \onev \piv^\top} \tilde{P}\p{R_{k-1} + \onev \piv^\top}\tilde{P}\\
&\quad = \p{\onev \piv^\top \tilde{P}}^3 + R_{i-1}\tilde{P} \p{\onev \piv^\top \tilde{P}}^2 + \p{\onev \piv^\top \tilde{P}} R_{j-1}\tilde{P} \p{\onev \piv^\top \tilde{P}} + \p{\onev \piv^\top \tilde{P}}^2 R_{k-1}\tilde{P}\\
&\qquad + R_{i-1}\tilde{P}R_{j-1}\tilde{P}\p{\onev \piv^\top \tilde{P}}
+ R_{i-1}\tilde{P}\p{\onev \piv^\top \tilde{P}}R_{k-1}\tilde{P}
+ \p{\onev \piv^\top \tilde{P}} R_{j-1}\tilde{P}R_{k-1}\tilde{P}\\
&\qquad + R_{i-1}\tilde{P}R_{j-1}\tilde{P}R_{k-1}\tilde{P}.
\end{split}
\end{equation}
By Lemma \ref{lemma:operator_norm}, there are constants $C_0>0$ and $C \in (0,1)$ such that $\Norm{R_{i-1}}_{op} \leq C_0C^{i-1}$. Therefore,
\begin{equation}
\begin{split}
 &\Norm{ \sum_{i+j+k \leq m} R_{i-1}\tilde{P}R_{j-1}\tilde{P}R_{k-1}\tilde{P}}_{op}\\
 &\qquad \leq  \sum_{i+j+k \leq m} \Norm{R_{i-1}}_{op}\Norm{\tilde{P}}_{op}\Norm{R_{j-1}}_{op}\Norm{\tilde{P}}_{op}\Norm{R_{k-1}}_{op}\Norm{\tilde{P}}_{op}\\
 &\qquad \leq C_0^3\sum_{i+j+k \leq m} C^{i+j+k-3} \leq C_1,
 \end{split}
\end{equation}
for some constant $C_1 > 0$.
Now we study terms involving two factors of the form $R_{i-1}$.
\begin{equation}
\begin{split}
 &\Norm{ \sum_{i+j+k \leq m} R_{i-1}\tilde{P}R_{j-1}\tilde{P}\p{\onev \piv^\top \tilde{P}}}_{op}\\
 &\qquad \leq  \sum_{i+j+k \leq m} \Norm{R_{i-1}}_{op}\Norm{\tilde{P}}_{op}\Norm{R_{j-1}}_{op}\Norm{\tilde{P}}_{op}\Norm{\onev \piv^\top}_{op}\Norm{\tilde{P}}_{op}\\
 &\qquad \leq C_0^2\sum_{i+j+k \leq m} C^{i+j-2}
 \leq C_0^2\sum_{i+j \leq m} C^{i+j-2} m \leq C_2 m,
 \end{split}
\end{equation}
for some constant $C_2 > 0$. The same inequalities hold for $R_{i-1}\tilde{P}\p{\onev \piv^\top \tilde{P}}R_{k-1}\tilde{P}$ and $\p{\onev \piv^\top \tilde{P}} R_{j-1}\tilde{P}R_{k-1}\tilde{P}$.

Then, we look at terms involving only one factor of the form $R_{i-1}$. Note that
\begin{equation}
\label{eqn:one_R_result}
\begin{split}
    \sum_{i+j+k \leq m} R_{i-1}
    &= \sum_{i = 1}^{m-2} \binom{m-i}{2} R_{i-1}
    = \sum_{i=1}^{m-2} \frac{(m-i)(m-i-1)}{2}R_{i-1}\\
    &= \frac{m^2}{2}\sum_{i=1}^{m-2} R_{i-1} + \sum_{i = 1}^{m-2} \frac{-(2i+1)m + i(i+1)}{2} R_{i-1}.
\end{split}
\end{equation}
For the first term, by Lemma~\ref{lemma:sum_of_power}, we have that $\sum_{i=1}^{m-2} R_{i-1} = -\alpha Q^\# + E$, where $\Norm{E}_{op} \leq C_1 C^m$ for some constants $C_1 > 0$ and $C \in (0,1)$. For the second term, Lemma~\ref{lemma:operator_norm} gives $\Norm{R_{i-1}}_{op} \leq C_0 C^{i-1}$ after adjusting the constants if necessary. Therefore,
\begin{equation}
    \Norm{\sum_{i = 1}^{m-2} \frac{-(2i+1)m + i(i+1)}{2} R_{i-1}}_{op}
    \leq C_0\sum_{i = 1}^{m-2} \abs{\frac{-(2i+1)m + i(i+1)}{2}} C^{i-1}
    \leq C_1 m,
\end{equation}
for some constant $C_1 > 0$. Combining the results and plugging back into \eqref{eqn:one_R_result}, we have that
\begin{equation}
     \sum_{i+j+k \leq m} R_{i-1} = -\frac{m^2}{2}\alpha Q^\# + E_1,
\end{equation}
where $\Norm{E_1}_{op} \leq C_1 m$ for some constant $C_1 > 0$. By symmetry, the same result holds for $\sum_{i+j+k \leq m} R_{j-1}$ and $\sum_{i+j+k \leq m} R_{k-1}$.

We have analyzed all terms in \eqref{eqn:A3_decomposition}. Combining the results, we have
\begin{equation}
\begin{split}
&\sum_{i+j+k \leq m} P^{i-1}\tilde{P} P^{j-1}\tilde{P} P^{k-1}\tilde{P}\\
&\qquad= \sum_{i+j+k \leq m}\p{\onev \piv^\top \tilde{P}}^3 + E_3 \\
& \qquad\qquad - \frac{\alpha m^2}{2} \sqb{ Q^\#\tilde{P} \p{\onev \piv^\top \tilde{P}}^2 +
 \p{\onev \piv^\top \tilde{P}} Q^\#\tilde{P} \p{\onev \piv^\top \tilde{P}}
+ \p{\onev \piv^\top \tilde{P}}^2 Q^\#\tilde{P}}\\
&\qquad= \binom{m}{3}\p{\onev \piv^\top \tilde{P}}^3 + E_3 \\
&\qquad \qquad - \frac{\alpha m^2}{2} \sqb{ Q^\#\tilde{P} \p{\onev \piv^\top \tilde{P}}^2 +
 \p{\onev \piv^\top \tilde{P}} Q^\#\tilde{P} \p{\onev \piv^\top \tilde{P}}
+ \p{\onev \piv^\top \tilde{P}}^2 Q^\#\tilde{P}}
\end{split}
\end{equation}
where $\Norm{E_3}_{op} \leq C_4 m$ for some constant $C_4 > 0$.

We can conduct the same analysis for $A_2$ and $A_4$ and get
\begin{equation}
\begin{split}
&  \sum_{i+j \leq m} P^{i-1}\tilde{P} P^{j-1}\tilde{P} \\
&\qquad = \binom{m}{2} \p{\onev \piv^\top \tilde{P}}^2  - \alpha m \sqb{ Q^\#\tilde{P} \p{\onev \piv^\top \tilde{P}}
+ \p{\onev \piv^\top \tilde{P}} Q^\#\tilde{P}} + E_2 ,
\end{split}
\end{equation}
where $\Norm{E_2}_{op} \leq C_2$ for some constant $C_2 > 0$, and
\begin{equation}
\begin{split}
&\sum_{i+j+k+l \leq m} P^{i-1}\tilde{P} P^{j-1}\tilde{P} P^{k-1}\tilde{P} P^{l-1}\tilde{P}\\
&\qquad = \binom{m}{4} \p{\onev \piv^\top \tilde{P}}^4
 - \frac{\alpha m^3}{6} \Big[ Q^\#\tilde{P} \p{\onev \piv^\top \tilde{P}}^3
+ \p{\onev \piv^\top \tilde{P}} Q^\#\tilde{P} \p{\onev \piv^\top \tilde{P}}^2\\
&\qquad \qquad \qquad \qquad+ \p{\onev \piv^\top \tilde{P}}^2 Q^\#\tilde{P}\p{\onev \piv^\top \tilde{P}}
+ \p{\onev \piv^\top \tilde{P}}^3Q^\#\tilde{P} \Big] + E_4,
\end{split}
\end{equation}
where $\Norm{E_4}_{op} \leq C_4 m^2$ for some constant $C_4 > 0$.
Here, details of the derivation have been omitted for brevity.

Therefore, by \eqref{eqn:A_2_formula} - \eqref{eqn:A_4_formula}, we have the following formulae for $A_1, \dots, A_4$.
\begin{equation}
\begin{split}
A_1 & = m \p{\piv^\top \tilde{P} \onev} - \alpha \nuv^\top Q^\# \tilde{P} \onev  + \errr_1,\\
A_2 & = \frac{m^2}{2} \p{\piv^\top \tilde{P} \onev}^2 - m \Bigg[\alpha \p{\piv^\top \tilde{P} Q^\# \tilde{P} \onev}\\
& \qquad \qquad  \qquad \qquad \qquad + \alpha \p{\nuv^\top Q^\# \tilde{P} \onev}\p{\piv^\top \tilde{P} \onev} + \frac{1}{2} \p{ \piv^\top \tilde{P}\onev}^2 \Bigg] + \errr_2,\\
A_3 & = \frac{m^3}{6} \p{\piv^\top \tilde{P} \onev}^3 - m^2 \Bigg[\alpha \p{\piv^\top \tilde{P} \onev}\p{\piv^\top \tilde{P} Q^\# \tilde{P} \onev} \\
 &\qquad \qquad  \qquad \qquad \qquad + \frac{1}{2}\alpha \p{\nuv^\top Q^\# \tilde{P} \onev}\p{\piv^\top \tilde{P} \onev}^2 + \frac{1}{2} \p{ \piv^\top \tilde{P}\onev}^3\Bigg] + \errr_3,\\
 A_4 & = \frac{m^4}{24} \p{\piv^\top \tilde{P} \onev}^4 - m^3 \Bigg[\frac{1}{2}\alpha \p{\piv^\top \tilde{P} \onev}^2\p{\piv^\top \tilde{P} Q^\# \tilde{P} \onev} \\
 &\qquad \qquad  \qquad \qquad \qquad + \frac{1}{6}\alpha \p{\nuv^\top Q^\# \tilde{P} \onev}\p{\piv^\top \tilde{P} \onev}^3 + \frac{1}{4} \p{ \piv^\top \tilde{P}\onev}^4\Bigg] + \errr_4,
\end{split}
\end{equation}
where $\abs{\errr_1} \leq C_1 C^m$, $\abs{\errr_2} \leq C_2$, $\abs{\errr_3} \leq C_3m$, $\abs{\errr_4} \leq C_4m^2$, for some constants $C \in (0,1), C_1,C_2, C_3, C_4 > 0$. To simplify notation, let
\begin{equation}
\gamma =  \p{\piv^\top \tilde{P} \onev}, \qquad
\theta = \p{\piv^\top \tilde{P} Q^\# \tilde{P} \onev}, \qquad
\delta = \p{\nuv^\top Q^\# \tilde{P} \onev}.
\end{equation}
Then we can express $A_1, \dots A_4$ as
\begin{equation}
\begin{split}
A_1 &= m\gamma - \alpha\delta + \errr_1,\\
A_2 &= \frac{m^2}{2} \gamma^2 - m\p{\alpha \theta + \alpha \gamma \delta + \frac{1}{2}\gamma^2} + \errr_2,\\
A_3 &= \frac{m^3}{6} \gamma^3  - m^2 \p{\alpha \gamma\theta + \frac{1}{2}\alpha \gamma^2 \delta + \frac{\gamma^3}{2}} + \errr_3,\\
A_4 &= \frac{m^4}{24} \gamma^4  - m^3 \p{\frac{1}{2}\alpha \gamma^2 \theta + \frac{1}{6}\alpha \gamma^3 \delta + \frac{\gamma^4}{4}} + \errr_4,\\
\end{split}
\end{equation}
where $\abs{\errr_1} \leq C_1 C^m$, $\abs{\errr_2} \leq C_2$, $\abs{\errr_3} \leq C_3m$, $\abs{\errr_4} \leq C_4m^2$, for some constants $C \in (0,1), C_1,C_2, C_3, C_4 > 0$.

Plugging the above into \eqref{eqn:Nj_conditional_moments}, we have that
\begin{equation}
\begin{split}
\EE{N_{\mathcal{J}} \mid M(T) = m} &= m\gamma - \alpha\delta + \errr_1,\\
\EE{N_{\mathcal{J}}^2 \mid M(T) = m} &= m^2 \gamma^2 - m\p{2\alpha \theta + 2\alpha \gamma \delta + \gamma^2 - \gamma} + \errr_2,\\
\EE{N_{\mathcal{J}}^3 \mid M(T) = m} &= m^3 \gamma^3  - m^2 \p{6\alpha \gamma\theta + 3\alpha \gamma^2 \delta + 3\gamma^3 - 3\gamma^2} + \errr_3,\\
\EE{N_{\mathcal{J}}^4 \mid M(T) = m} &= m^4 \gamma^4  - m^3 \p{12\alpha \gamma^2 \theta + 4\alpha \gamma^3 \delta + 6\gamma^4 - 6\gamma^3} + \errr_4,\\
\end{split}
\end{equation}
where $\abs{\errr_1} \leq C_1 C^m$, $\abs{\errr_2} \leq C_2$, $\abs{\errr_3} \leq C_3m$, $\abs{\errr_4} \leq C_4m^2$, for some constants $C \in (0,1), C_1,C_2, C_3, C_4 > 0$. Then note that since $M(T) \sim \operatorname{Poisson}(\alpha T)$, we have that
\begin{equation}
\begin{split}
\EE{M(T)} &= \alpha T, \\
\EE{M(T)^2} &= \p{\alpha T}^2 + \alpha T,\\
\EE{M(T)^3} &=  \p{\alpha T}^3+ 3\p{\alpha T}^2 + \alpha T, \\
\EE{M(T)^4}& =  \p{\alpha T}^4 + 6\p{\alpha T}^3+ 7\p{\alpha T}^2 + \alpha T.
\end{split}
\end{equation}
Therefore,
\begin{equation}
\begin{split}
\EE{N_{\mathcal{J}} } &= \p{\alpha T}\gamma - \alpha\delta + \errr_1,\\
\EE{N_{\mathcal{J}}^2} &= \p{\alpha T}^2 \gamma^2 - \p{\alpha T}\p{2\alpha \theta + 2\alpha \gamma \delta  - \gamma} + \errr_2,\\
\EE{N_{\mathcal{J}}^3} &= \p{\alpha T}^3 \gamma^3  - \p{\alpha T}^2 \p{6\alpha \gamma\theta + 3\alpha \gamma^2 \delta  - 3\gamma^2} + \errr_3,\\
\EE{N_{\mathcal{J}}^4} &= \p{\alpha T}^4 \gamma^4  - \p{\alpha T}^3 \p{12\alpha \gamma^2 \theta + 4\alpha \gamma^3 \delta  - 6\gamma^3} + \errr_4,\\
\end{split}
\end{equation}
where $\abs{\errr_1} \leq C_1 C^T$, $\abs{\errr_2} \leq C_2$, $\abs{\errr_3} \leq C_3 T$, $\abs{\errr_4} \leq C_4 T^2$, for some constants $C \in (0,1), C_1,C_2, C_3, C_4 > 0$.
Finally, we note that
\begin{equation}
\begin{split}
\EE{\p{N_{\mathcal{J}} - \EE{N_{\mathcal{J}}}}^4}
&= \EE{N_{\mathcal{J}}^4} - 4 \EE{N_{\mathcal{J}}^3} \EE{N_{\mathcal{J}}} + 6\EE{N_{\mathcal{J}}^2}\p{\EE{N_{\mathcal{J}}}}^2 - 3\EE{N_{\mathcal{J}}}^4\\
& = \Cancel[blue]{\alpha^4 T^4 \gamma^4} + \alpha^3 T^3\p{-\Cancel[red]{12\alpha \gamma^2 \theta} - \Cancel[cyan]{4 \alpha \gamma^3 \delta} + \Cancel[green]{6 \gamma^3}}\\
& \qquad  -\Cancel[blue]{ 4 \alpha^4T^4 \gamma^4} + \alpha^3 T^3\p{\Cancel[cyan]{4 \alpha \gamma^3 \delta} + \Cancel[red]{24\alpha \gamma^2 \theta} + \Cancel[cyan]{12 \alpha \gamma^3 \delta} - \Cancel[green]{12\gamma^3}}\\
& \qquad  + \Cancel[blue]{6\alpha^4 T^4\gamma^4} + \alpha^3 T^3\p{-\Cancel[cyan]{12\alpha \gamma^3 \delta} - \Cancel[red]{12\alpha \gamma^2 \theta} - \Cancel[cyan]{12\alpha \gamma^3 \delta} + \Cancel[green]{6\gamma^3}}\\
& \qquad  - \Cancel[blue]{3 \alpha^4T^4 \gamma^4} + \alpha^3T^3(\Cancel[cyan]{12\alpha \gamma^3 \delta}) + \errr\\
& = \errr,
\end{split}
\end{equation}
where $\abs{\errr} \leq C T^2$ for some constant $C > 0$.
\end{proof}

\begin{lemm}[Triangular-array Anscombe theorem]
\label{lemma:Lyap_anscombe}
For each $\zeta$, let $X_1(\zeta), X_2(\zeta), \dots$ be a sequence of i.i.d. random variables. Let $\mu(\zeta) = \EE{X_1(\zeta)}$, $\sigma^2(\zeta) = \Var{X_1(\zeta)}$, and $m(\zeta) = \EE{\p{X_1(\zeta) - \mu(\zeta)}^4}$. Assume that $\mu$, $\sigma^2$ and $m$ are continuous functions of $\zeta$ and that $m(\zeta) \leq C^2$ uniformly over $\zeta$ for some constant $C > 0$. Let $N_n$ be a sequence of random variables such that $N_n/n \stackrel{p}{\to} 1$. Let $\zeta_n$ be a sequence of numbers such that $\zeta_n \to 0$. Let
\begin{equation}
    S_n(\zeta) = X_1(\zeta) + X_2(\zeta) + \dots + X_n(\zeta).
\end{equation}
Then as $n \to \infty$,
\begin{equation}
    \sqrt{N_n}\p{\frac{S_{N_n}(\zeta_n)}{N_n} - \mu(\zeta_n)}  \Rightarrow \mathcal{N}(0, \sigma^2(0)).
\end{equation}
\end{lemm}

\begin{proof}
For notational simplicity, let $Y_i(\zeta) = X_i(\zeta) - \mu(\zeta)$ and
    $\tilde{S}_n(\zeta) = Y_1(\zeta) + Y_2(\zeta) + \dots + Y_n(\zeta)$.
We start by showing convergence of $S_{n}(\zeta_n)$.
This can be shown easily with the Lindeberg-Feller theorem \citep{durrett2019probability}. In particular, we have that  $\sum_{i = 1}^n \EE{Y_i(\zeta_n)^2}/n = \sigma^2(\zeta_n) \to \sigma^2(0)$, and that
\begin{equation}
    \sum_{i = 1}^n \EE{\frac{Y_i(\zeta_n)^2}{n}; \frac{\abs{Y_i(\zeta_n)}}{\sqrt n} > \epsilon}
    \leq \frac{1}{n^2\epsilon^2}\sum_{i = 1}^n \EE{Y_i(\zeta_n)^4} \to 0.
\end{equation}
Therefore, the Lindeberg conditions are satisfied and thus
\begin{equation}
    \sqrt{n}\p{\frac{S_{n}(\zeta_n)}{n} - \mu(\zeta_n)}  \Rightarrow \mathcal{N}(0, \sigma^2(0)).
\end{equation}

We will then show that $S_{n}(\zeta_n) - n\mu(\zeta_n)$ is close to $S_{N_n}(\zeta_n) - N_n\mu(\zeta_n)$, i.e., $\tilde{S}_{n}(\zeta_n)$ is close to $\tilde{S}_{N_n}(\zeta_n)$. For any $\epsilon, \delta > 0$, by Kolmogorov's inequality \citep{durrett2019probability},
\begin{equation}
\begin{split}
    \PP{\abs{\tilde{S}_{N_n}(\zeta_n) - \tilde{S}_{n}(\zeta_n) } > \epsilon \sqrt{n}}
    &\leq \PP{\abs{N_n - n} > \delta n} + 2\PP{\max_{ 1\leq m \leq \lceil \delta n \rceil}\abs{\tilde{S}_{m}(\zeta_n)} \geq \epsilon\sqrt{n}}\\
    & \leq \PP{\abs{N_n - n} > \delta n} + 2\Var{\tilde{S}_{\lceil \delta n \rceil }(\zeta_n)}/(n\epsilon^2)\\
    & \leq \PP{\abs{N_n - n} > \delta n} + 2C \lceil \delta n\rceil/(n\epsilon^2).
\end{split}
\end{equation}
Since $N_n/n \stackrel{p}{\to} 1$, we have that
\begin{equation}
    \limsup_{n \to \infty} \PP{\abs{\tilde{S}_{N_n}(\zeta_n) - \tilde{S}_{n}(\zeta_n) } > \epsilon \sqrt{n}} \leq \limsup_{n \to \infty} \PP{\abs{N_n - n} > \delta n} + \frac{2C\delta}{\epsilon^2}=  \frac{2C\delta}{\epsilon^2}.
\end{equation}
Finally, since the choice of $\delta$ is arbitrary, we have that
\begin{equation}
\limsup_{n \to \infty} \PP{\abs{\tilde{S}_{N_n}(\zeta_n) - \tilde{S}_{n}(\zeta_n) } > \epsilon \sqrt{n}} = 0
\end{equation}
and thus $\tilde{S}_{N_n}(\zeta_n) = \tilde{S}_{n}(\zeta_n) + o_p(\sqrt{n})$. Therefore,
\begin{equation}
    \sqrt{N_n}\p{\frac{S_{N_n}(\zeta_n)}{N_n} - \mu(\zeta_n)}  \Rightarrow \mathcal{N}(0, \sigma^2(0)).
\end{equation}

\end{proof}

\begin{lemm}
\label{lemma:suff_conds_Q_inv}
If a matrix $A$ satisfies the following two conditions:
\begin{enumerate}
\item $Q(p)A = I - \onev\piv(p)^\top$,
\item $\piv(p)^\top A = \zerov^\top$,
\end{enumerate}
then $A = Q^{\#}(p)$.
\end{lemm}

\begin{proof}
We will omit the dependence on $p$ throughout the proof. Recall that in Section~\ref{subsection:notation}, we define $\tilde{Q} = D_{\sqrt{\pi}} Q D_{\sqrt{\pi}}^{-1}$, which is a symmetric matrix with eigen-decomposition $\tilde{Q}  = \tilde{U} D \tilde{U}^\top$. Here $\tilde{U}$ is an orthogonal matrix.
Let $\tilde{A} = D_{\sqrt{\pi}} A D_{\sqrt{\pi}}^{-1}$. We have that
\begin{equation}
\tilde{Q} \tilde{A} = D_{\sqrt{\pi}} Q D_{\sqrt{\pi}}^{-1} D_{\sqrt{\pi}} A D_{\sqrt{\pi}}^{-1}
= D_{\sqrt{\pi}} Q A D_{\sqrt{\pi}}^{-1}
= D_{\sqrt{\pi}}(I - \onev\piv^\top)D_{\sqrt{\pi}}^{-1}
= I - \tilde{u}_1 \tilde{u}_1^\top,
\end{equation}
where $\tilde{u}_1 = (\sqrt{\pi_0}, \dots, \sqrt{\pi_K})^\top$. We also note that $\tilde{u}_1$ corresponds to the first eigenvector of $\tilde{Q}$. Now we write $\tilde{U} = (\tilde{u}_1, \tilde{U}_2)$.
Since $\tilde{Q}  = \tilde{U} D \tilde{U}^\top$, we have that
$\tilde{U} D \tilde{U}^\top \tilde{A} = I - \tilde{u}_1 \tilde{u}_1^\top$, which then implies that
\begin{equation}
D \tilde{U}^\top \tilde{A} = \tilde{U}^\top - \tilde{U}^\top\tilde{u}_1 \tilde{u}_1^\top
= \tilde{U}^\top - (1, 0, \dots, 0)^\top \tilde{u}_1 =
\begin{bmatrix}
		\tilde{u}_1^\top \\
\tilde{U}_2^\top
	\end{bmatrix}
-
\begin{bmatrix}
		\tilde{u}_1^\top \\
0
	\end{bmatrix}
= \begin{bmatrix}
		0 \\
\tilde{U}_2^\top
	\end{bmatrix}.
\end{equation}
Then since $D_{0,0} = 0$, we further have that $\tilde{U}_2^\top \tilde{A} = D_2^{-1} \tilde{U}_2^\top$, where $D_2 = D_{1:K, 1:K}$.

Now, Condition 2 implies that $\tilde{u}_1^{\top} \tilde{A} =  \onev^{\top} D_{\sqrt{\pi}} D_{\sqrt{\pi}} A D_{\sqrt{\pi}}^{-1} = \piv^\top A D_{\sqrt{\pi}}^{-1} = \zerov^\top$. Thus,
\begin{equation}
\tilde{U}^\top \tilde{A} =
\begin{bmatrix}
		\tilde{u}_1^\top \\
\tilde{U}_2^\top
\end{bmatrix}
\tilde{A}
= \begin{bmatrix}
\zerov^\top\\
D_2^{-1} \tilde{U}^\top_2
\end{bmatrix}
= D^\# \tilde{U}^\top.
\end{equation}
This further implies that
\begin{equation}
\tilde{A} = \tilde{U} D^\# \tilde{U}^\top,
\end{equation}
and hence
\begin{equation}
A = D_{\sqrt{\pi}}^{-1} \tilde{A} D_{\sqrt{\pi}} = D_{\sqrt{\pi}}^{-1} \tilde{U} D^\# \tilde{U}^\top D_{\sqrt{\pi}} = U D^\# U^{-1} = Q^\#.
\end{equation}

\end{proof}

\begin{lemm}
\label{lemm:ULLN}
Consider the continuous-time Markov chain with transition rate matrix defined in \eqref{eqn:Q_matrix}. Let $N_k(t)$ be the number of jumps from state $k$ to state $k+1$ in the time interval $[0,t]$. Let $T_{k}(t)$ be the amount of time in state $k$ in the time interval $[0,t]$. Then,
\begin{equation}
\sup_{t \in [0,T]} \abs{N_{k}(t) - t \lambda_k \pi_k } = o_p(T).
\end{equation}
\begin{equation}
\sup_{t \in [0,T]} \abs{T_{k}(t) -t \pi_k } = o_p(T).
\end{equation}
\end{lemm}
\begin{proof}
The ergodic theorem for a finite irreducible continuous-time Markov chain gives, almost surely,
\begin{equation}
\frac{N_k(t)}{t}\to\lambda_k\pi_k,
\qquad
\frac{T_k(t)}{t}\to\pi_k.
\end{equation}
For either $Y(t)=N_k(t)$ with $c=\lambda_k\pi_k$ or $Y(t)=T_k(t)$ with $c=\pi_k$, write $R(t)=Y(t)-ct$. On every sample path for which $R(t)/t\to0$, for any $\epsilon>0$ there is a finite $t_0$ such that $|R(t)|\leq\epsilon t$ for all $t\geq t_0$. Consequently,
\begin{equation}
\frac{1}{T}\sup_{0\leq t\leq T}|R(t)|
\leq \frac{1}{T}\sup_{0\leq t\leq t_0}|R(t)|+\epsilon
\end{equation}
for all $T\geq t_0$. Letting first $T\to\infty$ and then $\epsilon\downarrow0$ proves both claims, in fact almost surely.
\end{proof}

\begin{lemm}
\label{lemm:covariance_N_J}
Consider the continuous-time Markov chain with transition rate matrix defined in \eqref{eqn:Q_matrix}. For two sets $\mathcal{J}_1$ and $\mathcal{J}_2$, we have that
\begin{equation}
\EE{\tilde{N}_{\mathcal{J}_1}(m)} = m \piv^{\top} \tilde{P}_{\mathcal{J}_1} \onev + \oo(1),
\end{equation}
\begin{equation}
\EE{\tilde{T}_{\mathcal{J}_1}(m)} = \EE{\tilde{N}_{\mathcal{J}_1}(m)}/\alpha,
\end{equation}
\begin{equation}
\label{eqn:covariance_N_J}
\begin{split}
&\Cov{\tilde{N}_{\mathcal{J}_1}(m), \tilde{N}_{\mathcal{J}_2}(m)} =  m \bigg(\piv ^{\top} \tilde{P}_{\mathcal{J}_1\cap\mathcal{J}_2}\onev - (\piv ^{\top} \tilde{P}_{\mathcal{J}_1}\onev) (\piv ^{\top} \tilde{P}_{\mathcal{J}_2}\onev)\\
    &\qquad\qquad \qquad \qquad \qquad \qquad \qquad
    - \alpha \piv^{\top} \tilde{P}_{\mathcal{J}_1} Q^{\#} \tilde{P}_{\mathcal{J}_2} \onev
    - \alpha \piv^{\top} \tilde{P}_{\mathcal{J}_2} Q^{\#} \tilde{P}_{\mathcal{J}_1} \onev
    \bigg) + \oo(1),
    \end{split}
\end{equation}
\begin{equation}
    \Cov{\tilde{T}_{\mathcal{J}_1}(m), \tilde{T}_{\mathcal{J}_2}(m)} = \Cov{\tilde{N}_{\mathcal{J}_1}(m), \tilde{N}_{\mathcal{J}_2}(m)}/\alpha^2 + \EE{\tilde{N}_{\mJ_1\cap \mJ_2}(m)}/\alpha^2,
\end{equation}
\begin{equation}
    \Cov{\tilde{T}_{\mathcal{J}_1}(m), \tilde{N}_{\mathcal{J}_2}(m)} = \Cov{\tilde{N}_{\mathcal{J}_1}(m), \tilde{N}_{\mathcal{J}_2}(m)}/\alpha.
\end{equation}
Here, $\tilde{N}_{\mathcal{J}}(m)$, $\tilde{T}_{\mathcal{J}}(m)$, and $\tilde{P}_{\mathcal{J}}$ are defined in Definition \ref{defi:N_J}.
\end{lemm}

\begin{proof}
We focus on $\tilde{N}_{\mJ}(m)$ first and will move on to study $\tilde{T}_{\mJ}(m)$ later. For simplicity, we drop ``$(m)$''.

We write $\tilde{N}_{\mathcal{J}_1} = \sum_{i = 1}^{m} X_i$, where $X_i = 1$ if the $i$-th jump in the discrete-time Markov chain is an eligible jump in $\mJ_1$. Similarly, we write $\tilde{N}_{\mathcal{J}_2} = \sum_{i = 1}^{m} Y_i$, where $Y_i = 1$ if the $i$-th jump in the discrete-time Markov chain is an eligible jump in $\mJ_2$.

We are interested in the first and second moments of $\tilde{N}_{\mathcal{J}}$. Expanding the parentheses, we get that
	\begin{equation}
		\begin{split}
			\tilde{N}_{\mathcal{J}_1} &= \sum_{i = 1}^{m} X_i,\\
			\tilde{N}_{\mathcal{J}_1}^2 &= \sum_{i = 1}^{m} X_i + 2 \sum_{1\leq i < j\leq m} X_i X_j,\\
			\tilde{N}_{\mathcal{J}_1}\tilde{N}_{\mathcal{J}_2} &= \sum_{i = 1}^{m} X_i Y_i +  \sum_{1\leq i < j\leq m} X_i Y_j +\sum_{1\leq i < j\leq m} Y_i X_j  .
		\end{split}
	\end{equation}
	This then implies that
	\begin{equation}
		\label{eqn:Nj_conditional_moments_m}
		\begin{split}
			\EE{\tilde{N}_{\mathcal{J}_1}} & = A_1,\\
			\EE{\tilde{N}_{\mathcal{J}_1}^2 }&= A_1 + 2A_2,\\
			\EE{\tilde{N}_{\mathcal{J}_1}\tilde{N}_{\mathcal{J}_2} } &= B_1 + B_2 + B_2',
		\end{split}
	\end{equation}
	where
	\begin{equation}
		\begin{split}
			A_1 &= \sum_{i = 1}^{m} \EE{X_i}, \qquad A_2 = \sum_{1\leq i < j\leq m} \EE{X_i X_j}, \\
			B_1 &= \sum_{i = 1}^{m} \EE{X_i Y_i} , \qquad B_2 = \sum_{1\leq i < j\leq m} \EE{X_i Y_j} , \qquad B_2' = \sum_{1\leq i < j\leq m} \EE{Y_i X_j}.
		\end{split}
	\end{equation}
Furthermore, we have that
\begin{equation}
\label{eqn:Nj_conditional_covariance_m}
\Cov{\tilde{N}_{\mJ_1}, \tilde{N}_{\mJ_2}} = B_1 + B_2 + B_2' - A_1A_1',
\end{equation}
where $A_1' = \sum_{i = 1}^{m} \EE{Y_i}$.

	We will then study $A_1, A_2, B_1, B_2, B_2'$. The analysis of $A_1'$ is essentially the same as that of $A_1$. We have shown in the proof of Lemma \ref{lemma:num_eligible_jumps} that $A_1
	= m \piv^\top \tilde{P}_{\mJ_1} \onev - \alpha \nuv^\top Q^{\#} \tilde{P}_{\mJ_1} \onev + \errr$, where $\abs{\errr} \leq C_0 C^m$ for some constants $C_0 > 0$ and $C \in \p{0,1}$.
	For $A_2$, note that for $i < j$, $X_i X_j$ is the indicator that the $i$-th jump is eligible and the $j$-th jump is eligible. Therefore, $\EE{X_i X_j} = \nuv^\top P^{i-1}\tilde{P}_{\mJ_1} P^{j-i-1}\tilde{P}_{\mJ_1} \onev$. Hence, we can write
	\begin{equation}
		\label{eqn:A_2_formula_m}
		A_2 =  \sum_{1\leq i < j\leq m} \nuv^\top P^{i-1}\tilde{P}_{\mJ_1} P^{j-i-1}\tilde{P}_{\mJ_1} \onev
		= \nuv^\top \sum_{i+j \leq m}  P^{i-1}\tilde{P}_{\mJ_1} P^{j-1}\tilde{P}_{\mJ_1} \onev,
	\end{equation}
	where the last equality follows from a change of variable. Similarly, we have that $B_1
	= m \piv^\top \tilde{P}_{\mJ_1\cap\mJ_2} \onev - \alpha \nuv^\top Q^{\#} \tilde{P}_{\mJ_1 \cap \mJ_2} \onev + \errr_1$, where 	$\abs{\errr_1} \leq C_0 C^m$ for some constants $C_0 > 0$ and $C \in \p{0,1}$. Furthermore,
	\begin{equation}
	\label{eqn:B_2_formula_m}
		\begin{split}
B_2 &=  \sum_{1\leq i < j\leq m} \nuv^\top P^{i-1}\tilde{P}_{\mJ_1} P^{j-i-1}\tilde{P}_{\mJ_2} \onev
= \nuv^\top \sum_{i+j \leq m}  P^{i-1}\tilde{P}_{\mJ_1} P^{j-1}\tilde{P}_{\mJ_2} \onev,\\
B_2' &=  \sum_{1\leq i < j\leq m} \nuv^\top P^{i-1}\tilde{P}_{\mJ_2} P^{j-i-1}\tilde{P}_{\mJ_1} \onev
= \nuv^\top \sum_{i+j \leq m}  P^{i-1}\tilde{P}_{\mJ_2} P^{j-1}\tilde{P}_{\mJ_1} \onev.
		\end{split}
	\end{equation}

Now we focus on $B_2$. The analyses of $A_2$ and $B_2'$ are very similar. The key term in $B_2$ is $\sum_{i+j \leq m} P^{i-1}\tilde{P}_{\mJ_1} P^{j-1}\tilde{P}_{\mJ_2}$, which is a sum of products of terms like $P^{i-1}\tilde{P}_{\mJ}$. Let $R_i = P^i - \onev \piv^\top$. Then we can decompose each $P^{i-1}$ into two terms $R_{i-1}$ and $\onev \piv^\top$, and thus
	\begin{equation}
		\label{eqn:B2_decomposition}
		\begin{split}
			&P^{i-1}\tilde{P}_{\mJ_1} P^{j-1}\tilde{P}_{\mJ_2} \\
			&\quad = \p{R_{i-1} + \onev \piv^\top}\tilde{P}_{\mJ_1} \p{R_{j-1} + \onev \piv^\top} \tilde{P}_{\mJ_2}\\
			&\quad = \p{\onev \piv^\top \tilde{P}_{\mJ_1}}\p{\onev \piv^\top \tilde{P}_{\mJ_2}} + R_{i-1}\tilde{P}_{\mJ_1} \p{\onev \piv^\top \tilde{P}_{\mJ_2}}\\
 &\qquad\qquad\qquad\qquad + \p{\onev \piv^\top \tilde{P}_{\mJ_1}} R_{j-1}\tilde{P}_{\mJ_2} + R_{i-1}\tilde{P}_{\mJ_1}R_{j-1}\tilde{P}_{\mJ_2}.
		\end{split}
	\end{equation}
	By Lemma \ref{lemma:operator_norm}, there are constants $C_0>0$ and $C \in (0,1)$ such that $\Norm{R_{i-1}}_{op} \leq C_0C^{i-1}$. Therefore,
	\begin{equation}
		\begin{split}
			\Norm{ \sum_{i+j \leq m} R_{i-1}\tilde{P}_{\mJ_1}R_{j-1}\tilde{P}_{\mJ_2}}
 &\leq  \sum_{i+j \leq m} \Norm{R_{i-1}}_{op}\Norm{\tilde{P}_{\mJ_1}}_{op}\Norm{R_{j-1}}_{op}\Norm{\tilde{P}_{\mJ_2}}_{op}\\
			&\leq C_0^2\sum_{i+j \leq m} C^{i+j-2} \leq C_1,
		\end{split}
	\end{equation}
	for some constant $C_1 > 0$.

Then, we look at terms that involve only one $R_{i-1}$. Note that
	\begin{equation}
		\label{eqn:one_R_result_m}
		\begin{split}
			\sum_{i+j \leq m} R_{i-1}
			= \sum_{i = 1}^{m-1} (m-i) R_{i-1}
			= m\sum_{i=1}^{m-1} R_{i-1} - \sum_{i = 1}^{m-1}i R_{i-1}.
		\end{split}
	\end{equation}
	For the first term, by Lemma \ref{lemma:sum_of_power}, we have that $\sum_{i=1}^{m-1} R_{i-1} = -\alpha Q^\# + E$, where $\Norm{E}_{op} \leq C_1 C^m$ for some constants $C_1 > 0$ and $C \in (0,1)$. For the second term, Lemma~\ref{lemma:operator_norm} gives $\Norm{R_{i-1}}_{op} \leq C_0C^{i-1}$ after adjusting the constants if necessary. Therefore,
	\begin{equation}
		\Norm{\sum_{i = 1}^{m-1} i R_{i-1}}_{op}
		\leq C_0\sum_{i = 1}^{m-1} i C^{i-1}
		\leq C_1,
	\end{equation}
	for some constant $C_1 > 0$. Combining the results and plugging back into \eqref{eqn:one_R_result_m}, we have that
	\begin{equation}
		\sum_{i+j \leq m} R_{i-1} = -m\alpha Q^\# + E_1,
	\end{equation}
	where $\Norm{E_1}_{op} \leq C_1$ for some constant $C_1 > 0$.

	We have analyzed all terms in \eqref{eqn:B2_decomposition}. Combining the results, we have
	\begin{equation}
		\begin{split}
			&\sum_{i+j \leq m} P^{i-1}\tilde{P}_{\mJ_1} P^{j-1}\tilde{P}_{\mJ_2} \\
			&\qquad= \sum_{i+j \leq m} \p{\onev \piv^\top \tilde{P}_{\mJ_1}}\p{\onev \piv^\top \tilde{P}_{\mJ_2}}
			 - m\alpha \p{ Q^\#\tilde{P}_{\mJ_1} \onev \piv^\top \tilde{P}_{\mJ_2} +   \onev\piv^\top \tilde{P}_{\mJ_1} Q^\# \tilde{P}_{\mJ_2} } + E_2\\
			&\qquad= \binom{m}{2} \p{\onev \piv^\top \tilde{P}_{\mJ_1}}\p{\onev \piv^\top \tilde{P}_{\mJ_2}}   - m\alpha \p{ Q^\#\tilde{P}_{\mJ_1} \onev \piv^\top \tilde{P}_{\mJ_2} +   \onev\piv^\top \tilde{P}_{\mJ_1} Q^\# \tilde{P}_{\mJ_2} } + E_2
		\end{split}
	\end{equation}
	where $\Norm{E_2}_{op} \leq C_4$ for some constant $C_4 > 0$.

	Therefore, by \eqref{eqn:A_2_formula_m}-\eqref{eqn:B_2_formula_m}, we have the following formulae for $A_1, A_2, B_1, B_2, B_2'$.
	\begin{equation}
		\begin{split}
			A_1 & = m \p{\piv^\top \tilde{P}_{\mJ_1} \onev} - \alpha \nuv^\top Q^\# \tilde{P}_{\mJ_1} \onev  + \errr_1,\\
			A_2 & = \frac{m^2}{2} \p{\piv^\top \tilde{P}_{\mJ_1}  \onev}^2 - m \Bigg[\alpha \p{\piv^\top \tilde{P}_{\mJ_1}  Q^\# \tilde{P}_{\mJ_1}  \onev}\\
			& \qquad \qquad  \qquad \qquad \qquad + \alpha \p{\nuv^\top Q^\# \tilde{P}_{\mJ_1}  \onev}\p{\piv^\top \tilde{P}_{\mJ_1}  \onev} + \frac{1}{2} \p{ \piv^\top \tilde{P}_{\mJ_1}\onev }^2 \Bigg] + \errr_2,\\
			B_1 & = m \p{\piv^\top \tilde{P}_{\mJ_1 \cap \mJ_2} \onev} - \alpha \nuv^\top Q^\# \tilde{P}_{\mJ_1\cap \mJ_2} \onev  + \errr_3,\\
			B_2 & = \frac{m^2}{2} \p{ \piv^\top \tilde{P}_{\mJ_1} \onev}\p{\piv^\top \tilde{P}_{\mJ_2}\onev} - m \Bigg[\alpha \p{\piv^\top \tilde{P}_{\mJ_1}  Q^\# \tilde{P}_{\mJ_2}  \onev}\\
			& \qquad \qquad \quad + \alpha \p{\nuv^\top Q^\# \tilde{P}_{\mJ_1}  \onev}\p{\piv^\top \tilde{P}_{\mJ_2}  \onev} + \frac{1}{2}  \p{ \piv^\top \tilde{P}_{\mJ_1} \onev}\p{\piv^\top \tilde{P}_{\mJ_2}\onev} \Bigg] + \errr_4,\\
			B_2' & = \frac{m^2}{2} \p{ \piv^\top \tilde{P}_{\mJ_1} \onev}\p{\piv^\top \tilde{P}_{\mJ_2}\onev} - m \Bigg[\alpha \p{\piv^\top \tilde{P}_{\mJ_2}  Q^\# \tilde{P}_{\mJ_1}  \onev}\\
			& \qquad \qquad \quad + \alpha \p{\nuv^\top Q^\# \tilde{P}_{\mJ_2}  \onev}\p{\piv^\top \tilde{P}_{\mJ_1}  \onev} + \frac{1}{2}  \p{ \piv^\top \tilde{P}_{\mJ_2} \onev}\p{\piv^\top \tilde{P}_{\mJ_1}\onev} \Bigg] + \errr_5.
		\end{split}
	\end{equation}
	where $\abs{\errr_1} \leq C_1 C^m$, $\abs{\errr_2} \leq C_2$, $\abs{\errr_3} \leq C_3 C^m$, $\abs{\errr_4} \leq C_4$, and $\abs{\errr_5} \leq C_5$, for some constants $C \in (0,1), C_1,C_2, C_3, C_4, C_5 > 0$.

Plugging things back into \eqref{eqn:Nj_conditional_moments_m} and \eqref{eqn:Nj_conditional_covariance_m}, we have that
\begin{equation}
\EE{\tilde{N}_{\mathcal{J}_1}} = m \piv^{\top} \tilde{P}_{\mathcal{J}_1} \onev + \oo(1),
\end{equation}
\begin{equation}
\begin{split}
&\Cov{\tilde{N}_{\mathcal{J}_1}, \tilde{N}_{\mathcal{J}_2}} =
     m \bigg(\piv ^{\top} \tilde{P}_{\mathcal{J}_1\cap\mathcal{J}_2}\onev - (\piv ^{\top} \tilde{P}_{\mathcal{J}_1}\onev) (\piv ^{\top} \tilde{P}_{\mathcal{J}_2}\onev) \\
   &\qquad \qquad \qquad \qquad \qquad \qquad \qquad  - \alpha \piv^{\top} \tilde{P}_{\mathcal{J}_1} Q^{\#} \tilde{P}_{\mathcal{J}_2} \onev - \alpha \piv^{\top} \tilde{P}_{\mathcal{J}_2} Q^{\#} \tilde{P}_{\mathcal{J}_1} \onev
    \bigg) + \oo(1).
    \end{split}
\end{equation}

Finally, for $\tilde{T}_{\mJ}$, we note that we can write $\tilde{T}_{\mJ} = \sum_{i = 1}^{\tilde{N}_{\mJ}} H_i$, where $H_i \sim \operatorname{Exp}(\alpha)$ independently and $H_i$'s are independent of $\tilde{N}_{\mJ}$. Therefore,
\begin{equation}
\EE{\tilde{T}_{\mJ}} = \EE{\tilde{N}_{\mJ}}/\alpha,
\end{equation}
\begin{equation}
\begin{split}
&\Cov{\tilde{T}_{\mJ_1}, \tilde{T}_{\mJ_2}}\\
 &\qquad = \Cov{\EE{\tilde{T}_{\mJ_1}\mid \tilde{N}_{\mJ_1}}, \EE{\tilde{T}_{\mJ_2}\mid \tilde{N}_{\mJ_2}}} + \EE{\Cov{\tilde{T}_{\mJ_1},\tilde{T}_{\mJ_2}\mid \tilde{N}_{\mJ_1}, \tilde{N}_{\mJ_2}}}\\
& \qquad = \Cov{ \tilde{N}_{\mJ_1}, \tilde{N}_{\mJ_2}}/\alpha^2 +  \EE{\tilde{N}_{\mJ_1 \cap \mJ_2}}/\alpha^2,
\end{split}
\end{equation}
and
\begin{equation}
    \Cov{\tilde{T}_{\mathcal{J}_1}, \tilde{N}_{\mathcal{J}_2}} = \Cov{\tilde{N}_{\mathcal{J}_1}, \tilde{N}_{\mathcal{J}_2}}/\alpha.
    \end{equation}

\end{proof}

\begin{lemm}[Uniform-integrability bound]
\label{lemm:uniform_intergral_disc}
Consider the continuous-time Markov chain with transition rate matrix $Q$ defined in \eqref{eqn:Q_matrix}. There exists a constant $C$ such that
\begin{equation}
	\EE{\p{\tilde{N}_{\mJ}(m) - \piv^{\top} \tilde{P}_{\mJ}\onev m }^4} \leq  C m^2,
\end{equation}
\begin{equation}
	\EE{\p{\tilde{T}_{\mJ}(m) - \piv^{\top} \tilde{P}_{\mJ}\onev m/\alpha}^4} \leq  C m^2,
\end{equation}
where $\tilde{N}_{\mJ}(m)$, $\tilde{T}_{\mJ}(m)$ and $\tilde{P}_{\mJ}$ are defined in Definition~\ref{defi:N_J}.
\end{lemm}

\begin{proof}
The results for $\tilde{N}_{\mJ}(m)$ are included in the proof of Lemma \ref{lemma:uniform_intergral}. For $\tilde{T}_{\mJ}(m) $, we note that we can write
\begin{equation}
\tilde{T}_{\mJ}(m) = \sum_{i = 1}^{\tilde{N}_{\mJ}(m)} H_i
\end{equation}
where the $H_i \sim \operatorname{Exp}(\alpha)$ are independent of one another and of $\tilde{N}_{\mJ}(m)$. Let $\gamma=\piv^\top\tilde P_{\mJ}\onev$, $N=\tilde N_{\mJ}(m)$, and $Y_i=H_i-1/\alpha$. Then
\begin{equation}
\begin{split}
&\EE{\p{\tilde{T}_{\mJ}(m) - \gamma m/\alpha}^4}\\
&\qquad = \EE{\p{\sum_{i=1}^{N}Y_i + \frac{N-\gamma m}{\alpha}}^4}\\
&\qquad \leq 8\EE{\p{\sum_{i=1}^{N}Y_i}^4}
    +\frac{8}{\alpha^4}\EE{\p{N-\gamma m}^4}.
\end{split}
\end{equation}
Conditionally on $N=n$, the variables $Y_i$ are independent, centered, and satisfy $\EE{Y_i^2}=1/\alpha^2$ and $\EE{Y_i^4}=9/\alpha^4$. Hence
\begin{equation}
\EE{\p{\sum_{i=1}^{N}Y_i}^4\mid N=n}
=\frac{3n^2+6n}{\alpha^4}.
\end{equation}
Since $N\leq m$ and the first assertion of the lemma gives $\EE{(N-\gamma m)^4}\leq Cm^2$, the preceding bounds imply
\begin{equation}
\EE{\p{\tilde{T}_{\mJ}(m) - \gamma m/\alpha}^4}\leq C'm^2
\end{equation}
for a constant $C'>0$.
\end{proof}

\begin{lemm}
\label{lemm:binom_clt}
Let $X_n \sim \operatorname{Binom}(n, p_n)$ be a sequence of Binomial random variables. Assume that $p_n \to c$ for some constant $c \in (0,1)$. Then
\begin{equation}
\frac{\sqrt{n}}{\sqrt{p_n(1-p_n)}} \p{\frac{X_n}{n} - p_n} \Rightarrow \mathcal{N}(0,1).
\end{equation}
\end{lemm}
\begin{proof}
We prove the lemma using the Lindeberg-Feller theorem \citep{durrett2019probability}. We start by rewriting $X_n = Y_{n,1} + Y_{n,2} + \cdots + Y_{n,n}$, where the $Y_{n,i}$ are i.i.d. $\operatorname{Bern}(p_n)$ random variables. Let $Z_{n,i} = (Y_{n,i} - p_n)/\sqrt{n p_n(1-p_n)}$. Then, $\EE{Z_{n,i}} = 0$ and $\Var{Z_{n,i}} = 1/n$. For any $\epsilon > 0$,
\begin{equation}
\EE{\left|Z_{n, i}\right|^{2} ;\left|Z_{n, i} \right| > \epsilon}
\leq \frac{\EE{\abs{Z_{n, i}}^4}}{\epsilon^2}
\leq \frac{1}{n^2p_n^2(1-p_n)^2 \epsilon^2}.
\end{equation}
Thus,
\begin{equation}
\sum_{i=1}^{n} \EE{\left|Z_{n, i}\right|^{2} ;\left|Z_{n, i}\right|>\epsilon}
\leq  \frac{n}{n^2p_n^2(1-p_n)^2 \epsilon^2},
\end{equation}
which goes to 0 as $n \to \infty$. By Lindeberg-Feller theorem,
\begin{equation}
\frac{\sqrt{n}}{\sqrt{p_n(1-p_n)}} \p{\frac{X_n}{n} - p_n} = \sum_{i = 1}^n Z_{n,i}  \Rightarrow \mathcal{N}(0,1).
\end{equation}

\end{proof}

\begin{lemm}
\label{lemma:N_k_T_converge}
Under the conditions of Theorem \ref{theo:LE_CLT}, if we define $N_{k}(T,\zeta_T) = N_{k,+}(T,\zeta_T) + N_{k,-}(T,\zeta_T)$, then
\begin{equation}
N_{k}(T,\zeta_T)/T \stackrel{p}{\to} \pi_k(p)\lambda_k(p).
\end{equation}

\end{lemm}
\begin{proof}
We start by noting that the arrival rate when the length of the queue is $k$ is $\tilde{\lambda}_k(p,\zeta) = (\lambda_k(p + \zeta) + \lambda_k(p - \zeta))/2$. In other words, the arrival rate can be written as a function of $\zeta$. If we treat $\zeta$ as the new price and $\tilde{\lambda}$ as the new $\lambda$ function, we can still apply results in Section \ref{subsection:notation} and Lemma \ref{lemma:num_eligible_jumps}. By Lemma \ref{lemma:num_eligible_jumps}, if we define the eligible jumps to be jumps from state $k$ to state $k+1$, then we have that
\begin{equation}
    N_{k}(T,\zeta_T)/T = \pi_k(p,\zeta_T) \tilde{\lambda}_k(p,\zeta_T) + o_p(1),
\end{equation}
where $\pi_k(p,\zeta)$ denotes the corresponding steady-state probability of state $k$. Both quantities are continuous at $\zeta=0$, and hence
\begin{equation}
\pi_k(p,\zeta_T) \to \pi_k(p,0) = \pi_k(p),
\qquad
\tilde{\lambda}_k(p,\zeta_T) \to \tilde{\lambda}_k(p,0) = \lambda_k(p).
\end{equation}
It follows that
\begin{equation}
    N_{k}(T,\zeta_T)/T = \pi_k(p)\lambda_k(p) + o_p(1).
\end{equation}
\end{proof}

\begin{lemm}
\label{lemma:pi_converge}
Under the conditions of Theorem \ref{theo:LE_CLT}, if we define $\pi_i(p, \zeta_T)$ to be the steady-state probability of state $i$ when the state-dependent arrival rate is set to be $\tilde{\lambda}_i(p, \zeta_T) = (\lambda_i(p-\zeta_T) + \lambda_i(p+\zeta_T))/2$, then for any $k$,
\begin{equation}
    \abs{\pi_k(p, \zeta_T) - \pi_k(p, 0) } \leq C \zeta_T^2
\end{equation}
for some constant $C$ depending on $K$, $B_0$, $B_1$, $B_2$. Note that $\pi_k(p, 0)$ corresponds to our previous notation of $\pi_k(p)$.
\end{lemm}
\begin{proof}
We start by noting that
\begin{equation}
    \pi_0(p, \zeta) = \frac{1}{1 + \sum_{k=1}^{K} \prod_{i = 1}^k \tilde{\lambda}_{i-1}(p, \zeta)/\mu}.
\end{equation}
Therefore, it is straightforward to see that
\begin{equation}
    \abs{\frac{d}{d\zeta}\pi_0(p, \zeta)} \leq C \max_k \abs{\frac{d}{d\zeta}\tilde{\lambda}_k(p, \zeta)},
\end{equation}
for some constant $C$ depending on $K$, $B_0$, $B_1$, $B_2$. We also note that
\begin{equation}
\pi_k(p, \zeta)=\pi_0(p, \zeta) \prod_{i=1}^k \frac{\tilde{\lambda}_{i-1}(p, \zeta)}{\mu}.
\end{equation}
Therefore, for any $k \geq 0$,
\begin{equation}
    \abs{\frac{d}{d\zeta}\pi_k(p, \zeta)} \leq C \max_k \abs{\frac{d}{d\zeta}\tilde{\lambda}_k(p, \zeta)},
\end{equation}
for some constant $C$ depending on $K$, $B_0$, $B_1$, $B_2$.
Then we will turn our attention to $\frac{d}{d\zeta}\tilde{\lambda}_k(p, \zeta)$. Note that since $\abs{\lambda_k''(p)} \leq B_2$,
\begin{equation}
    \abs{\frac{d}{d\zeta}\tilde{\lambda}_k(p, \zeta)}
    = \frac{1}{2}\abs{\lambda_k'(p+\zeta) - \lambda_k'(p-\zeta)}
    \leq \zeta B_2.
\end{equation}
This then implies that
\begin{equation}
    \abs{\frac{d}{d\zeta}\pi_k(p, \zeta)} \leq C \max_j \abs{\frac{d}{d\zeta}\tilde{\lambda}_j(p, \zeta)} \leq C_1 \zeta
\end{equation}
for some constant $C_1$. Finally, the above result holds for every $\zeta$ between $0$ and $\zeta_T$, and thus
\begin{equation}
    \abs{\pi_k(p, \zeta_T) - \pi_k(p, 0) } \leq C_1 \zeta_T^2.
\end{equation}
\end{proof}

\subsection{Proofs of Propositions}

\subsubsection{Proof of Proposition \ref{prop:first_esti_triangular}}
\label{proof:prop:first_esti_triangular}
\paragraph{Regenerative switchback experiment}
We start by analyzing the regenerative switchback experiment. Let $M(T, \zeta_T)$ be the number of times the queue hits the state $k_r$.
Let $W_i(\zeta_T) = 1$ if the price changes to $p+\zeta_T$ at the $i$-th time of hitting $k_r$ and let $W_i(\zeta_T) = 0$ otherwise. Let $A_{i,+}(\zeta_T)$ ($A_{i,-}(\zeta_T)$) be the number of arrivals between the $i$-th and the $(i+1)$-th time of hitting $k_r$ if the price is $p+\zeta_T$ ($p-\zeta_T$).
Let $B_{i,+}(\zeta_T)$ ($B_{i,-}(\zeta_T)$) be the amount of time between the $i$-th and the $(i+1)$-th hitting of $k_r$ if the price is $p+\zeta_T$ ($p-\zeta_T$). In reality, we observe only one potential pair; without loss of generality, couple the potential outcomes so that $(A_{i,+}(\zeta_T),B_{i,+}(\zeta_T)) \indp (A_{i,-}(\zeta_T),B_{i,-}(\zeta_T))$.\footnote{For notational simplicity, here we set 0 to be the 0-th time of hitting $k_r$, even if we do not start from $k_r$.}
Let $A_{i}(\zeta_T) = W_{i}(\zeta_T) A_{i,+}(\zeta_T) + (1-W_{i}(\zeta_T)) A_{i,-}(\zeta_T)$ and $B_{i}(\zeta_T) = W_{i}(\zeta_T) B_{i,+}(\zeta_T) + (1-W_{i}(\zeta_T)) B_{i,-}(\zeta_T)$ be the realized number of arrivals and inter-hitting duration, respectively.

Note that $A_{1,+}(\zeta_T), A_{2,+}(\zeta_T), \dots$ are i.i.d. random variables with bounded fourth moment (uniformly across $\zeta_T$). The same holds for $A_{i,-}(\zeta_T)$, $B_{i,+}(\zeta_T)$ and $B_{i,-}(\zeta_T)$.
By definition,
\begin{equation}
\sum_{i = 0}^{M(T, \zeta_T)-1} B_i(\zeta_T) \leq T
\leq \sum_{i = 0}^{M(T, \zeta_T)} B_i(\zeta_T).
\end{equation}
Thus, by the law of large numbers, $M(T, \zeta_T) / T \stackrel{p}{\to} 1/\EE{B_i(0)}$.

We can write similar inequalities for $N_{ +}(T, \zeta_T)$, $T_{ +}(T, \zeta_T)$, $N_{ -}(T, \zeta_T)$ and $T_{ -}(T, \zeta_T)$:
\begin{equation}
\begin{split}
\sum_{i = 0}^{M(T, \zeta_T)-1} A_{i,+}(\zeta_T)W_{i}(\zeta_T)
&\leq N_{ +}(T, \zeta_T)
\leq \sum_{i = 0}^{M(T, \zeta_T)} A_{i,+}(\zeta_T)W_{i}(\zeta_T),\\
\sum_{i = 0}^{M(T, \zeta_T)-1} B_{i,+}(\zeta_T)W_{i}(\zeta_T)
&\leq T_{ +}(T, \zeta_T)
\leq \sum_{i = 0}^{M(T, \zeta_T)} B_{i,+}(\zeta_T)W_{i}(\zeta_T),\\
\sum_{i = 0}^{M(T, \zeta_T)-1} A_{i,-}(\zeta_T)(1-W_{i}(\zeta_T))
&\leq N_{ -}(T, \zeta_T)
\leq \sum_{i = 0}^{M(T, \zeta_T)} A_{i,-}(\zeta_T)(1-W_{i}(\zeta_T)),\\
\sum_{i = 0}^{M(T, \zeta_T)-1} B_{i,-}(\zeta_T)(1-W_{i}(\zeta_T))
&\leq T_{ -}(T, \zeta_T)
\leq \sum_{i = 0}^{M(T, \zeta_T)} B_{i,-}(\zeta_T)(1-W_{i}(\zeta_T)).
\end{split}
\end{equation}
By Lemma \ref{lemma:Lyap_anscombe}, we can show that
\begin{equation}
\sqrt{M(T, \zeta_T)}
\begin{pmatrix}
\frac{N_{ +}(T, \zeta_T)}{M(T, \zeta_T)} - \EE{A_{1,+}(\zeta_T)/2}\\
\frac{T_{ +}(T, \zeta_T)}{M(T, \zeta_T)} - \EE{B_{1,+}(\zeta_T)/2}\\
\frac{N_{ -}(T, \zeta_T)}{M(T, \zeta_T)} - \EE{A_{1,-}(\zeta_T)/2}\\
\frac{T_{ -}(T, \zeta_T)}{M(T, \zeta_T)} - \EE{B_{1,-}(\zeta_T)/2}
\end{pmatrix}
\Rightarrow \mathcal{N}(\zerov, \Sigma),
\end{equation}
as $T \to \infty$, where
$
\Sigma = \begin{pmatrix}
	\Sigma_0 & \Sigma_{+-}\\
	\Sigma_{+-} & \Sigma_0
\end{pmatrix},
$
and
\begin{equation}
\label{ref:eqn_sigma0}
\begin{split}
& \Sigma_0 =\\
& \frac{1}{2}\begin{pmatrix}
	\Var{A_1(0)} + \frac{1}{2}\EE{A_1(0)}^2   & \Cov{A_1(0), B_1(0)} + \frac{1}{2}\EE{A_1(0)}\EE{B_1(0)}\\
\Cov{A_1(0), B_1(0)} + \frac{1}{2}\EE{A_1(0)}\EE{B_1(0)} &\Var{B_1(0)} + \frac{1}{2}\EE{B_1(0)}^2
\end{pmatrix}.
\end{split}
\end{equation}
where
\begin{equation}
\Sigma_{+-} = -\frac{1}{4}
\begin{pmatrix}
    \EE{A_1(0)}^2 & \EE{A_1(0)}\EE{B_1(0)}\\
    \EE{A_1(0)}\EE{B_1(0)} & \EE{B_1(0)}^2
\end{pmatrix}.
\end{equation}
The cross-covariance vanishes after applying the ratio map. Then, by the delta method \citep{lehmann2006theory}, the above implies that
\begin{equation}
\sqrt{M(T, \zeta_T)}
\begin{pmatrix}
\frac{N_{ +}(T, \zeta_T)}{T_{ +}(T, \zeta_T)} - \frac{\EE{A_{1,+}(\zeta_T)}}{ \EE{B_{1,+}(\zeta_T)}}\\
\frac{N_{ -}(T, \zeta_T)}{T_{ -}(T, \zeta_T)} - \frac{\EE{A_{1,-}(\zeta_T)}}{ \EE{B_{1,-}(\zeta_T)}}\\
\end{pmatrix}
\Rightarrow \mathcal{N}(\zerov, \Sigma_1),
\end{equation}
as $T \to \infty$, where
$
\Sigma_1 = \begin{pmatrix}
	\sigma_1^2 & 0\\
	0 & \sigma_1^2
\end{pmatrix}
$
and
\begin{equation}
	\label{eqn:variance_formula}
	\sigma_1^2 = 2\frac{\EE{A_1(0)}^2}{\EE{B_1(0)}^2}\p{\frac{\Var{A_1(0)}}{\EE{A_1(0)}^2} - \frac{2\Cov{A_1(0), B_1(0)}}{\EE{A_1(0)}\EE{B_1(0)}}  + \frac{\Var{B_1(0)}}{\EE{B_1(0)}^2} }  .
\end{equation}
This further implies that
\begin{equation}
\sqrt{T}
\begin{pmatrix}
\frac{N_{ +}(T, \zeta_T)}{T_{ +}(T, \zeta_T)} - \frac{\EE{A_{1,+}(\zeta_T)}}{ \EE{B_{1,+}(\zeta_T)}}\\
\frac{N_{ -}(T, \zeta_T)}{T_{ -}(T, \zeta_T)} - \frac{\EE{A_{1,-}(\zeta_T)}}{ \EE{B_{1,-}(\zeta_T)}}\\
\end{pmatrix}
\Rightarrow \mathcal{N}(\zerov,  \EE{B_1(0)}\Sigma_1),
\end{equation}
as $T \to \infty$.

\paragraph{Interval switchback experiment}
For the interval switchback experiment, the analysis is very similar.

Let $\tau_1, \tau_2, \dots, \tau_{n_t}$ represent the times at which the queue hits state $k_r$. Let $I_1 = [\tau_1, \tau_2)$, $I_2 = [\tau_2, \tau_3)$, $\dots$, $I_{n_t-1} = [\tau_{n_t-1}, \tau_{n_t})$ be the time intervals between the hitting times. We call an interval a positive interval if the price stays at $p + \zeta_T$ throughout the time interval, a negative interval if the price stays at $p - \zeta_T$ throughout the time interval, and a mixed interval if the price changes during the interval. Let $M_{+}(T, \zeta_T)$ be the total number of positive intervals. Let $A_{i,+}(\zeta_T)$ be the number of arrivals in the $i$-th positive interval. Let $B_{i,+}(\zeta_T)$ be the amount of time in the $i$-th positive interval. We define $M_{-}(T, \zeta_T)$, $A_{i,-}(\zeta_T)$ and $B_{i,-}(\zeta_T)$ similarly for the negative intervals.

By Assumption \ref{assu:two_intervals}, the total number of switches is upper bounded by $T / l_T = o(\sqrt{T})$, and thus the total number of mixed intervals is $o_p(\sqrt{T})$.
Let $T_+$ be the amount of time the price is at $p + \zeta_T$.
Therefore, we have that
\begin{equation}
    \sum_{i = 1}^{M_{+}(T, \zeta_T)}  B_{i,+}(\zeta_T) = T_+ + o_p(\sqrt{T}).
\end{equation}
Furthermore, because $T_+ = T/2 + o_p(T)$, we have that
\begin{equation}
    \sum_{i = 1}^{M_{+}(T, \zeta_T)}  B_{i,+}(\zeta_T) = T/2 + o_p(T).
\end{equation}
Note that $(A_{1,+}(\zeta_T),B_{1,+}(\zeta_T)), (A_{2,+}(\zeta_T),B_{2,+}(\zeta_T)), \dots$ are i.i.d. random vectors with bounded fourth moments (uniformly across $\zeta_T$), and thus, by the law of large numbers, $M_{+}(T, \zeta_T) / T \stackrel{p}{\to} 1/(2\EE{B_i(0)})$.

Similarly, we have that
\begin{equation}
\begin{split}
    \sum_{i = 1}^{M_{+}(T, \zeta_T)}  A_{i,+}(\zeta_T) &= N_+(T,\zeta_T) + o_p(\sqrt{T}),\\
    \sum_{i = 1}^{M_{-}(T, \zeta_T)}  B_{i,-}(\zeta_T) &= T_- + o_p(\sqrt{T}),\\
    \sum_{i = 1}^{M_{-}(T, \zeta_T)}  A_{i,-}(\zeta_T) &= N_-(T,\zeta_T) + o_p(\sqrt{T}).
\end{split}
\end{equation}
By Lemma \ref{lemma:Lyap_anscombe}, we can show that
\begin{equation}
\sqrt{M_{+}(T, \zeta_T)}
\begin{pmatrix}
\frac{N_{ +}(T, \zeta_T)}{M_{+}(T, \zeta_T)} - \EE{A_{1,+}(\zeta_T)}\\
\frac{T_+}{M_{+}(T, \zeta_T)} - \EE{B_{1,+}(\zeta_T)}
\end{pmatrix}
\Rightarrow \mathcal{N}(\zerov, \Sigma_{\mathrm{cycle}}),
\end{equation}
as $T \to \infty$, where
\begin{equation}
\Sigma_{\mathrm{cycle}} =
\begin{pmatrix}
    \Var{A_1(0)} & \Cov{A_1(0),B_1(0)}\\
    \Cov{A_1(0),B_1(0)} & \Var{B_1(0)}
\end{pmatrix}.
\end{equation}
Then, by the delta method \citep{lehmann2006theory}, the above implies that
\begin{equation}
\sqrt{M_{+}(T, \zeta_T)}
\p{\frac{N_{ +}(T, \zeta_T)}{T_+} - \frac{\EE{A_{1,+}(\zeta_T)}}{ \EE{B_{1,+}(\zeta_T)}}}
\Rightarrow \mathcal{N}(0, \sigma_1^2/2),
\end{equation}
as $T \to \infty$, where $\sigma_1^2$ is defined in  \eqref{eqn:variance_formula}.
This further implies that
\begin{equation}
\sqrt{T}
\p{\frac{N_{ +}(T, \zeta_T)}{T_+} - \frac{\EE{A_{1,+}(\zeta_T)}}{ \EE{B_{1,+}(\zeta_T)}}}
\Rightarrow \mathcal{N}(0,  \EE{B_1(0)}\sigma_1^2).
\end{equation}
The same analysis for the price $p-\zeta_T$ gives
\begin{equation}
\sqrt{T}
\p{\frac{N_{ -}(T, \zeta_T)}{T_-} - \frac{\EE{A_{1,-}(\zeta_T)}}{ \EE{B_{1,-}(\zeta_T)}}}
\Rightarrow \mathcal{N}(0,  \EE{B_1(0)}\sigma_1^2).
\end{equation}

Finally, note that in the proof, we indeed have $A_{i,+}(\zeta_T)$ and $B_{i,+}(\zeta_T)$ independent of $A_{i,-}(\zeta_T)$ and $B_{i,-}(\zeta_T)$, and thus a joint analysis of them gives
\begin{equation}
\sqrt{T}
\begin{pmatrix}
\frac{N_{ +}(T, \zeta_T)}{T_+} - \frac{\EE{A_{1,+}(\zeta_T)}}{ \EE{B_{1,+}(\zeta_T)}}\\
\frac{N_{ -}(T, \zeta_T)}{T_-} - \frac{\EE{A_{1,-}(\zeta_T)}}{ \EE{B_{1,-}(\zeta_T)}}\\
\end{pmatrix}
\Rightarrow \mathcal{N}(\zerov,  \EE{B_1(0)}\Sigma_1),
\end{equation}
as $T \to \infty$, where
$
\Sigma_1 = \begin{pmatrix}
	\sigma_1^2 & 0\\
	0 & \sigma_1^2
\end{pmatrix}
$, and $\sigma_1^2$ is defined in \eqref{eqn:variance_formula}.

\paragraph{Notation} To make the notation consistent with the one in the statement of the proposition, we take
$\tilde{\mu}(p+\zeta_T) = \frac{\EE{A_{1,+}(\zeta_T)}}{ \EE{B_{1,+}(\zeta_T)}}$,
$\tilde{\mu}(p-\zeta_T) = \frac{\EE{A_{1,-}(\zeta_T)}}{ \EE{B_{1,-}(\zeta_T)}}$,
and $\tilde{\sigma}^2 = \EE{B_1(0)} \sigma_1^2$. Here, $\sigma_1^2$ is defined in \eqref{eqn:variance_formula}.

\subsubsection{Proof of Proposition \ref{prop:first_esti_zero}}
\label{proof:prop:first_esti_zero}
Let $N_{\arr}(T)$ be the number of arrivals to the system, and let $M_{k_r}(T)$ be the number of times the queue hits the state $k_r$. Let $A_i$ be the number of arrivals between the $i$-th and the $(i+1)$-th time of hitting $k_r$.
Let $B_i$ be the amount of time between the $i$-th and the $(i+1)$-th hitting of $k_r$.

Note that $A_{1}, A_{2}, \dots$ are i.i.d. random variables with bounded fourth moment. The same holds for $B_{i}$.
By definition,
\begin{equation}
	\sum_{i = 0}^{M_{k_r}(T)-1} B_i \leq T
	\leq \sum_{i = 0}^{M_{k_r}(T)} B_i.
\end{equation}
Thus, by the law of large numbers, $M_{k_r}(T) / T \stackrel{p}{\to} 1/\EE{B_1}$.

We can write similar inequalities for $N_{\arr}(T)$:
\begin{equation}
		\sum_{i = 0}^{M_{k_r}(T)-1} A_i
		\leq N_{\arr}(T)
			\leq \sum_{i = 0}^{M_{k_r}(T)} A_i.
\end{equation}
By Anscombe's theorem \citep{anscombe1952large}, we can show that
	\begin{equation}
		\sqrt{M_{k_r}(T)}
		\begin{pmatrix}
			\frac{N_{\arr}(T)}{M_{k_r}(T)} - \EE{A_i}\\
			\frac{T}{M_{k_r}(T)} - \EE{B_1}\\
		\end{pmatrix}
		\Rightarrow \mathcal{N}(\zerov, \Sigma_2),
	\end{equation}
	where
	\begin{equation}
		\Sigma_2 = \begin{pmatrix}
			\Var{A_1}    & \Cov{A_1,B_1}\\
			\Cov{A_1,B_1}  & \Var{B_1}
		\end{pmatrix}.
	\end{equation}
		Then, by the delta method \citep{lehmann2006theory}, the above implies that
	\begin{equation}
		\sqrt{M_{k_r}(T)}\p{
			\frac{N_{\arr}(T)}{T} - \frac{\EE{A_1}}{ \EE{B_1}}
		}
		\Rightarrow \mathcal{N}(0, \sigma_2^2/2),
	\end{equation}
	as $T \to \infty$, where
	\begin{equation}
		\sigma_2^2 = 2\frac{\EE{A_1}^2}{\EE{B_{1}}^2}\p{\frac{\Var{A_1}}{\EE{A_1}^2} - \frac{2\Cov{A_1, B_{1}}}{\EE{A_1}\EE{B_{1}}}  + \frac{\Var{B_{1}}}{\EE{B_{1}}^2} }  .
	\end{equation}
Furthermore, since $M_{k_r}(T)/T \stackrel{p}{\to} 1/\EE{B_1}$,
	\begin{equation}
		\label{eqn:CLT_one_chain}
		\sqrt{T}\p{
			\frac{N_{\arr}(T)}{T} - \frac{\EE{A_1}}{ \EE{B_1}}
		}
			\Rightarrow \mathcal{N}(0, \EE{B_1}\sigma_2^2/2),
	\end{equation}
	as $T \to \infty$.

Note that $A_1$ has the same distribution as $A_{1,+}(0)$ and $A_{1,-}(0)$ defined in the proof of Proposition \ref{prop:first_esti_triangular} (Section \ref{proof:prop:first_esti_triangular}). The same holds for $B_1$, $B_{1,+}(0)$ and $B_{1,-}(0)$. Therefore, $\sigma_2^2 = \sigma_1^2$, where $\sigma_1^2$ is defined in \eqref{eqn:variance_formula}. Consequently,
as $T \to \infty$,
\begin{equation}
\sqrt{T}
\p{\frac{N_{\arr}(T, p)}{T} - \tilde{\mu}(p)}
\Rightarrow \mathcal{N}(0,  \tilde{\sigma}^2(p)/2).
\end{equation}

\subsubsection{Proof of Proposition \ref{prop:coro:num_arrival}}
\label{proof:prop:coro:num_arrival}
Note that the queue length is bounded above by $K$ and thus the absolute difference between the number of arrivals and the number of departures is always bounded above by $K$.
It then follows directly from Lemma \ref{lemma:num_arrival} and Proposition \ref{prop:Q_inverse} that
\begin{equation}
\begin{split}
\sigma_{\arr}^2(p)
&= (1-\pi_0(p))\mu + 2\mu^2 \sum_{i = 0}^{K-1}\frac{\pi_{i+1}(p)\pi_0(p)}{\mu}\p{\sum_{ j = 1}^i \frac{S_j(p)}{\pi_j(p)} - \sum_{ j =1}^K \frac{S_j^2(p)}{\pi_j(p)}}\\
& = (1-\pi_0(p))\mu + 2\mu \pi_0(p)\p{\sum_{i = 0}^{K-1}\pi_{i+1}(p)\sum_{ j = 1}^i \frac{S_j(p)}{\pi_j(p)} - \sum_{i = 0}^{K-1}\pi_{i+1}(p)\sum_{ j =1}^K \frac{S_j^2(p)}{\pi_j(p)}}\\
& = (1-\pi_0(p))\mu + 2\mu \pi_0(p)\p{\sum_{j = 1}^{K-1}\frac{S_j(p)}{\pi_j(p)} \sum_{i = j}^{K-1}\pi_{i+1}(p) - (1-\pi_0(p))\sum_{ j =1}^K \frac{S_j^2(p)}{\pi_j(p)}}\\
& = (1-\pi_0(p))\mu + 2\mu \pi_0(p)\p{\sum_{j = 1}^{K-1}\frac{S_j(p)S_{j+1}(p)}{\pi_j(p)}  - (1-\pi_0(p))\sum_{ j =1}^K \frac{S_j^2(p)}{\pi_j(p)}}.
\end{split}
\end{equation}

\subsubsection{Proof of Proposition \ref{prop:second_esti_triangular}}

The proof is very similar to that of Proposition \ref{prop:first_esti_triangular}. We only need to change the definition of $A_{i,+}$ and $A_{i,-}$. Here, we define $A_{i,+}(\zeta_T)$ ($A_{i,-}(\zeta_T)$) to be the amount of time in state $0$ between the $i$-th and the $(i+1)$-th time of hitting $k_r$ if the price is $p+\zeta_T$ ($p-\zeta_T$). For simplicity's sake, we omit the details.

\subsubsection{Proof of Proposition \ref{prop:second_esti_zero}}

The proof is very similar to that of Proposition \ref{prop:first_esti_zero}. We only need to change the definition of $A_i$. Here, we define $A_{i}$ to be the amount of time in state $0$ between the $i$-th and the $(i+1)$-th time of hitting $k_r$.

\subsubsection{Proof of Proposition \ref{prop:Q_inverse}}

We will omit the dependence on $p$ in this proof. We will make use of Lemma \ref{lemma:suff_conds_Q_inv}. Let $A$ be a matrix with entries
\begin{equation}
A_{k, i} = \frac{\pi_i}{\mu} \p{- \sum_{j = 1}^{\min(i,k)} \frac{1}{\pi_j} + \sum_{j = 1}^k \frac{S_j}{\pi_j} + \sum_{j = 1}^i \frac{S_j}{\pi_j} - \sum_{j = 1}^{K} \frac{S_j^2}{\pi_j}}.
\end{equation}
Write $A = (a_0, \dots, a_K)$.
We will then show that $QA = I - \onev\piv^\top$
and $\piv^\top A = \zerov$.  We start with the second condition. For any $i \in \cb{0, \dots, K}$, note that since $\sum_{k = 0}^K \pi_k = 1$,
\begin{equation}
\begin{split}
\piv^\top a_i
&= \frac{\pi_i}{\mu}\p{-\sum_{k = 1}^K \sum_{j = 1}^{\min(i,k)} \frac{\pi_k}{\pi_j} +\sum_{k=1}^K \sum_{j = 1}^k \frac{\pi_k S_j}{\pi_j} + \sum_{j = 1}^i \frac{S_j}{\pi_j} - \sum_{j = 1}^{K} \frac{S_j^2}{\pi_j}  }\\
& = \frac{\pi_i}{\mu}\p{-\sum_{j = 1}^{i}\sum_{k = j}^K  \frac{\pi_k}{\pi_j} +  \sum_{j = 1}^K \sum_{k=j}^K \frac{\pi_kS_j}{\pi_j} +   \sum_{j = 1}^i \frac{S_j}{\pi_j} - \sum_{j = 1}^{K} \frac{S_j^2}{\pi_j}  }\\
& = \frac{\pi_i}{\mu}\p{-\sum_{j = 1}^{i}  \frac{S_j}{\pi_j} +  \sum_{j = 1}^K  \frac{S_j^2}{\pi_j} +   \sum_{j = 1}^i \frac{S_j}{\pi_j} - \sum_{j = 1}^{K} \frac{S_j^2}{\pi_j} } = 0.
\end{split}
\end{equation}
Therefore, $\piv^\top A = \zerov$.

For the first condition, we need to verify that $QA = I - \onev\piv^\top$. Write $Q = (q_0, q_1, \dots, q_K)^\top$. In particular, we will show that $q_k^\top a_i = -\pi_i$ if $k \neq i$, and that $q_k^\top a_i = 1 - \pi_i$ if $k = i$. Note that $q_k$ has a very special structure: it has at most three nonzero entries and $q_k^\top \onev = 0$. Therefore,
\begin{equation}
\begin{split}
q_k^\top a_i &= Q_{k,k-1} A_{k-1, i} + Q_{k,k} A_{k, i} + Q_{k,k+1} A_{k+1, i} \\
&= -Q_{k,k-1} (A_{k, i} - A_{k-1, i}) + Q_{k,k+1} (A_{k+1, i} - A_{k, i}).
\end{split}
\end{equation}
For indices out of range, we treat the corresponding matrix entries as zero. Note that for $0 \leq k \leq K - 1$,
\begin{equation}
A_{k+1, i} - A_{k, i} = \frac{\pi_i}{\mu}\p{-\frac{\mathbbm{1}\cb{k+1 \leq i}}{\pi_{k+1}} + \frac{S_{k+1}}{\pi_{k+1}}}.
\end{equation}
Thus, for $1 \leq k \leq K - 1$,
\begin{equation}
\begin{split}
q_k^\top a_i &= -Q_{k,k-1} (A_{k, i} - A_{k-1, i}) + Q_{k,k+1} (A_{k+1, i} - A_{k, i})\\
& = -\mu \frac{\pi_i}{\mu}\p{-\frac{\mathbbm{1}\cb{k \leq i}}{\pi_{k}} + \frac{S_{k}}{\pi_{k}}} + \lambda_{k}\frac{\pi_i}{\mu}\p{-\frac{\mathbbm{1}\cb{k+1 \leq i}}{\pi_{k+1}} + \frac{S_{k+1}}{\pi_{k+1}}}\\
& = \frac{\pi_i}{\pi_k}\p{\mathbbm{1}\cb{k \leq i} - S_k -\mathbbm{1}\cb{k \leq i - 1} + S_{k+1}}\\
& = \frac{\pi_i}{\pi_k}\p{\mathbbm{1}\cb{k = i} - \pi_k} = \mathbbm{1}\cb{k = i} - \pi_i.
\end{split}
\end{equation}
Here we make use of the steady-state equation of $\pi_k \lambda_k = \pi_{k+1} \mu$. For the boundary cases, the results are similar. For $k = 0$,
\begin{equation}
\begin{split}
q_0^\top a_i &= Q_{0,1} (A_{1, i} - A_{0, i})= \lambda_0\frac{\pi_i}{\mu}\p{-\frac{\mathbbm{1}\cb{i \geq 1}}{\pi_{1}} + \frac{S_{1}}{\pi_{1}}} \\
&= \frac{\pi_i}{\pi_0}\p{-1+\mathbbm{1}\cb{i = 0} + (1-\pi_0)}
= \mathbbm{1}\cb{i = 0} - \pi_i.
\end{split}
\end{equation}
For $k = K$,
\begin{equation}
\begin{split}
q_K^\top a_i &= -Q_{K,K-1} (A_{K, i} - A_{K-1, i})
= -\mu \frac{\pi_i}{\mu}\p{-\frac{\mathbbm{1}\cb{K \leq i}}{\pi_{K}} + \frac{S_K}{\pi_K}}\\
&= \frac{\pi_i}{\pi_K}\p{\mathbbm{1}\cb{i = K} - \pi_K}
= \mathbbm{1}\cb{i = K} - \pi_i.
\end{split}
\end{equation}
Combining the above three results, we get for any $k \in \cb{0, \dots, K}$,
\begin{equation}
\begin{split}
q_k^\top a_i = \mathbbm{1}\cb{k = i} - \pi_i.
\end{split}
\end{equation}
Hence, $QA = I - \onev\piv^\top$.

Finally, by Lemma \ref{lemma:suff_conds_Q_inv}, we have that $Q^\# = A$.

\subsubsection{Proof of Propositions \ref{prop:CLT_varying_para} and \ref{prop:CLT_fixed_para_tilde}}
Again, the proof closely resembles those of Propositions \ref{prop:first_esti_triangular} and \ref{prop:first_esti_zero}. Rather than focusing only on the total number of arrivals and the total amount of time when the price is either $p+\zeta_T$ or $p - \zeta_T$, we consider multiple quantities at the same time.

For Proposition \ref{prop:CLT_varying_para}, similar to the proof of Proposition \ref{prop:first_esti_triangular}, we define $A_{i,k,+}$ as the count of jumps from $k$ to $k+1$ in the $i$-th positive interval. Similarly, we define $A_{i,k,-}$. We let $B_{i,k,+}$ represent the amount of time the queue length is $k$ in the $i$-th positive interval. We define $B_{i,k,-}$ similarly. The subsequent analysis mirrors that of Proposition \ref{prop:first_esti_triangular}.

For Proposition \ref{prop:CLT_fixed_para_tilde}, similar to the proof of Proposition \ref{prop:first_esti_zero}, we define $A_{i,k}$ as the number of jumps from $k$ to $k+1$ between the $i$-th and the $i+1$-th time of hitting $k_r$. We define $B_{i,k}$ as the amount of time in state $k$ between the $i$-th and the $i+1$-th time of hitting $k_r$. The subsequent analysis mirrors that of Proposition \ref{prop:first_esti_zero}.

For the sake of brevity, we omit further details.

\subsubsection{Proof of Proposition \ref{prop:CLT_fixed_para}}
We are essentially dealing with a continuous-time Markov chain with transition rate matrix $Q$ defined in \eqref{eqn:Q_matrix}. Take $\alpha = 2\max_{i,j} \abs{Q_{i,j}}$. We conduct a uniformization step similar to the one in the proof of Lemma \ref{lemma:num_eligible_jumps}. The continuous-time Markov chain can be described by a discrete-time Markov chain with transition matrix $P = I + Q/\alpha$
where jumps occur according to a Poisson process with intensity $\alpha$. Let $M(T)$ be the total number of jumps of the discrete-time Markov chain. Note that among the jumps, there are some self jumps that go from state $k$ to state $k$. From now on, when we refer to jumps, we are referring to jumps of the discrete-time Markov chain, and thus we include such self jumps.

One useful property of the uniformization step is that the times between jumps are independent of the behavior of the discrete-time Markov chain. However, because $M(T)$ is defined as the total number of jumps by time $T$, $M(T)$ itself is correlated with the times between jumps, which complicates our analysis. To deal with this, we instead consider a process indexed by the total number of jumps. Define $\tilde{N}_k(m)$ to be the number of jumps from state $k$ to state $k+1$ within the first $m$ jumps and $\tilde{T}_k(m)$ to be the amount of time in state $k$ within the first $m$ jumps. Let $\tilde{T}(m)$ be the total amount of time spent until the $m$-th jump.

We next relate the fixed-jump process to the process observed at time $T$. For each $k$, define the compensated counting process
\[
\mathcal{M}_k(t)=N_k(t)-\lambda_k T_k(t).
\]
This is a square-integrable martingale with predictable quadratic variation
\[
\langle\mathcal{M}_k\rangle_t=\lambda_kT_k(t)\leq\lambda_k t.
\]
Without loss of generality, suppose that $\alpha T$ is an integer. Since $\tilde T(\alpha T)$ is the sum of $\alpha T$ independent $\operatorname{Exp}(\alpha)$ random variables,
\[
\tilde T(\alpha T)-T=O_p(\sqrt T).
\]
Moreover, Doob's maximal inequality, applied on the interval
$[T-C\sqrt T,T+C\sqrt T]$, implies that for every fixed $C>0$,
\[
\sup_{|s-T|\leq C\sqrt T}
|\mathcal{M}_k(s)-\mathcal{M}_k(T)|=O_p(T^{1/4})=o_p(\sqrt T).
\]
Letting $C\to\infty$ after $T\to\infty$ therefore gives
\begin{equation}
\label{eqn:fixed_jump_time_change}
\mathcal{M}_k(T)
=\mathcal{M}_k(\tilde T(\alpha T))+o_p(\sqrt T)
=\tilde N_k(\alpha T)-\lambda_k\tilde T_k(\alpha T)+o_p(\sqrt T),
\end{equation}
where the last equality follows directly from the uniformization construction (any endpoint convention changes the expression by at most $O_p(1)$).

By the ergodic theorem, $T_k(T)/T\stackrel{p}{\to}\pi_k$. Combining this fact with \eqref{eqn:fixed_jump_time_change}, we obtain
\begin{equation}
\label{eqn:ratio_to_diff}
\frac{N_k(T)}{T_k(T)}
=\lambda_k+\frac{1}{\pi_kT}
\p{\tilde N_k(\alpha T)-\lambda_k\tilde T_k(\alpha T)}(1+o_p(1))
+o_p(1/\sqrt T).
\end{equation}

We then turn our attention to the behavior of the Markov chain over a fixed total of $m = \alpha T$ jumps. First, we can establish a multivariate central limit theorem for $\tilde{N}_0(m), \dots, \tilde{N}_{K-1}(m), \tilde{T}_0(m), \dots, \tilde{T}_{K-1}(m)$ because the Markov chain is a regenerative process. Second, the uniform integrability result in Lemma \ref{lemm:uniform_intergral_disc} ensures that the limiting means, variances, and covariances of $\tilde{N}_k(m) - \lambda_k \tilde{T}_k(m)$ correspond to those in the central limit theorem.

It thus suffices to study the limiting mean, variance, and covariance of $\tilde{N}_k(m) - \lambda_k \tilde{T}_k(m)$. Here, we study more general forms of ``number of jumps'' and ``amount of time.''
\begin{defi}
\label{defi:N_J}
Consider the continuous-time Markov chain with transition rate matrix $Q$ defined in \eqref{eqn:Q_matrix}. With a step of uniformization, the continuous-time Markov chain can be described by a discrete-time Markov chain with transition matrix $P = I + Q/\alpha$ where jumps occur according to a Poisson process with intensity $\alpha$.

Let $\mathcal{J}$ be a set of eligible jumps, i.e., $\mathcal{J}$ is a set of $(i,j)$ pairs such that $i,j \in \cb{0, \dots, K}$. Here the pair $(i,j)$ stands for a jump from state $i$ to state $j$. Let $\tilde{N}_{\mathcal{J}}(m)$ be the number of eligible jumps within the first $m$ jumps. Let $\tilde{T}_{\mathcal{J}}(m)$ be the cumulative amount of time between each eligible jump and its immediately preceding jump within the first $m$ jumps. Let $\tilde{Q}_{\mathcal{J}}$ be a matrix such that $(\tilde{Q}_{\mathcal{J}})_{i,j} = Q_{i,j}$ if $(i,j) \in \mathcal{J}$ and $(\tilde{Q}_{\mathcal{J}})_{i,j} = 0$ otherwise. Let $\tilde{P}_{\mathcal{J}}$ be a matrix such that $(\tilde{P}_{\mathcal{J}})_{i,j} = P_{i,j}$ if $(i,j) \in \mathcal{J}$ and $(\tilde{P}_{\mathcal{J}})_{i,j} = 0$ otherwise.
\end{defi}

Lemma \ref{lemm:covariance_N_J} gives the means, variances, and covariances of $\tilde{N}_{\mathcal{J}}(m)$ and $\tilde{T}_{\mathcal{J}}(m)$.
Using our notation for eligible jumps, $\tilde{N}_k(m) = \tilde{N}_{\cb{(k,k+1)}}(m)$ and $\tilde{T}_k(m) = \tilde{T}_{k\operatorname{out}}(m)$, where $k\operatorname{out} = \cb{(k,k-1), (k,k), (k,k+1)} \cap \p{\cb{0, \dots, K}\times \cb{0, \dots, K}}$.
In particular, we note that for the mean,
\begin{equation}
\EE{\tilde{N}_k(m) - \lambda_k \tilde{T}_k(m)} = \oo(1).
\end{equation}
For the variance and covariance, Lemma \ref{lemm:covariance_N_J} implies that for disjoint sets ${\mJ_i}$ and ${\tilde{\mJ}_i}$, if there exist real numbers $w_i$ and $\tilde{w}_i$ such that $(\sum_i w_i\tilde{P}_{\mJ_i})\onev = \zerov$ and $(\sum_i \tilde{w}_i\tilde{P}_{\tilde{\mJ}_i})\onev = \zerov$, then
\begin{equation}
\label{eqn:cases_variance_simple}
\Var{\sum_i w_i \tilde{N}_{\mJ_i}(m)} = m \sum_i w_i^2 \piv^{\top} \tilde{P}_{\mathcal{J}_i} \onev + \oo(1),
\end{equation}
\begin{equation}
\label{eqn:cases_covariance_zero}
\Cov{\sum_i w_i \tilde{N}_{\mJ_i}(m), \sum_i \tilde{w}_i \tilde{N}_{\tilde{\mJ}_i}(m)} = \oo(1).
\end{equation}

With this result, we can study $\Cov{\tilde{N}_k(m) - \lambda_k \tilde{T}_k(m), \tilde{N}_l(m) - \lambda_l \tilde{T}_l(m)}$. Consider the case where $k \neq l$. By Lemma \ref{lemm:covariance_N_J},
\begin{equation}
\begin{split}
&\Cov{\tilde{N}_k(m) - \lambda_k \tilde{T}_k(m), \tilde{N}_l(m) - \lambda_l \tilde{T}_l(m)}\\
&= \frac{1}{\alpha^2} \Cov{ \alpha \tilde{N}_{\cb{(k,k+1)}}(m) - \lambda_k \tilde{N}_{k \operatorname{out}}(m),  \alpha \tilde{N}_{\cb{(l,l+1)}}(m) - \lambda_l \tilde{N}_{l \operatorname{out}}(m)}.
\end{split}
\end{equation}
Clearly, $\tilde{P}_{k \operatorname{out}} \onev = \ev_k$ and   $\tilde{P}_{\cb{(k,k+1)}} \onev = (\lambda_k/\alpha)\ev_k$. Therefore, $\p{\alpha \tilde{P}_{\cb{(k,k+1)}} - \lambda_k \tilde{P}_{k \operatorname{out}}} \onev = \zerov$. By \eqref{eqn:cases_covariance_zero}, we have that
\begin{equation}
\Cov{\tilde{N}_k(m) - \lambda_k \tilde{T}_k(m), \tilde{N}_l(m) - \lambda_l \tilde{T}_l(m)} = \oo(1),
\end{equation}
and thus
\begin{equation}
\lim_{m \to \infty}\frac{1}{m} \Cov{\tilde{N}_k(m) - \lambda_k \tilde{T}_k(m), \tilde{N}_l(m) - \lambda_l \tilde{T}_l(m)} = 0.
\end{equation}

For the variance, again by Lemma \ref{lemm:covariance_N_J},
\begin{equation}
\begin{split}
& \Var{\tilde{N}_k(m) - \lambda_k \tilde{T}_k(m)}\\
&\qquad = \frac{1}{\alpha^2} \p{\Var{\alpha \tilde{N}_k(m) - \lambda_k \tilde{N}_{k \operatorname{out}}(m)} + \lambda_k^2\EE{\tilde{N}_{k \operatorname{out}}(m)}}\\
&\qquad = \frac{1}{\alpha^2} \p{\Var{\alpha \tilde{N}_{\cb{(k,k+1)}}(m) - \lambda_k \tilde{N}_{k \operatorname{out}}(m)} + \lambda_k^2\EE{\tilde{N}_{k \operatorname{out}}(m)}}.
\end{split}
\end{equation}
For $1 \leq k \leq K-1$, by \eqref{eqn:cases_variance_simple},
\begin{equation}
\begin{split}
&\Var{\alpha \tilde{N}_{\cb{(k,k+1)}}(m) - \lambda_k \tilde{N}_{k \operatorname{out}}(m)} + \lambda_k^2\EE{\tilde{N}_{k \operatorname{out}}(m)}\\
& \qquad= m\Big(\lambda_k^2 \piv^{\top} \tilde{P}_{\cb{(k,k-1)}} \onev
+ \lambda_k^2 \piv^{\top} \tilde{P}_{\cb{(k,k)}} \onev\\
&\qquad\qquad + (\alpha-\lambda_k)^2 \piv^{\top} \tilde{P}_{\cb{(k,k+1)}} \onev
+ \lambda_k^2 \piv^{\top} \tilde{P}_{k \operatorname{out}} \onev\Big) + \oo(1)\\
&\qquad =   m \pi_k \p{\lambda_k^2 \mu/\alpha + \lambda_k^2(1-(\mu+\lambda_k)/\alpha) + (\alpha-\lambda_k)^2\lambda_k/\alpha + \lambda_k^2}+ \oo(1)\\
& \qquad =  m\pi_k \lambda_k \alpha + \oo(1).
\end{split}
\end{equation}
For $k = 0$, by \eqref{eqn:cases_variance_simple},
\begin{equation}
\begin{split}
&\Var{\alpha \tilde{N}_{\cb{(k,k+1)}}(m) - \lambda_k \tilde{N}_{k \operatorname{out}}(m)} + \lambda_k^2\EE{\tilde{N}_{k \operatorname{out}}(m)}\\
& \qquad=   m\p{\lambda_k^2 \piv^{\top} \tilde{P}_{\cb{(k,k)}} \onev + (\alpha-\lambda_k)^2 \piv^{\top} \tilde{P}_{\cb{(k,k+1)}} \onev + \lambda_k^2 \piv^{\top} \tilde{P}_{k \operatorname{out}} \onev} + \oo(1)\\
&\qquad =  m\pi_k\p{\lambda_k^2(1-\lambda_k/\alpha) +  (\alpha-\lambda_k)^2\lambda_k/\alpha + \lambda_k^2} + \oo(1)\\
& \qquad = m\pi_k \lambda_k \alpha + \oo(1).
\end{split}
\end{equation}
Therefore,
\begin{equation}
\Var{\tilde{N}_k(m) - \lambda_k \tilde{T}_k(m)} = \frac{m}{\alpha} \pi_k \lambda_k + \oo(1),
\end{equation}
and hence
\begin{equation}
    \Var{\tilde{N}_k(\alpha T) - \lambda_k \tilde{T}_k(\alpha T)} = T \pi_k\lambda_k + \oo(1).
\end{equation}

Combining with our discussion before Definition \ref{defi:N_J}, we have that
\begin{equation}
\frac{1}{\sqrt{T}} (\operatorname{diff}_0, \operatorname{diff}_1, \dots, \operatorname{diff}_{K-1}) \Rightarrow  \mathcal{N}\p{\vec{0}, \Sigma_{\operatorname{diff}}},
\end{equation}
where
$\operatorname{diff}_k =\tilde{N}_k(\alpha T) - \lambda_k \tilde{T}_k(\alpha T)$,
$\Sigma_{{\operatorname{diff}},k,k} = \pi_k\lambda_k$ and $\Sigma_{\operatorname{diff}, k, l} = 0$ for $k \neq l$.

Finally, together with \eqref{eqn:ratio_to_diff}, we have that
\begin{equation}
\sqrt{T} (\operatorname{err}_0, \operatorname{err}_1, \dots, \operatorname{err}_{K-1}) \Rightarrow  \mathcal{N}\p{\vec{0}, \Sigma_{\ratio}(p)},
\end{equation}
where
$\operatorname{err}_k = \frac{N_k(T,p)}{T_k(T,p)} -  \lambda_{k}(p)$,
$\Sigma_{{\ratio},k,k}(p) = \lambda_{k}(p)/\pi_k(p)$ and $\Sigma_{\ratio, k, l} = 0$ for $k \neq l$.

\subsubsection{Proof of Proposition \ref{prop:CLT_ratio}}

Let $N_{k}(T,\zeta_T) = N_{k,+}(T,\zeta_T) + N_{k,-}(T,\zeta_T) $. For notational simplicity, we sometimes omit $(T,\zeta_T)$ and the subscript $\LE$ in the proof.

We start by noting that, by Lemma \ref{lemma:N_k_T_converge}, $N_k/T \stackrel{p}{\to} \pi_k(p)\lambda_k(p)$.
Conditional on $N_{k}$,
\begin{equation}
N_{k, +} \sim \operatorname{Binomial}\p{N_k, \frac{\lambda_{k}(p + \zeta_T) }{ \lambda_{k}(p + \zeta_T) + \lambda_{k}(p - \zeta_T)}}.
\end{equation}
By Lemma \ref{lemm:binom_clt} and Anscombe's theorem \citep{anscombe1952large},
\begin{equation}
\sqrt{N_k} \frac{\lambda_{k}(p + \zeta_T) + \lambda_{k}(p - \zeta_T)}{\sqrt{\lambda_{k}(p + \zeta_T) \lambda_{k}(p - \zeta_T)}}\p{ \frac{N_{k, +} }{N_k} - \frac{\lambda_{k}(p + \zeta_T) }{\lambda_{k}(p + \zeta_T) + \lambda_{k}(p - \zeta_T)}} \Rightarrow  \mathcal{N}\p{0, 1},
\end{equation}
as $T \to \infty$. Thus,
\begin{equation}
\sqrt{N_k} \frac{\lambda_{k}(p + \zeta_T) + \lambda_{k}(p - \zeta_T)}{2 \sqrt{\lambda_{k}(p + \zeta_T) \lambda_{k}(p - \zeta_T)}}\p{ \frac{N_{k, +} - N_{k, -}}{N_k} - \frac{\lambda_{k}(p + \zeta_T) - \lambda_{k}(p - \zeta_T)}{\lambda_{k}(p + \zeta_T) + \lambda_{k}(p - \zeta_T)}} \Rightarrow  \mathcal{N}\p{0, 1},
\end{equation}
as $T \to \infty$. Since $\lambda_{k}(p + \zeta_T) \to \lambda_{k}(p)$, $\lambda_{k}(p - \zeta_T) \to \lambda_{k}(p)$, and $N_k/T \stackrel{p}{\to} \pi_k(p)\lambda_{k}(p)$, by Slutsky's theorem,
 \begin{equation}
\sqrt{T} \p{ \frac{N_{k, +} - N_{k, -}}{N_k} - \frac{\lambda_{k}(p + \zeta_T) - \lambda_{k}(p - \zeta_T)}{\lambda_{k}(p + \zeta_T) + \lambda_{k}(p - \zeta_T)}} \Rightarrow  \mathcal{N}\p{0,
\frac{1}{\pi_k(p)\lambda_{k}(p)}}.
\end{equation}
Since the second derivative of $\lambda_{k}$ is bounded in absolute value by a constant, we have that $\lambda_{k}(p + \zeta_T) - \lambda_{k}(p - \zeta_T) = 2 \zeta_T \lambda_{k}'(p) + \oo\p{\zeta_T^2}$, and $\lambda_{k}(p + \zeta_T) + \lambda_{k}(p - \zeta_T) = 2\lambda_{k}(p) + \mathcal{O}\p{\zeta_T^2}$. Hence,
\begin{equation}
\frac{\lambda_{k}(p + \zeta_T) - \lambda_{k}(p - \zeta_T)}{\lambda_{k}(p + \zeta_T) + \lambda_{k}(p - \zeta_T)}
= \zeta_T \frac{\lambda_{k}'(p)}{\lambda_{k}(p)} + \oo\p{\zeta_T^2}.
\end{equation}
Since $\sqrt{T} \zeta_T^2 \to 0$, we have that
 \begin{equation}
\sqrt{T} \p{ \frac{N_{k, +} - N_{k, -}}{N_k} - \zeta_T \frac{\lambda_{k}'(p)}{\lambda_{k}(p)}} \Rightarrow  \mathcal{N}\p{0,
\frac{1}{\pi_k(p)\lambda_{k}(p)}},
\end{equation}
i.e.,
 \begin{equation}
\sqrt{T \zeta_T^2} \p{\hat{\Delta}_k -  \frac{\lambda_{k}'(p)}{\lambda_{k}(p)}} \Rightarrow  \mathcal{N}\p{0,
\frac{1}{\pi_k(p)\lambda_{k}(p)}}.
\end{equation}

For the joint distribution, the result is a consequence of the following fact: Conditioning on $(N_0, N_1, \dots, N_{K-1})$, the random variables $N_{0,+}, \dots, N_{K-1,+}$ are independent.

\section{Additional Simulation Results and Details}
\label{appendix:simulation_details_section}

\subsection{Variance estimators}
\label{appendix:variance_estimator_simulation}
\begin{figure}
	\centering

		\begin{subfigure}{0.87\textwidth}
			\centering
			\includegraphics[width=\textwidth]{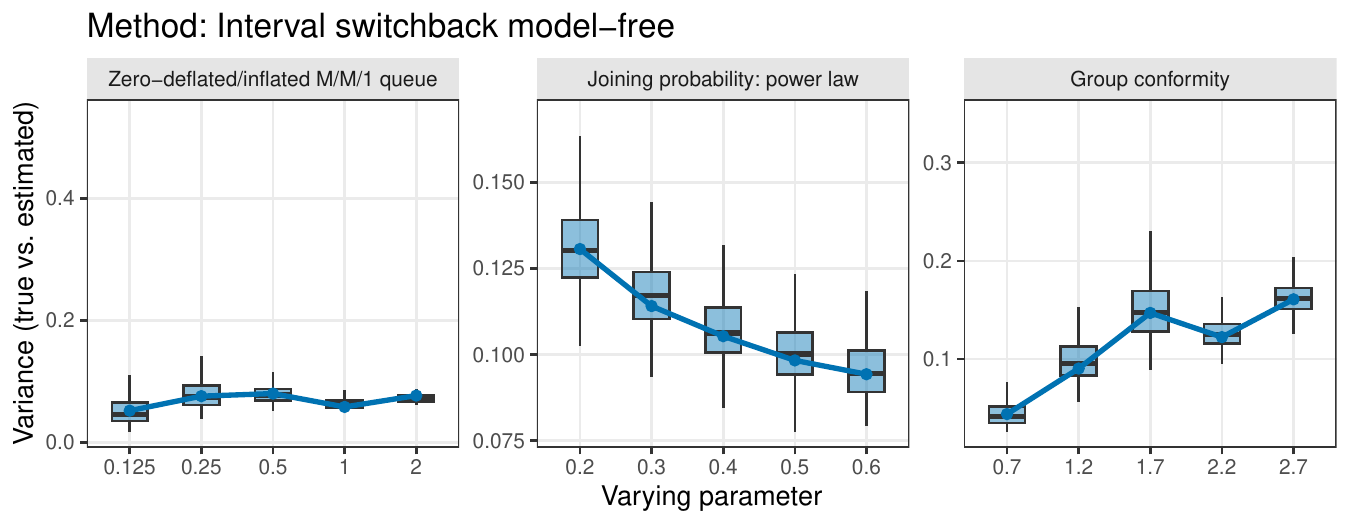}
		\end{subfigure}

		\begin{subfigure}{0.87\textwidth}
			\centering
			\includegraphics[width=\textwidth]{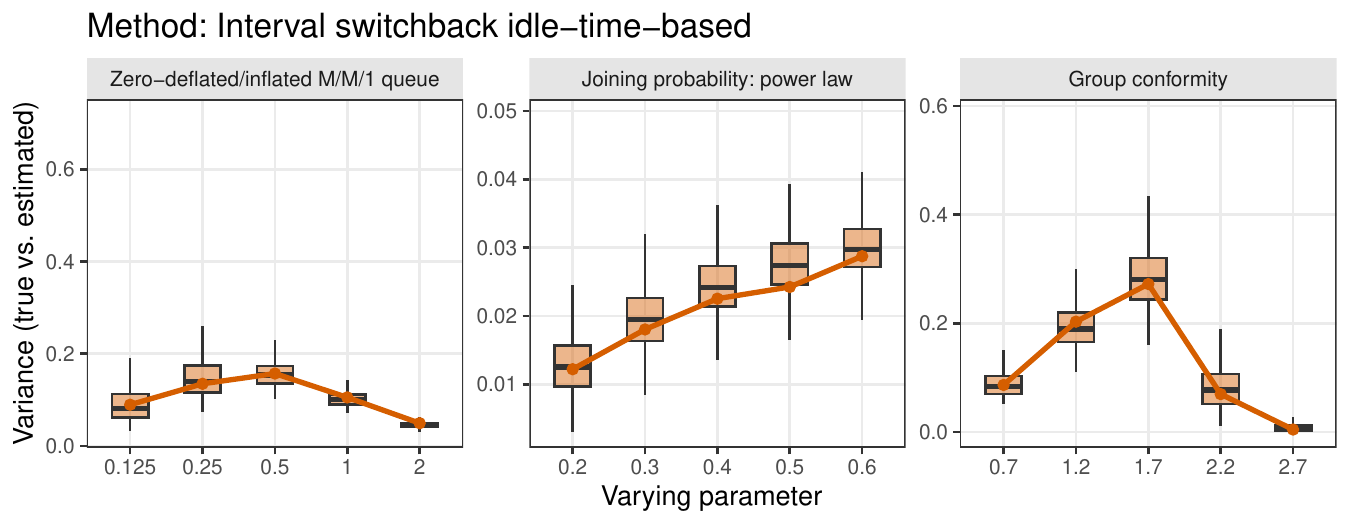}
		\end{subfigure}

		\begin{subfigure}{0.87\textwidth}
			\centering
			\includegraphics[width=\textwidth]{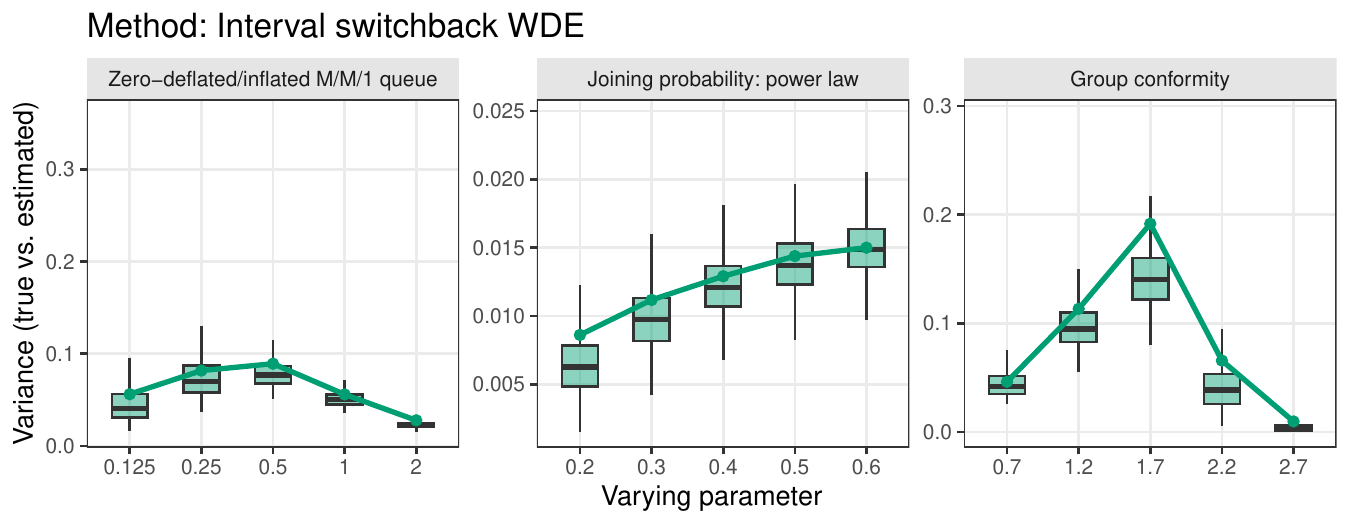}
		\end{subfigure}

		\begin{subfigure}{0.87\textwidth}
			\centering
			\includegraphics[width=\textwidth]{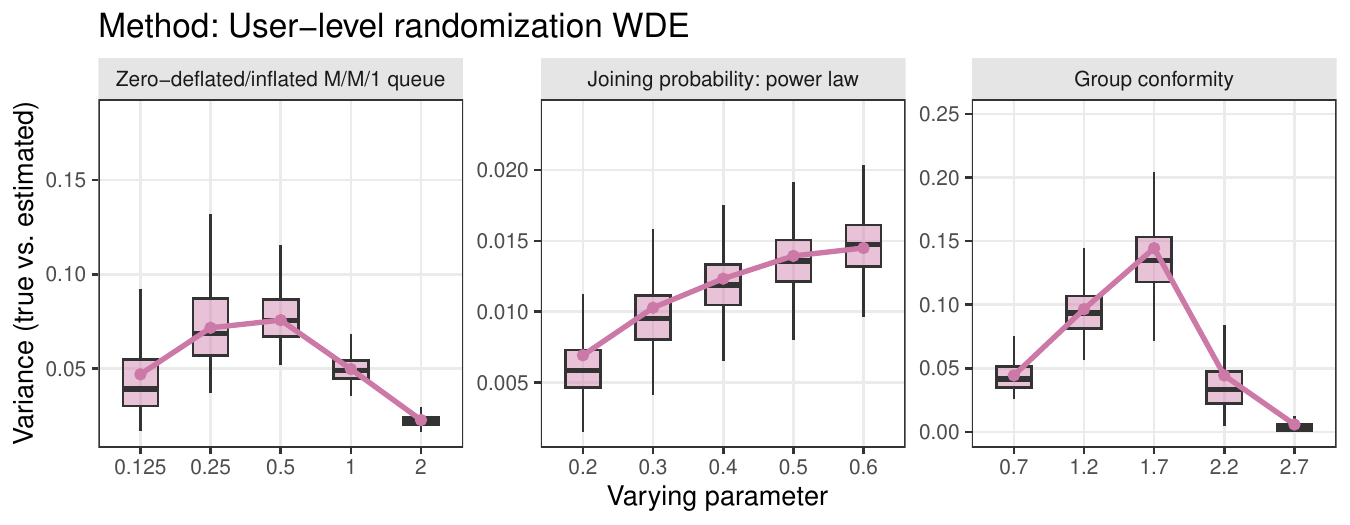}
		\end{subfigure}

		\caption{Comparison between the true variance (solid lines) and the estimated variance (boxplots) across four estimator--experimental-design configurations.}
	\label{fig:variance_estimator_all}
\end{figure}

In Figure~\ref{fig:variance_estimator_all}, we examine the accuracy of the proposed variance estimators discussed in Section~\ref{subsection:variance_estimator}. The variance estimators are generally quite accurate. For the model-free estimator, the idle-time-based estimator, and the WDE estimator in user-level randomization, the estimated variances closely track the true variances. For the WDE estimator in interval switchback experiments, the variance estimator can occasionally underestimate the true variance.

The simulation settings follow those in Section~\ref{section:illus_asymp_var}. We take $\zeta_T = 0.05$, hold the target price at $p = 1$, and assume that $\lambda_k(p) = \lambda_k(1)(2 - p)$.

\subsection{Comparing the WDE estimator with the DQ estimator}

Another estimator that addresses interference arising from Markovian systems using data from user-level randomization is the DQ estimator from \citet{farias2022Markovian}. To illustrate the performance of the WDE estimator in user-level randomization, we present a simulation-based comparison between the proposed WDE estimator and the DQ estimator.

As shown in Figure~\ref{fig:comparison_DQ}, the proposed WDE estimator achieves lower RMSE than both the DQ estimator and the difference-in-means estimator in most settings. We believe the efficiency gain arises from directly using the structural properties of the queueing system.

\begin{figure}
	\centering
	\includegraphics[width = \textwidth]{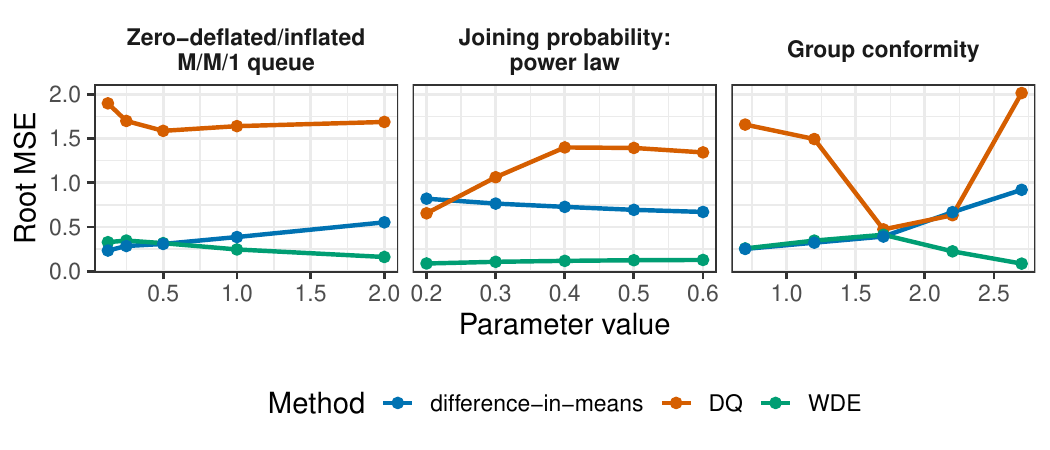}
		\caption{Root mean squared error (RMSE) comparison of the estimators under user-level randomization. We compare the proposed WDE estimator, the DQ estimator, and the difference-in-means estimator.}
		\label{fig:comparison_DQ}
\end{figure}

The simulation settings follow those in Section~\ref{section:illus_asymp_var}. We take $\zeta_T = 0.05$, hold the target price at $p = 1$, and assume that $\lambda_k(p) = \lambda_k(1)(2 - p)$. We also note that our analysis is conducted in a continuous-time Markov chain setting, while the DQ estimator was developed for discrete-time systems. To enable comparison, we discretize the trajectories when applying the DQ estimator, using a discretization interval of $0.1$.

The DQ estimator was implemented with LSTD, with a single regularization hyperparameter $\alpha$ as described in Section G.1.2 (Implementation details, point 3) of \citet{farias2022Markovian}. We estimate the state--action value function $Q$ and add an $L_2$ regularization term. In short, we solve for a fixed point $\widehat{\mathbf q}$ of the regularized least-squares problem
\[
\widehat{\mathbf q} = \arg \min_{\mathbf q'} \left\| \mathbf q' - \widehat{\mathbf r} - \widehat{\mathbf P}\widehat{\mathbf q} + \widehat{\lambda}\onev \right\|_2^2 + \alpha \left\| \mathbf q' \right\|_2^2 ,
\]
where $\widehat{\mathbf q},\widehat{\mathbf r},\onev \in \mathbb{R}^{2(K+1)}$, $\widehat{\mathbf P} \in \mathbb{R}^{2(K+1)\times 2(K+1)}$, and $\widehat{\lambda}\in\mathbb R$ is the estimated long-run average reward. We construct the plug-in estimates $\widehat{\mathbf r}$, $\widehat{\mathbf P}$, and $\widehat{\lambda}$ from the observed state--action transitions and rewards. We set the hyperparameter to $\alpha = 0.1$.

\subsection{Nonstationary Simulator in Section \ref{section:nonstationarity}}
\label{appendix:nonstationary_simulator}
The departure rate $\mu = 2$.
In week $w$, on day $d$, at time $t$, the state-dependent arrival rate is given by
\begin{equation}
\lambda_{w,d,t,k}(p) = \frac{4(2-p)}{1+k} a_w b_{d,t},
\end{equation}
where $b_{d,t}$ is the arrival rate depicted in Figure~\ref{fig:ED_data} and $a_1 = 0.9, a_2 = 1, a_3 = 1.1, a_4 = 1.2$.

\subsection{Figure \ref{fig:MSE_comp_third}}
\label{appendix:simulation_details_third}
We study the settings in Section \ref{section:illus_asymp_var}. In Figure \ref{fig:MSE_comp_third}, we plot the mean squared error (scaled by $T\zeta_T^2$) of the estimators. We take $\zeta_T = 0.05$.
We hold the target price at $p = 1$. In the first row of Figure \ref{fig:MSE_comp_third}, we assume that $\lambda_k(p) = \lambda_k(1) (2-p)$; in the second row of Figure \ref{fig:MSE_comp_third}, we assume that $\lambda_k(p) = \lambda_k(1)((2-p) + (1-p)^2)$.

\end{document}